\documentclass[11pt,a4paper]{article}
\usepackage[utf8]{inputenc}
\usepackage[margin=1in]{geometry}
\usepackage{color}
\usepackage{graphicx}
\usepackage{array}
\usepackage{verbatim}
\usepackage{caption}
\usepackage{subcaption}
\usepackage{amsmath,amsthm,amsfonts,amssymb,latexsym}
\usepackage{bbm}
\usepackage{setspace}
\usepackage{xparse}
\usepackage{epstopdf}
\usepackage{pgf}
\usepackage[colorlinks=true,citecolor=blue]{hyperref}
\usepackage[nameinlink,capitalise]{cleveref}

\usepackage{tikz}
\usepackage{tikz-cd}
\usepackage{pgfplotstable}
\pgfplotsset{compat=1.14}
\usetikzlibrary{patterns}
\usetikzlibrary{calc}
\usetikzlibrary{angles}
\usetikzlibrary{quotes}
\usetikzlibrary{external}

\DeclareDocumentCommand\abs{s m} {\IfBooleanTF{#1}{\left|#2\right|}{\left|#2\right|}}
\DeclareDocumentCommand\cont{o m o} {C\IfNoValueF{#1}{^{#1}}(#2\IfNoValueF{#3}{;#3})}
\DeclareDocumentCommand\contc{o m o} {C_c\IfNoValueF{#1}{^{#1}}(#2\IfNoValueF{#3}{;#3})}
\DeclareDocumentCommand\sobolev{m m o} {H^{#1}(#2 \IfNoValueF{#3}{,#3})}
\DeclareDocumentCommand\lp{m m o} {L^{#1}\left(#2 \IfNoValueF{#3}{,#3}\right)}
\DeclareDocumentCommand\norm{s m o} {\IfBooleanTF{#1}{\|#2\|}{\left\|#2\right\|}\IfNoValueF{#3}{_{#3}}}
\DeclareDocumentCommand\seminorm{m o o} {\left|#1\right|\IfNoValueF{#2}{_{#2 \IfNoValueF{#3}{,#3}}}}
\DeclareDocumentCommand\ip{s m m o} {\IfBooleanTF{#1}{\langle #2,#3 \rangle}{\left\langle #2,#3 \right\rangle}\IfNoValueF{#4}{_{#4}}}
\DeclareDocumentCommand\dup{m m o} {\left\langle{#1,#2}\right\rangle\IfNoValueF{#3}{_{#3', #3}}}
\DeclareDocumentCommand\gaussian{O{0} O{I}} {g_{#1, #2}}
\DeclareDocumentCommand\littleo{s o m} {o\IfNoValueF{#2}{_{#2}}\IfBooleanTF{#1}{(#3)}{\left(#3\right)}}
\DeclareDocumentCommand\bigo{s o m} {\mathcal O\IfNoValueF{#2}{_{#2}}\IfBooleanTF{#1}{(#3)}{\left(#3\right)}}

\DeclareMathOperator{\sign}{sign}

\DeclareMathOperator{\e}{e}

\newcommand{\revision}[1]{\textcolor{blue}{#1}}
\renewcommand{\revision}[1]{#1}

\newcommand{\commut}[2]{[#1, #2]}

\newcommand{\dummy}{\,\cdot\,}
\newcommand{\expect}[0]{\mathbf{E}}
\newcommand{\iip}[2]{\left(\!\left(#1, #2\right)\!\right)}
\newcommand{\nat}{\mathbf N}

\newcommand{\real}{\mathbf R}
\newcommand{\integer}{\mathbf Z}
\newcommand{\torus}{\mathbf T}

\newcommand{\hess}{\nabla^2}
\newcommand{\vect}[1]{\boldsymbol{\mathbf #1}}
\newcommand{\mat}[1]{\vect #1}

\renewcommand{\d}{\mathrm d}
\renewcommand{\t}{\mathsf T}

\makeatletter
\DeclareDocumentCommand \derivative{s m o m}{%
    \def\@der{\IfBooleanTF{#1}{\mathrm{d}}{\partial}}
    \def\@default{%
        \mathchoice{%
                \frac{%
                    \@der\ifnum\pdfstrcmp{#2}{1}=0\else^{#2}\fi {\IfNoValueTF{#3}{}{#3}}
                }{%
                    \@for\@token:={#4}\do{\@der \@token}
                }
            } {%
                \@for\@token:={#4}\do{\@der_\@token} \IfNoValueTF{#3}{}{#3}
            } {} {}
    }
    \IfBooleanTF{#1}{\IfNoValueTF{#3}{\@default}{%
                #3%
                \ifnum\pdfstrcmp{#2}{1}=0'\else%
                \ifnum\pdfstrcmp{#2}{2}=0''\else%
                \ifnum\pdfstrcmp{#2}{3}=0^{(3)}\else%
                \ifnum\pdfstrcmp{#2}{4}=0^{(4)}\else%
                \ifnum\pdfstrcmp{#2}{5}=0^{(5)}\else%
                ^{(#2)}\fi\fi\fi\fi\fi
            }
        }{\@default}
}
\makeatother

\definecolor{darkred}{rgb}{.5,0,0}
\definecolor{darkgreen}{rgb}{0,.5,0}
\definecolor{darkblue}{rgb}{0,0,.5}

\theoremstyle{plain}
\newtheorem{assumption}{Assumption}[section]
\newtheorem{lemma}{Lemma}[section]
\newtheorem{corollary}{Corollary}[section]
\newtheorem{theorem}{Theorem}[section]
\newtheorem{proposition}{Proposition}[section]
\newtheorem{result}{Result}[section]
\newtheorem{remark}{Remark}[section]
\numberwithin{equation}{section}

\newcounter{urbainCounter}

\crefname{equation}{}{}
\crefname{paragraph}{\S\!}{\S}
\newcommand{\email}[1]{\href{#1}{#1}}

\newcommand{\orcid}[1]{\href{https://orcid.org/#1}{\includegraphics[width=.4cm]{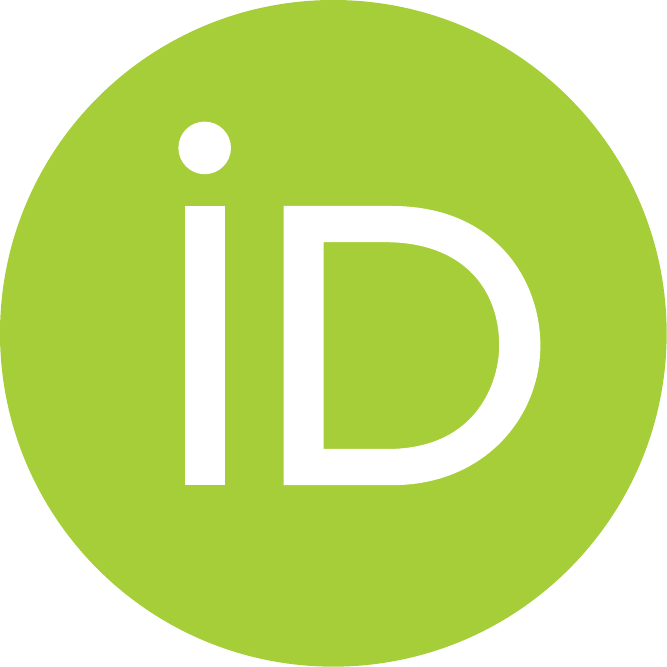}}}

\usepackage{enumerate}
\newcommand{\eps}{\varepsilon}
\newcommand{\dps}{\displaystyle}
\newcommand{\cX}{\mathcal{X}}
\newcommand{\ri}{\mathrm{i}}
\renewcommand{\leq}{\leqslant}
\renewcommand{\geq}{\geqslant}

\usepackage{mathrsfs}

\date{\today}
\title{Scaling limits for the generalized Langevin equation}
\author{%
  G.A. Pavliotis\thanks{Department of Mathematics, Imperial College London (\email{g.pavliotis@imperial.ac.uk})}%
  \hspace{2mm}\orcid{0000-0002-3468-9227}%
  \and G. Stoltz\thanks{CERMICS, \'Ecole des Ponts, France \& MATHERIALS, Inria Paris (\email{gabriel.stoltz@enpc.fr})}
  \hspace{2mm}\orcid{0000-0002-2797-5938}%
  \and U. Vaes\thanks{Department of Mathematics, Imperial College London (until October 2020) and MATHERIALS, Inria Paris (since November 2020) (\email{urbain.vaes@inria.fr})}%
  \hspace{2mm}\orcid{0000-0002-7629-7184}%
}

\begin{document}
\maketitle
\begin{abstract}
  In this paper, we study the diffusive limit of solutions to the \emph{generalized Langevin equation} (GLE) in a periodic potential.
  Under the assumption of quasi-Markovianity,
  we obtain sharp longtime equilibration estimates for the GLE using techniques from the theory of hypocoercivity.
  We then show asymptotic results for the effective diffusion coefficient in the small correlation time regime,
  as well as in the overdamped and underdamped limits.
  Finally,
  we employ a recently developed numerical method~\cite{roussel2018spectral} to calculate the effective diffusion coefficient for a wide range of (effective) friction coefficients,
  confirming our asymptotic results.
  \\[.1cm]

  % In this paper, we study the diffusive limit of solutions to the generalized Langevin equation (GLE) in a periodic potential. Under the assumption of quasi-Markovianity, we obtain sharp longtime equilibration estimates for the GLE using techniques from the theory of hypocoercivity. We then show asymptotic results for the effective diffusion coefficient in three different regimes: the small correlation, the overdamped and the underdamped limits. Finally, we employ a recently developed spectral numerical method in order to calculate the effective diffusion coefficient for a wide range of (effective) friction coefficients,

\end{abstract}

\newcommand{\n}{n}
\newcommand{\N}{N}

\section{Introduction}%
\label{sec:introduction}

The \emph{generalized Langevin \revision{equation}} (GLE) was originally proposed in the context of molecular dynamics and nonequilibrium statistical mechanics, in order to describe the motion of a particle interacting with a heat bath at equilibrium~\cite{mori1,mori1965continued,zwanzig1973nonlinear};
% It can be derived by model reduction using the Mori--Zwanzig formalism~\cite{MR2097022,MR3986068},
see also~\cite{KSTT02,MR2248987,pavliotis2011applied} for a rigorous derivation of the equation from a simple model of an open system,
consisting of a small Hamiltonian system coupled to an infinite-dimensional, Hamiltonian heat reservoir \revision{modeled by the linear wave equation}.
% with initial conditions distributed according to the canonical Gibbs measure.
The GLE has applications in many areas of science and engineering,
ranging from atom/solid-surface scattering~\cite{doll1975generalized} to polymer dynamics~\cite{snook2006langevin},
sampling in molecular dynamics~\cite{PhysRevLett.102.020601,ceriotti2010colored,ceriotti2010novel},
and global optimization with simulated annealing~\cite{gidas1985global,2020arXiv200306448C}.

The GLE is closely related,
in a sense made precise below,
to the simpler \emph{Langevin} (also known as \emph{underdamped Langevin}) equation,
which itself reduces to the \emph{overdamped Langevin} equation in the large friction limit.
Arranged from the simplest to the most general,
and written in one dimension for simplicity,
these three standard models are the following:
\begin{subequations}%
  \label{eq:models}%
  \begin{align}
    \label{eq:model:overdamped}%
    % \tag{ovd}
    \dot q &= - V'(q) + \sqrt{2 \, \beta^{-1}} \, \dot W, \\
    \label{eq:model:langevin}%
    % \tag{Lang}
    \ddot q &= - V'(q) - \gamma \, \dot q + \sqrt{2 \, \gamma \, \beta^{-1}} \, \dot W, \\
    \label{eq:model:generalized}%
    % \tag{GLE}
    \ddot q &= -V'(q) - \int_{0}^{t} \gamma(t-s) \, \dot q(s) \, \d s + F(t), \quad \text{where} \quad \expect(F(t_1) F(t_2)) = \beta^{-1} \, \gamma(|t_1-t_2|).
  \end{align}
\end{subequations}
Here $V$ is an external potential, $F$ is a \revision{stationary Gaussian} stochastic forcing,
$\beta$ is the inverse temperature,
the parameter $\gamma$ in~\eqref{eq:model:langevin} is the \emph{friction coefficient},
and the function $\gamma(\cdot)$ in~\eqref{eq:model:generalized} is a \emph{memory kernel}.
The constraint $\expect(F(t_1) F(t_2)) = \beta^{-1} \, \gamma(\revision{|t_1-t_2|})$ is known as the \emph{fluctuation/dissipation} relation,
and it guarantees that the canonical measure at temperature $\beta^{-1}$ is a stationary distribution of~\eqref{eq:model:generalized}; see \cref{sec:model_and_derivation_of_the_effective_diffusion}.
Since the force field in~\eqref{eq:model:generalized} is conservative
-- it derives from the potential~$V$ -- and the fluctuation/dissipation relation is assumed to hold,
equation~\eqref{eq:model:generalized} is sometimes called an \emph{equilibrium GLE}~\cite{MR3986068}.

The GLE~\eqref{eq:model:generalized} is a non-Markovian stochastic integro-differential equation which,
in general, is less amenable to analysis than the Langevin~\eqref{eq:model:langevin} and overdamped Langevin~\eqref{eq:model:overdamped} equations.
Instead of studying the GLE in its full generality,
we will restrict our attention to the case where the GLE is equivalent to a finite-dimensional system of Markovian stochastic differential equations (SDEs).
This assumption is known as the quasi-Markovian approximation,
and it is employed in many mathematical works on the GLE.
It is possible to show that it is verified when
the Laplace transform of the memory kernel is a finite continued fraction~\cite{mori1965continued} or,
relatedly, when the spectral density of the memory kernel is rational, in the sense of~\cite{MR1889227,MR2248987};
see also~\cite{pavliotis2011applied} for more details on the quasi-Markovian approximation.
In this paper, we will study two particular quasi-Markovian GLEs,
corresponding the cases where $\gamma(\cdot)$ is the autocorrelation function of scalar Ornstein--Uhlenbeck (OU) noise and harmonic noise;
see \cref{sec:model_and_derivation_of_the_effective_diffusion} for precise definitions.

For quasi-Markovian GLEs,
it is possible to \revision{rigorously} prove the passage to \cref{eq:model:langevin}
in the so-called \emph{white noise limit}.
This was done in~\cite{MR2793823} by leveraging recent developments in multiscale analysis~\cite{pavliotis2008multiscale}.
More precisely, in~\cite{MR2793823} the authors showed that,
with appropriate scalings,
the solution of the quasi-Markovian GLE with OU noise converges,
in the sense of weak convergence of probability measures on the space of continuous functions,
to that of the Langevin equation~\eqref{eq:model:langevin}
when the autocorrelation function of the noise converges to a Dirac delta measure.

Our objectives in this paper are twofold:
to study the longtime behavior of solutions to a simple quasi-Markovian GLE under quite general assumptions on the potential $V$ and,
based on this analysis,
to study scaling limits of the effective diffusion coefficient associated with the dynamics in the particular case where $V$ is periodic.

\paragraph{Longtime behavior.}%
\label{par:longtime_behavior_}
The longtime behavior of quasi-Markovian GLEs was studied in several settings in the literature.
The exponential convergence of the corresponding semigroup to equilibrium was proved in~\cite{MR1889227}.
In this paper,
which is part of a series of papers studying a model consisting of a chain of anharmonic oscillators coupled to Hamiltonian heat reservoirs~\cite{MR1685893,MR1705589,MR1764365},
the authors proved the convergence in an appropriately weighted $L^{\infty}$ norm,
by relying on Lyapunov-based techniques for Markov chains.
We also mention~\cite{MR1924935,MR1931266,rey-bellet,MR2857021,MR3509213} as useful references on the Lyapunov-based approach.
Later, in~\cite{MR2793823},
the exponential convergence to equilibrium was proved for the GLE driven by OU noise using Villani's hypocoercivity framework.
The authors showed the exponential convergence of the Markov semigroup both in relative entropy and in a weighted $H^1$ space.
More recently, exponential convergence results in an appropriately weighted $L^\infty$ norm were obtained in~\cite{MR3986068}
for a more general class of quasi-Markovian GLEs than had been considered previously,
allowing non-conservative forces and position-dependent noise.

Roughly speaking, the first aim of this paper is to obtain,
for quasi-Markovian GLEs, $H^1$ and $L^2$ convergence estimates similar to those of~\cite{MR2793823} but valid
\emph{uniformly over the space of parameters} that enter the equations, i.e. the parameters of the noise process driving the dynamics.
This turns out to be crucial for proving the validity of asymptotic expansions for the effective diffusion coefficient in several limits of interest -- our second goal.

\paragraph{Effective diffusion in a periodic potential.}%
\label{par:paragraph_name}
The behavior of a Brownian particle in a periodic potential has applications
\revision{in} many areas of science,
including electronics~\cite{MR0158437,strat2},
biology~\cite{MR1895137},
surface diffusion~\cite{gomer1990diffusion} and Josephson tunneling~\cite{barone1982physics}.
For the Langevin~\eqref{eq:model:langevin} and overdamped Langevin~\eqref{eq:model:overdamped} equations,
as well as for all finite-dimensional approximations of the GLE~\eqref{eq:model:generalized},
a functional central limit theorem (FCLT) holds under appropriate assumptions on the initial condition (e.g.\ stationarity, see~\cite[Theorem 2.5]{MR2793823}):
applying the diffusive rescaling,
the position process $q$ converges as $\varepsilon \to 0$,
in the sense of weak convergence of probability measures on $C([0, T], \real)$,
to a Brownian motion:
\begin{equation}
  \label{eq:functional_central_limit}
  \left\{ \varepsilon \, q(t/\varepsilon^2) \right\}_{t \in [0, T]} \Rightarrow \left\{ \sqrt{2 D} \, W(t) \right\}_{t \in [0, T]},
\end{equation}
where the \emph{effective diffusion coefficient} $D$ depends on the model and its parameters.
This is shown in, for example,~\cite{MR2427108} for the overdamped Langevin and Langevin dynamics,
and was proved more recently in~\cite[Theorem 2.5]{MR2793823} for finite-dimensional approximations of the GLE.

In \revision{spatial dimension one},
the behavior of the effective diffusion coefficient associated with Langevin dynamics~\eqref{eq:model:langevin} is well understood;
see, for example,~\cite{MR2394704} for a theoretical treatment
and~\cite{MR2427108} for numerical experiments.
The scaling of the effective diffusion coefficient with respect to the friction coefficient for Langevin dynamics has also been studied extensively in the physics literature.
Whereas in the large friction limit a universal bound scaling as $\frac{1}{\gamma}$ holds for the diffusion coefficient in arbitrary dimensions,
such a bound is true in the underdamped limit only in one dimension~\cite{MR2394704}.
Claims that underdamped Brownian motion in periodic and random potentials in dimensions higher than one can lead to anomalous diffusion
have been made~\cite{sancho_al04a,sancho_al04b} but \revision{seem} hard to justify rigorously.

The case of non-Markovian Brownian motion in a periodic potential has received less attention, even in one dimension.
Early quantitative results were obtained in~\cite{igarashi1988non} by means of numerical experiments
using the matrix-continued fraction method (see, e.g.,~\cite[Section~9.1.2]{MR987631}),
and verified in~\cite{igarashi1992velocity} by analog simulation.
In these papers,
the authors studied the dependence of the diffusion coefficient on the memory of the noise, and
they were also able to calculate the velocity autocorrelation function and to study
its dependence on the type of noise, i.e.\ OU or harmonic noise.
Given that few authors have investigated the problem quantitatively since then,
and in light of the increased computational power available today,
there is now scope for a more in-depth numerical study of the problem.

\paragraph{Our contributions.}%
\label{par:our_contributions}
Our contributions in this paper are the following:
\begin{itemize}
  \item
    We obtain sharp parameter-dependent estimates for the rate of convergence of the GLE to equilibrium
    in the particular cases of scalar OU and harmonic noises,
    thereby complementing previous results in~\cite{MR2793823}.
    Our approach is an explicit version of the standard hypocoercivity method~\cite{MR2562709,Herau07}
    and uses ideas from~\cite{OL15,MR3925138} for the definition of an appropriate auxiliary norm.
  \item
    We show rigorously that the diffusive and white noise limits commute for quasi--Markovian approximations of the GLE.
    In other words,
    assuming that the memory of the noise in the GLE is encoded by a small parameter $\nu$,
    and denoting by $D_{\gamma}$ and $D_{\gamma, \nu}$ the effective diffusion coefficients associated with~\eqref{eq:model:langevin} and~\eqref{eq:model:generalized},
    respectively, we prove that
    \[
      \lim_{\nu \to 0} D_{\gamma, \nu} = D_{\gamma}.
    \]
  \item
    For the case of OU noise,
    we study the influence on the effective diffusion coefficient of the friction coefficient
    that appears in the $\nu \to 0$ limiting Langevin equation,
    a coefficient that we will also refer to as the \emph{friction coefficient} by a slight abuse of terminology.
    We show in particular, both by rigorous asymptotics and by numerical experiments,
    that the diffusive limit commutes with the overdamped limit $\gamma \to \infty$.

  \item
    We corroborate most of our theoretical analysis by careful numerical experiments,
    thereby complementing the results of the early studies~\cite{igarashi1992velocity,igarashi1988non}.
    In these studies,
    because of the hardware limitations at the time,
    only about 15 basis functions per dimension could be used in 3 dimensions (3D)
    -- position, momentum, and one auxiliary variable --
    and very few simulations could be achieved in 4 dimensions (4D).
    With today's hardware and the availability of high-quality mathematical software libraries,
    we were able to run accurate simulations in both 3D and
    4D over a wide range of friction coefficients, \revision{including the underdamped limit~$\gamma\to 0$.}
    % over the ranges $\gamma \in [1/100, 1000]$ and $\gamma \in [1/10, 1000]$ respectively.
\end{itemize}

The rest of the paper is organized as follows.
In \cref{sec:model_and_derivation_of_the_effective_diffusion},
we present the finite-dimensional Markovian models of the GLE that we focus on throughout the paper
and we summarize our main results.
In \cref{sec:convergence_of_the_gle_dynamics},
we obtain an explicit estimate for the rate of convergence to equilibrium of the solution to the GLE.
In \cref{sec:multiscale_analysis},
we carry out a multiscale analysis with respect to the correlation time of the noise,
and we also study the overdamped and underdamped limits of the effective diffusion \revision{coefficient of the} GLE.
\Cref{sec:conclusions} is reserved for conclusions and perspectives for future work.

In the appendices,
we present a few auxiliary results:
in \cref{sec:confirmation_of_the_rate_of_convergence_in_the_quadratic_case},
we assess the sharpness of the convergence rate found in~\cref{sec:convergence_of_the_gle_dynamics},
\revision{in the particular case of a quadratic potential};
in \cref{sec:longtime_behavior_for_model_gl2},
we present a convergence estimate for harmonic noise;
in \cref{sec:estimates_underdamped_limit},
we derive the technical results used in \cref{sub:gle:the_underdamped_limit}.

\section{Model and main results}%
\label{sec:model_and_derivation_of_the_effective_diffusion}

The model and the results we present are all stated in a one-dimensional setting.
This allows to simplify the presentation and \revision{reduces} the number of parameters to be considered:
the mass of the system is set to~1
(instead of considering a general symmetric positive definite mass matrix)
and the friction~$\gamma(t)$ is scalar valued (whereas in general it would be a function with values in the set of symmetric positive matrices).
The extension of our analysis to higher dimensional cases poses however no difficulties for most of the arguments
-- with the notable exception of the underdamped limit in \cref{sub:gle:the_underdamped_limit}.

\subsection{Model}%
\label{sub:model}

Throughout this paper,
we assume that $V$ is \revision{a smooth} one-dimensional potential
that is either confining (in particular, $\e^{-\beta V} \in L^1(\real)$) or periodic with period $2 \pi$.
The configuration of the system is described by its position~$q \in \cX$ and the associated momentum~$p \in \real$.
Positions are either in~$\mathcal{X} = \real$ for confining potentials,
or in the torus $\cX = \torus = \revision{\real / 2\pi \integer}$ for periodic potentials.

\paragraph{General structure of the colored noise.}
Let us first consider a memory kernel of the form
\begin{equation}
  \label{eq:memory_kernel_markovian_approximation}
  \gamma(t) = \mathbbm{1}_{t \geq 0} \, \vect\lambda^\t \e^{-t \mat A} \vect \lambda,
\end{equation}
for a (possibly nonsymmetric) matrix $\mat A \in \real^{\n \times \n}$ with eigenvalues with positive real parts,
and a vector $\vect \lambda \in \real^\n$.
It is well-known that the GLE associated with~\eqref{eq:memory_kernel_markovian_approximation} is quasi-Markovian~\cite[Proposition 8.1]{pavliotis2011applied}:
it is equivalent to a Markovian system of stochastic differential equations (SDEs),
\begin{subequations}
  \label{eq:markovian_approximation}
  \begin{align}
    \label{eq:markovian_approximation_q}
    & \d q_t = p_t \, \d t, \\
    \label{eq:markovian_approximation_p}
    & \d p_t = - V'(q_t) \, \d t + \vect \lambda^\t \vect z_t \, \d t, \\
    \label{eq:markovian_approximation_z}
    & \d \vect z_t = - p_t \, \vect \lambda \, \d t -  \mat A \vect z_t \, \d t + \mat \Sigma \, \d \vect W_t, \qquad \vect z_0 \sim \mathcal N(0, \beta^{-1} \mat I_n),
  \end{align}
\end{subequations}
where $\Sigma \in \real^{\n \times \n}$ is related to $\mat A$ by the fluctuation/dissipation theorem:
\begin{equation*}
  % \label{eq:fluctuation_dissipation_markovian_approximation}
  \mat \Sigma \mat \Sigma^\t = \beta^{-1} \, \left(\mat A + \mat A^\t\right).
\end{equation*}
The equivalence comes from the fact that~\eqref{eq:markovian_approximation_z} can be integrated as
\[
  \vect z_t = -\int_0^t \e^{-(t-s) \mat A} \vect \lambda \, p_s \, \d s + \mathcal{N}_t,
  \qquad \mathcal{N}_t = \e^{-t \mat A} \vect z_0 + \int_0^t \e^{-(t-s)\mat A} \mat \Sigma \, \d\vect W_s,
\]
with $\expect \left( \mathcal{N}_t \right) = 0$ and, by the It\^o isometry,
\[
  \begin{aligned}
    \expect \left( \mathcal{N}_{t_1}\mathcal{N}_{t_2}^\t \right)
    & = \e^{-t_1 \mat A} \expect \left(\vect z_0 \vect z_0^\t\right) \e^{-t_2 \mat A^\t}
        + \frac1\beta \int_0^{\min(t_1,t_2)} \e^{-(t_1-s) \mat A} \left( \mat A + \mat A^\t \right) \e^{-(t_2-s) \mat A^\t} \d s \\
    & = \e^{-t_1 \mat A} \left( \frac{1}{\beta} \, \mat I_n \right) \e^{-t_2 \mat A^\t}
    + \frac1\beta \, \e^{-t_1 \mat A} \left(\int_0^{\min(t_1,t_2)} \derivative*{1}{u}\left(\e^{u \mat A} \, \e^{u \mat A^\t}\right) \Big|_{u=s} \, \d s \right) \e^{-t_2 \, \mat A^\t}\\
    & = \frac1\beta \e^{-(t_1-\min(t_1,t_2)) \mat A} \e^{-(t_2-\min(t_1,t_2)) \mat A^\t},
  \end{aligned}
\]
\revision{so} $\expect \left[ \left(\vect \lambda^\t \mathcal{N}_{t_1}\right)\left(\vect \lambda^\t \mathcal{N}_{t_2}\right)\right] = \beta^{-1} \gamma\left(|t_1-t_2|\right)$.

The dynamics~\eqref{eq:markovian_approximation} is ergodic with respect to the probability measure
\begin{align}
  \label{eq:mu}
  \mu(\d q \, \d p \, \d \vect z) = Z^{-1} \exp \left( -\beta \left( H(q, p) + \frac{{\abs{\vect z}}^2}{2} \right) \right) \d q \, \d p \, \d \vect z,
  \qquad H(q, p) = V(q) + \frac{p^2}{2},
\end{align}
with $Z$ the normalization constant.
\revision{Note} that the invariant measure is independent of the parameters of the noise $\vect \lambda$ and $\mat A$.
The generator of the Markov semigroup associated with the dynamics is given by
\[
  \mathcal L = \mathcal L_{\textrm{anti}} + \mathcal L_{\textrm{sym}},
\]
where the symmetric part $\mathcal L_{\rm sym}$ of the generator, considered as an operator on~$L^2(\mu)$, is related to the fluctuation and dissipation terms in~\cref{eq:markovian_approximation_z}; while the antisymmetric part $\mathcal L_{\textrm{anti}}$ corresponds to the Hamiltonian part of the dynamics (with Hamiltonian subdynamics for the couples $(q,p)$ and
$(p,\vect \lambda^\t \vect z)$) and an additional evolution in the $\vect z$ degrees of freedom associated with the antisymmetric part of the matrix~$\mat A$:
\[
\begin{aligned}
  \mathcal L_{\textrm{anti}} &=  p \, \derivative{1}{q} - \derivative*{1}[V]{q}(q) \, \derivative{1}{p} + \vect \lambda^\t \vect z \, \derivative{1}{p}
    - p \, \vect \lambda^\t \nabla_{\vect z}
    + \vect z^\t \mat A_{\rm a} \nabla_{\vect z}, \\
  \mathcal L_{\textrm{sym}} &= -\vect z^\t  \mat A_{\rm s} \nabla_{\vect z} + \beta^{-1} \, \mat A_{\rm s} : \nabla_{\vect z}^2,
\end{aligned}
\]
with $\mat A_{\rm s} = (\mat A + \mat A^\t)/2$ and $\mat A_{\rm a} = \revision{(\mat A - \mat A^\t)}/2$ the symmetric and antisymmetric parts of $\mat A$, respectively, $\nabla_{\vect z}^2$ the Hessian operator and $:$ the Frobenius inner product.

\paragraph{Specific models for the colored noise.}
In this study, we consider the two following models for the process $\vect z$:
\begin{description}
  \item[GL1]
    The noise is modeled by a scalar OU process ($\n = 1$),
    so $\vect \lambda$, $\mat A$, $\mat \Sigma$ and $\vect z$ are scalar quantities.
    We employ the parametrization
    \[
      \vect \lambda = \frac{\sqrt{\gamma}}{\nu}, \qquad \mat A = \frac{1}{\nu^2},
    \]
    for two positive parameters $\nu$ and $\gamma$. The associated memory kernel is
    \[
      \gamma(t) = \frac{\gamma}{\nu^2} \, \e^{-t/\nu^2} \, \mathbbm{1}_{t \geq 0}.
    \]
    Note that
    \[
      \int_0^{+\infty} \gamma(t) \, \d t = \gamma,
    \]
    which motivates the abuse of notation between the constant~$\gamma$ and the function~$\gamma(t)$
    (the meaning of the object~$\gamma$ under consideration should however be clear from the context).
    Moreover,
    \begin{equation}
      \label{eq:Langevin_limit}
      \gamma(t) \xrightarrow[\nu \to 0]{} \gamma \delta_0,
    \end{equation}
    which corresponds to a memoryless, Markovian limit.

  \item[GL2]
    The noise is modeled by a generalized version of harmonic noise:
    \begin{equation*}
      \vect \lambda = \frac{1}{\nu} \, \begin{pmatrix} \sqrt{\gamma} \\ 0 \end{pmatrix},
      \qquad
      \mat A = \frac{1}{\nu^2} \begin{pmatrix} 0 & -\alpha \\ \alpha & \alpha^2 \end{pmatrix},
      \qquad
      \mat \Sigma = \sqrt{\frac{2 \alpha^2}{\beta \nu^2}} \, \begin{pmatrix} 0 & 0 \\ 0 & 1 \end{pmatrix}.
    \end{equation*}
    The associated memory kernel is given by $\dps \frac{\gamma}{\nu^2} \, \e^{-2t/\nu^2} \, \mathbbm{1}_{t \geq 0}$ when $\alpha = 2$ and otherwise by
    \begin{equation}
      \label{eq:harmonic:autocorrelation}
      \gamma(t) = \frac{\gamma}{\nu^2} \, \exp\left( -\frac{\alpha^2 t}{2 \, \nu^2}\right) \left[
        \frac{\alpha}{\sqrt{\abs{4 - \alpha^{2}}}} \, s_{\alpha} \left(\frac{\sqrt{\abs{4 - \alpha^{2}}}}{2} \, \frac{\alpha t}{\nu^2} \right)
        + c_{\alpha} \left({\frac{\sqrt{\abs{4 - \alpha^{2}}}}{2}} \, \frac{\alpha t}{\nu^2} \right) \right] \mathbbm{1}_{t \geq 0},
    \end{equation}
    where $(s_{\alpha}, c_{\alpha})$ are the functions
    $(\sin, \cos)$ when $\alpha < 2$ and $(\sinh, \cosh)$ when $\alpha > 2$.
    The latter expression can be found by computing the eigenvalues of~$\mat A$, writing the solution as a sum of exponentials of these eigenvalues multiplied by the time~$t$, and adjusting the coefficients in the linear combination so that $\gamma(0) = \vect\lambda^\t \vect\lambda = \gamma/\nu^2$ and $\gamma'(0^+) = -\vect \lambda^\t \mat A \vect \lambda$.

    In particular, since $\alpha^2 - \alpha \sqrt{\alpha^2 - 4}\to 2$ as $\alpha \to \infty$,
    we obtain that
    \[
      \gamma(t) \xrightarrow[\alpha \to \infty]{} \frac{\gamma}{\nu^2} \, \e^{-t/\nu^2} \, \mathbbm{1}_{t \geq 0}
    \]
    for \revision{any} $t \geq 0$, which is the autocorrelation function of the noise in the model GL1.
    The limit $\alpha \to \infty$ corresponds to an overdamped limit of the noise since~\eqref{eq:markovian_approximation_z} reads,
    in the absence of the forcing term~$-p_t \vect\lambda \, \d t$ and with the notation $\vect z = (z_1,z_2)$,
    \[
      \begin{aligned}
        \d z_{1,t} & = \frac{\alpha}{\nu^2} z_{2,t} \, \d t, \\
        \d z_{2,t} & = -\frac{\alpha}{\nu^2} z_{1,t} \, \d t - \frac{\alpha^2}{\nu^2} z_{2,t} \, \d t + \sqrt{\frac{2 \alpha^2}{\beta \nu^2}} \, \d W_t, \\
      \end{aligned}
    \]
    which, after rescaling time by a factor~$\alpha/\nu^2$, corresponds to a Langevin dynamics with friction~$\alpha$ for the $\vect z$ variable.
\end{description}

In both models~GL1 and~GL2,
the parameters $\gamma$ and $\nu$ (or, rather, $\nu^2$) enter as scalings in the autocorrelation function,
with $\nu$ being essentially the square root of the correlation time of the noise.
In \revision{the} model~GL2, the parameter $\alpha$ encodes the shape of the function.
% We note that $\int_{0}^{\infty} \gamma(t)  \, \d t = \gamma$ in both cases.
Examples of memory kernels are illustrated in~\cref{fig:memory_kernel} for the two models under consideration and various values of~$\alpha$.

\begin{figure}[ht]
  \centering
  \includegraphics[width=0.8\linewidth]{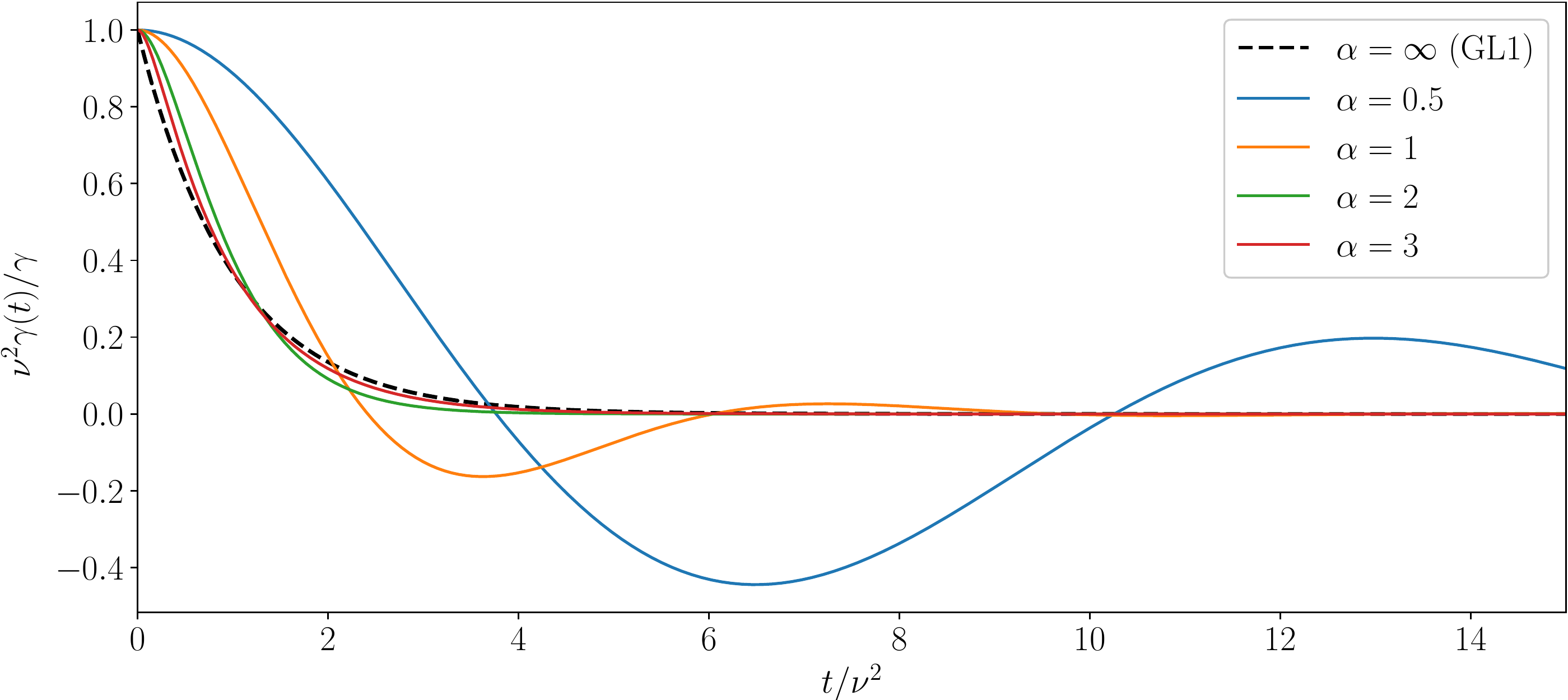}
  \caption{Memory kernels associated with models GL1 and GL2 (with time rescaled by a factor~$\nu^{-2}$ and memory kernel rescaled by a factor~$\nu^2/\gamma$ in order to have shapes independent of the choices of~$\gamma,\nu$).}
  \label{fig:memory_kernel}
\end{figure}

\paragraph{Effective diffusion coefficient.}
We consider here the case when $\cX = \torus$,
since the definition of an effective diffusion does not make sense for confining potentials.
The derivation of the effective diffusion coefficient for systems of SDEs of the type~\eqref{eq:markovian_approximation} in a periodic potential is well understood;
see for example~\cite{MR2427108, pavliotis2008multiscale} for a formal derivation and~\cite{MR2793823} for a rigorous proof for GLE. The diffusion coefficient can be expressed in terms of the solution to a Poisson equation:
\begin{subequations}
\begin{align}
  \label{eq:introduction:effective_diffusion}
  D &= \int_{\torus \times \real \times \real^\n}  \phi \, p \, \d \mu, \\
  \label{eq:introduction:poisson_equation}
  - \mathcal L \phi &= p.
\end{align}
\end{subequations}
The Poisson equation~\cref{eq:introduction:poisson_equation} is equipped with periodic boundary conditions in $[-\pi, \pi]$
and the condition that $\phi$ is square integrable with respect to the Gibbs measure,
\emph{i.e.} $\phi \in \lp{2}{\mu}$. In fact, in order to guarantee the uniqueness of the solution to~\eqref{eq:introduction:poisson_equation}, one should further assume that~$\phi$ has average~0 with respect to~$\mu$.  Since $\mathcal L_{\rm FD}$ and $\mathcal L_{\textrm{ham}}$ are symmetric and antisymmetric operators on $\lp{2}{\mu}$, respectively, and since $\mathcal L_{\rm FD} = \beta^{-1} \, \e^{\beta \abs{\vect z}^2} \, \nabla_{\vect z} \cdot (\e^{-\beta \abs{\vect z}^2} \mat A_{\rm s} \nabla_{\vect z})$,
the effective diffusion coefficient can also be rewritten as
\begin{equation}
  \label{eq:diffusion_coefficient_depending_on_grad_z}
  D = \beta^{-1} \, \int_{\torus \times \real \times \real^\n} \nabla_{\vect z}\phi^\t \mat A_{\rm s} \nabla_{\vect z} \phi\, \d \mu.
\end{equation}
Note finally that definitions similar to~\cref{eq:introduction:effective_diffusion,eq:introduction:poisson_equation} hold also for the Langevin dynamics, provided that $\mathcal L$ is defined as the corresponding generator (see~\eqref{eq:Dgamma_Lang}).

\subsection{Main results}%
\label{sub:main_results}
Before stating our main results, let us introduce some notation.
In this paper, $H^1(\mu)$ denotes the subspace of $\lp{2}{\mu}$ of functions whose gradient is in~$\lp{2}{\mu}$,
equipped with the usual weighted Sobolev inner product.
The spaces $L^2_0(\mu)$ and $H^1_0(\mu)$ are the subspaces of $\lp{2}{\mu}$ and $H^1(\mu)$ of functions with mean~0 with respect to~$\mu$, respectively,
and $\mathcal B(E)$ is the space of bounded linear operators on a Banach space~$E$,
equipped with the operator norm
\[
  \norm{\mathcal A}[\mathcal B(E)] = \sup_{f \in E\backslash \{0\} } \frac{\norm{\mathcal A f}[E]}{\norm{f}[E]}.
\]
\paragraph{Exponential decay and resolvent estimates.}%
\label{par:longtime_estimates}
Our first result concerns the convergence to equilibrium for the model GL1.
In line with the standard approach in molecular dynamics,
we state our results for the semigroup $\e^{t\mathcal L}$,
\revision{which describes the time evolution of average properties,}
but we note that our estimates apply equally to the adjoint semigroup $\e^{t \mathcal L^*}$,
where $\mathcal L^*$ denotes the $\lp{2}{\mu}$ adjoint of $\mathcal L$.
This can be shown by duality and,
as mentioned in~\cite{MR3509213, pavliotis2011applied},
can also be understood from the fact that
$\mathcal L^* = - \mathcal L_{\textrm{ham}} + \mathcal L_{\textrm{FD}}$ coincides with $\mathcal L$ up to the sign of the Hamiltonian part.
In turn, convergence estimates for $\e^{t\mathcal L^*}$ can be translated into
estimates for $\e^{t \mathcal L^\dagger}$,
where $\mathcal L^\dagger$ denotes the Fokker--Planck operator associated with the dynamics.
This is because $\e^{t\mathcal L^\dagger} \psi = \mu \, \e^{t \mathcal L^*} (\mu^{-1} \psi)$
for any test function $\psi \in L^2(\mu^{-1})$ where,
by a slight abuse of notation,
$\mu$ denotes in this context the Lebesgue density of the measure~\eqref{eq:mu}.

Although our results on the effective diffusion apply only to the case of a periodic potential,
\cref{thm:hypocoercivity_h1,thm:hypoelliptic_regularization} below apply also when $V$ is a confining potential,
provided that the following assumption holds.
\begin{assumption}
  \label{assumption:assumption_potential}
  The potential $V$ is smooth. When the position space is~$\real$ we moreover assume that the following conditions are satisfied:
  \begin{enumerate}[(i)]
    \item $\e^{-\beta V} \in L^1(\real)$,
    \item $\dps \norm{V''}[\infty]:=\sup_{q \in \real}|V''(q)| < \infty$,
    \item the following Poincar\'e inequality holds true for some constant $R_{\beta} > 0$:
  \begin{align}
    \label{eq:poincare}
    \forall \varphi \in H^1(\e^{-\beta V}), \qquad \norm{\varphi - \bar \varphi}[L^2(\e^{-\beta V})]^2
    \leq \frac{1}{R_{\beta}} \, \norm{\nabla \varphi}[L^2(\e^{-\beta V})]^2,
    \qquad \bar \varphi := \frac{\dps \int_\cX \varphi \e^{-\beta V}}{\dps \int_\cX \e^{-\beta V}}.
  \end{align}
  \end{enumerate}
\end{assumption}
Note that, for smooth periodic potentials,
the Poincar\'e inequality~\eqref{eq:poincare} holds true without any additional condition on~$V$
(see for instance the discussion in~\cite[Section~2.2.1]{MR3509213}).
\revision{The above assumption allows us to prove the following results.}
\begin{theorem}[Hypoelliptic regularization]
\label{thm:hypoelliptic_regularization}
Let $\mathcal L$ denote the generator associated with the model GL1
and suppose that \cref{assumption:assumption_potential} holds.
\revision{Then for any $h \in L^2_0(\mu)$ and any parameters~$\gamma > 0$ and $\nu > 0$,
it holds $\e^{t \mathcal L} h \in H^1_0(\mu)$ for all $t > 0$ and
there is an inner product $\iip{\dummy}{\dummy}_{\gamma, \nu}$ equivalent to the usual $H^1(\mu)$ inner product such that}
  \begin{align}
    \label{eq:hypoelliptic_regularization}
    \forall h \in L^2_0(\mu), \qquad
    \iip{\e^{\mathcal L} h}{\e^{\mathcal L} h}_{\gamma,\nu} \leq \norm{h}^2_{L^2(\mu)}.
  \end{align}
\end{theorem}

We \revision{next} state a convergence result in $H^1_0(\mu)$,
which relies on the hypocoercive framework described in~\cite{MR2562709}.

\begin{theorem}[\revision{$H^1(\mu)$} hypocoercivity]
\label{thm:hypocoercivity_h1}
  Suppose that~\cref{assumption:assumption_potential} holds,
  \revision{and consider the inner product $\iip{\dummy}{\dummy}_{\gamma, \nu}$ constructed in the proof of \cref{thm:hypoelliptic_regularization}.}
  Then there exists a constant $C_1 \in \real_+$ such that, \revision{for any $t \geq 0$} and any parameters $\nu > 0$ and $\gamma > 0$, it holds
  \begin{align}
    \label{eq:hypocoercivity_h1_auxiliary}
    \forall f \in H^1_0(\mu), \qquad \iip{\e^{t \mathcal L}f}{\e^{t \mathcal L}f}_{\gamma, \nu}
    \leq \exp \left(- C_1 \min \left( \gamma, \frac{1}{\gamma}, \frac{\gamma}{\nu^4} \right) t \right) \iip{f}{f}_{\gamma, \nu}.
  \end{align}
  In particular, there exists $C_2(\gamma,\nu) \in \real_+$ such that
  \[
    \norm{\e^{t \mathcal L}}[\mathcal B\left(H^1_0(\mu)\right)]
    \leq C_2(\gamma, \nu) \, \exp \left(- C_1 \min \left( \gamma, \frac{1}{\gamma}, \frac{\gamma}{\nu^4} \right) t \right).
  \]
\end{theorem}

\Cref{thm:hypocoercivity_h1,thm:hypoelliptic_regularization} can be combined to obtain a decay estimate in $L^2(\mu)$.
More precisely, employing the fact that $\norm{\dummy}^2 \leq \iip{\dummy}{\dummy}_{\gamma, \nu}$
(see the definition~\eqref{eq:hypocoercivity:inner_product} and~\eqref{eq:hypocoercivity:equivalence_with_sobolev_norm}) together with~\cref{eq:hypocoercivity_h1_auxiliary,eq:hypoelliptic_regularization},
we obtain, for $h \in L^2_0(\mu)$,
\begin{align}
  \norm{\e^{t \mathcal L}h}^2_{L^2(\mu)}
  \leq \iip{\e^{t \mathcal L}h}{\e^{t \mathcal L}h}_{\nu, \gamma}
  \notag
  &\leq \exp \left(-2 \, C_1 \, \min \left(\gamma, \frac{1}{\gamma}, \frac{\gamma}{\nu^4} \right) (t - 1) \right) \, \iip{\e^{\mathcal L} h}{\e^{\mathcal L} h}_{\nu,\gamma} \\
  \notag
  &\leq \exp \left(-2 \, C_1 \, \min \left(\gamma, \frac{1}{\gamma}, \frac{\gamma}{\nu^4} \right) (t - 1) \right) \, \norm{h}^2_{L^2(\mu)} \\
  \label{eq:exponential_growth_bound}%
  &\leq \e^{2 C_1}  \exp \, \left(-2 \, C_1 \, \min \left(\gamma, \frac{1}{\gamma}, \frac{\gamma}{\nu^4} \right) t  \right) \, \norm{h}^2_{L^2(\mu)}
\end{align}
for any $t \geq 1$. This inequality also holds for $0 \leq t \leq 1$
in view of the trivial bound $\norm{\e^{t \mathcal L}}[\mathcal{B}\left(L_0^2(\mu)\right)] \leq 1$,
so the following resolvent bound holds~\cite[Proposition 2.1]{MR3509213}.

\begin{corollary}
  Under the same assumptions as in \cref{thm:hypocoercivity_h1},
  there exist a constant~$C \in \real_+$ independent of $\gamma$ and $\nu$ such that
\begin{align}
  \label{eq:resolvent_bound_GL1}
  \norm{\mathcal L^{-1}}[\mathcal B\left(L^2_0(\mu)\right)] \leq C \, \max \left( \gamma, \frac{1}{\gamma}, \frac{\nu^4}{\gamma} \right).
\end{align}
\end{corollary}

In fact, $C = \e^{2 C_1}/2C_1$.
The dependence of the exponential decay rate in~\eqref{eq:exponential_growth_bound} with respect to the parameters~$\nu,\gamma$ is illustrated in~\cref{figure:exponential_growth_bound}.
It is worth comparing the scaling of the exponential decay rate to the scalings obtained for Langevin dynamics,
for which $\lambda_{\rm Lang}(\gamma)$ scales as $\min(\gamma,\gamma^{-1})$, \revision{as proved in~\cite{DKMS13,GS16,MR3925138} (see also the discussion at the beginning of~\cref{sec:convergence_of_the_gle_dynamics}).}
The rates obtained for GLE are therefore in line with these rates in the limit $\nu \to 0$,
which is precisely the limit~\eqref{eq:Langevin_limit} in which GLE reduces to Langevin dynamics.
The additional term~$\gamma/\nu^4$ in the scaling factor for the exponential decay rate of GLE is important only in the limit~$\nu \to \infty$
(with the additional condition $\nu^2 \gg \gamma$ if the limits $\gamma,\nu \to \infty$ are taken simultaneously).

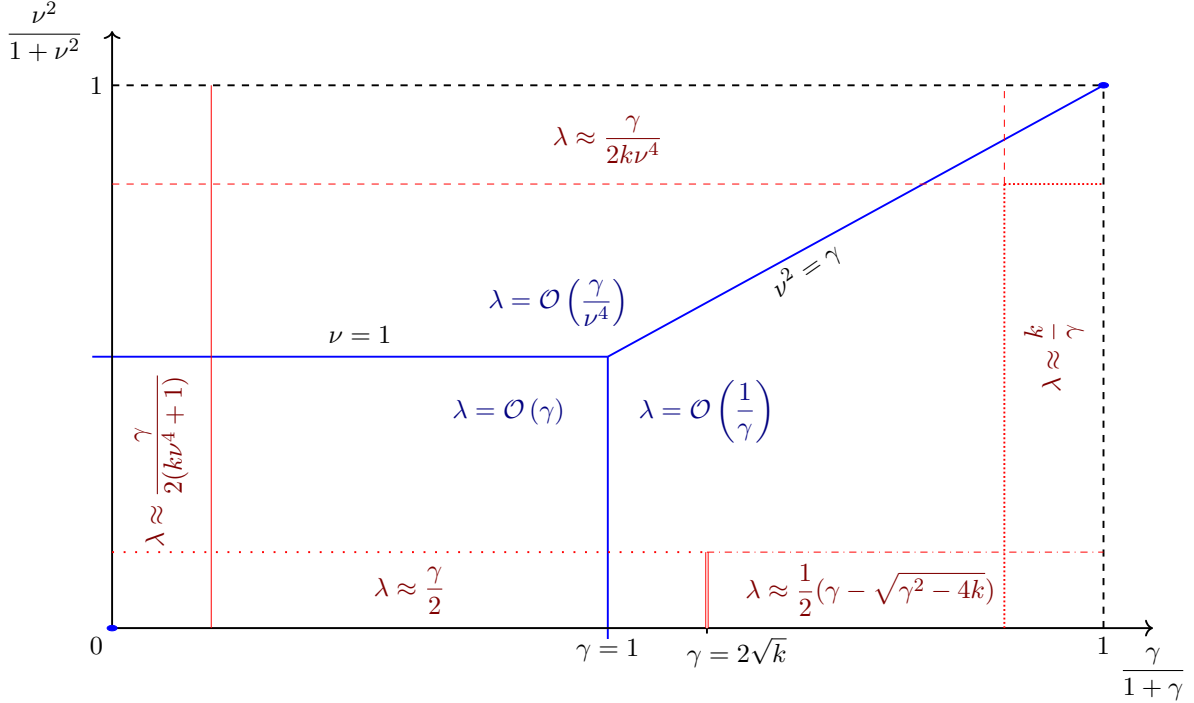
\begin{figure}[ht!]
  \centering
  \resizebox{.99\textwidth}{!}{%
  \begin{tikzpicture}[xscale=14.5,yscale=8]

   % \filldraw[fill=gray, draw=darkred]
    % (1,1) -- (.9,1) arc (180:270:.1) -- cycle;

    % \draw [fill=orange,orange] (.9,.8181) rectangle (1,1);

    \draw[blue,thick] (.5,.5) -- (1,1);
    \draw[blue,thick] (.5,-.02) -- (.5,.5);
    \draw[blue,thick] (-.02,.5) -- (.5,.5);
    \node[rotate=35,black] at (.7, .66) {$\nu^2 = \gamma$};
    \node[black] at (.5, -.04) {$\gamma = 1$};
    \node[black] at (.25, .54) {$\nu = 1$};
    \node[darkblue] at (.45, .6) {$\dps \lambda = \bigo{\frac{\gamma}{\nu^4}}$};
    \node[darkblue] at (.6, .4) {$\dps \lambda = \bigo{\frac{1}{\gamma}}$};
    \node[darkblue] at (.4, .4) {$\dps \lambda = \bigo{\gamma}$};

    \draw[thick,->] (0,0) -- (1.05,0);
    \draw[-] (0,0) -- (1,0);
    \draw[thick,->] (0,0) -- (0,1.1);
    \draw[dashed,thick] (1,0) -- (1,1);
    \draw[dashed,thick] (0,1) -- (1,1);
    \node[below left] at (0, 0) {$0$};
    \node[below] at (1, 0) {$1$};
    \node[left] at (0, 1) {$1$};
    \node[below] at (1.05, -0.02) {$\dps \frac{\gamma}{1 + \gamma}$};
    \node[left] at (-0.02, 1.1) {$\dps \frac{\nu^2}{1 + \nu^2}$};
    % \draw[double,scale=1,domain=.8:1.,variable=\x,samples=500,red] plot ({\x},{\x/(2 - \x)});
    % \node[rotate=45,black] at (.95, .85) {$\gamma = 2 k \nu^2$};
    \draw[fill,blue] (1,1) circle [radius=.005];
    \draw[fill,blue] (0,0) circle [radius=.005];

    \node[darkred] at (.5, .9) {$\dps \lambda \approx \frac{\gamma}{2k\nu^4}$}; %{$\dps \lambda \approx \frac{\gamma}{2(k\nu^4 + \gamma \nu^2)}$};
    \draw[dashed,red] (0,.8181) -- (.9,.8181);
    \draw[dashed,red] (.9,.8181) -- (.9, 1);
    \node[rotate=90,darkred] at (.95, .5) {$\dps \lambda \approx \frac{k}{\gamma} $}; %{$\lambda \approx \frac{k}{\gamma + k\nu^2} $};
    \draw[thick,densely dotted,red] (.9,0) -- (.9,.8181);
    \draw[thick,densely dotted,red] (.9,.8181) -- (1,.8181);

    \draw[red] (.1, 0) -- (.1,1);
    \node[rotate=90,darkred] at (.05, .31) {$\dps \lambda \approx \frac{\gamma}{2(k\nu^4 + 1)} $};

    \draw[black,thick] (.6, -.01) -- (.6,0);
    \node[below] at (.63, .0) {$\gamma = 2 \sqrt{k}$};
    \draw[loosely dotted, thick, red] (0, .14) -- (.6,.14);
    \node[darkred] at (.3, .07) {$\dps \lambda \approx \frac{\gamma}{2} $};
    \draw[double,red] (.6, 0) -- (.6,.14);

    \draw[dash dot,red] (.60, .14) -- (1,.14);
    \node[darkred] at (.765, .07) {$\dps \lambda \approx \frac{1}{2} (\gamma - \sqrt{\gamma^2 - 4k})$};
  \end{tikzpicture}}
  \caption{%
    Schematic illustration of the scaling of the exponential decay rate~$\lambda$ in~$L^2_0(\mu)$.
    The general estimate,
    implied by~\cref{eq:exponential_growth_bound} and valid under the general~\cref{assumption:assumption_potential},
    is illustrated in blue.
    The behavior of the exponential growth bound
    in the particular case of the quadratic potential $V(q) = k\frac{q^2}{2}$
    is depicted in red.
    In this context, the symbol $\approx$ means that
    the relative error is arbitrarily close to zero in the corresponding limit.
  }%
  \label{figure:exponential_growth_bound}
\end{figure}

\paragraph{Sharpness of the bounds on the exponential decay rate.}
In the particular case where~$V$ is the quadratic potential $k\frac{q^2}{2}$ with $k>0$,
the scaling of the exponential growth bound \revision{with respect to~$\gamma,\nu$} in all the limits of interest can be calculated explicitly.
Indeed, in this case~\cref{eq:markovian_approximation} can be written in the general form
\begin{equation}
  \label{eq:OU_form_GLE}
  \dot {\vect X} = \mat D \vect X \, \d t + \mat \sigma \, \d \vect W,
\end{equation}
where $\vect X^\t = (q, p, \vect z^\t)$, $\mat D$ and $\mat \sigma$ are constant matrices,
and $\vect W$ is a standard Brownian motion on $\real^{2 + n}$.
It is known by a result from Metafune, Pallara and Priola~\cite{MR1941990} that
the corresponding generator $\mathcal L$ generates a strongly continuous and compact semigroup in $\lp{p}{\mu}$ for any $p \in (1, \infty)$,
and that the associated spectrum can be obtained explicitly by a linear combination of
the eigenvalues of the drift matrix $\mat D$ in~\cref{eq:markovian_approximation}:
\begin{equation}
  \label{eq:sigma_OU_form_GLE}
  \sigma(\mathcal L) = \left\{ \sum_{\lambda \in \sigma(\mat D)} \lambda \, k_{\lambda}, \quad k_{\lambda} \in \nat_{\geq 0} \right\};
\end{equation}
see also~\cite[Section 9.3]{lorenzi2006analytical} and~\cite{MR2899986}.
By \cite[Theorem 5.3]{MR1886588},
the spectral bound of the generator,
\emph{i.e.} the eigenvalue of~$\mathcal{L}$ with the largest real part,
coincides up to a sign change with the exponential decay rate of the semigroup
(and in fact the norms of the propagators~$\e^{t \mat D}$ and~$\e^{t \mathcal{L}}$ coincide,
as made precise in~\cite{ASS20}),
so estimating this growth bound in the quadratic case amounts to
calculating the eigenvalues of the matrix $\mat D$.
This can be achieved either numerically or analytically in the limiting regimes where the parameters go to either~0 or~$\infty$,
based on rigorous asymptotics for the associated characteristic polynomial (as made precise in~\cref{sec:confirmation_of_the_rate_of_convergence_in_the_quadratic_case}). The behavior of the spectral bound in the limiting regimes is indicated in \cref{figure:exponential_growth_bound}.

\paragraph{Scaling limits for the effective diffusion coefficient.}%
In \cref{sub:gle:the_overdamped_limit,sub:gle:the_underdamped_limit,sub:gle:short_memory_limit},
we establish, either formally or rigorously,
the limits in solid arrows in the following diagram
(the limits in dashed arrows are already known results):
\begin{equation*}
  \begin{tikzpicture}[baseline={([yshift=-.5ex]current bounding box.center)}]
    \node (underdamped) at (-2, 0) {$D_{\rm und}$};
    \node (nolimit) at (3, 0) {$\gamma \, D_{\gamma}$};
    \node (overdamped) at (6, 0) {$D_{\rm ovd}$};
    \node (underdampedGL) at (-2, 2) {$D_{\nu}^*$};
    \node (nolimitGL) at (3, 2) {$\gamma \, D_{\gamma,\nu}$};
    \draw[thick,->] (nolimitGL) -- (overdamped) node [midway, above right] {\footnotesize{$\gamma \to \infty$ (\cref{sub:gle:the_overdamped_limit})}};
    \draw[thick,->] (nolimitGL) -- (underdampedGL) node [midway, above] {\footnotesize{$\gamma \to 0$ (\cref{sub:gle:the_underdamped_limit})}};
    \draw[dashed,thick,->] (nolimit) -- (underdamped) node [midway, below] {\footnotesize{$\gamma \to 0$}};
    \draw[dashed,thick,->] (nolimit) -- (overdamped) node [midway, below] {\footnotesize{$\gamma \to \infty$}};
    \draw[thick,->] (underdampedGL) -- (underdamped) node [midway, left] {\footnotesize{$\nu \to 0$ (\cref{sub:gle:the_underdamped_limit})}};
    \draw[thick,->] (nolimitGL) -- (nolimit) node [midway, left] {\footnotesize{$\nu \to 0$ (\cref{sub:gle:short_memory_limit})}};
  \end{tikzpicture}%
\end{equation*}
Here $D_{\rm ovd}$ denotes the effective diffusion coefficient associated with the overdamped Langevin dynamics~\eqref{eq:model:overdamped}, while $D_{\rm und}$ and~$D_\gamma$ are diffusion coefficients associated with the Langevin dynamics~\eqref{eq:model:langevin}.
Let us emphasize that $D^*_{\nu}$ depends only on $\nu$ and that $D_{\rm und}$ is a constant, independent of the parameters of the noise, that can be calculated using the approach outlined in~\cite{MR2427108}.

More precisely,
we establish rigorously in \cref{sub:gle:short_memory_limit} that,
for fixed $\gamma$ and for quite general Markovian approximations of the noise,
$D_{\gamma,\nu} = D_{\gamma} + \bigo{\nu^2}$ as $\nu \to 0$.
Our proof is based on an asymptotic expansion of the solution to the Poisson equation~\eqref{eq:introduction:poisson_equation}
and on the resolvent estimate~\eqref{eq:resolvent_bound_GL1}.
Using similar techniques,
we show in \cref{sub:gle:the_overdamped_limit}
that $D_{\gamma,\nu} \to D_{\rm ovd}$ in the limit $\gamma \to \infty$ for fixed $\nu>0$ and
in the particular case of model GL1.
Finally, in \cref{sub:gle:the_underdamped_limit} we motivate,
with a formal asymptotic expansion similar to that employed in~\cite{MR2427108},
that $\gamma D_{\gamma,\nu} \to D^*_{\nu}$ in the underdamped limit $\gamma \to 0$.

In principle, we could also study the limit $\nu \to \infty$.
However, we refrain from doing so here because,
first, this limit is less relevant from a physical viewpoint than the other limits considered and,
second, this limit is technically more difficult.
The main difficulty originates from the fact that the leading-order part of the generator is $p \, \derivative{1}{q} - \derivative*{1}[V]{q} \, \derivative{1}{p}$,
so the terms in the asymptotic expansion of the solution to~\eqref{eq:introduction:poisson_equation} are not explicit.
In addition, the operator norm of the resolvent scales as $\nu^{4}$,
so a large number of terms are required for proving a rigorous result.

\subsection{Numerical experiments}%
\label{par:numerical_results}

Here we verify numerically the limits $\gamma D_{\gamma,\nu} \to D_{\rm ovd}$ as $\gamma \to \infty$ and $D_{\gamma,\nu} \to D_{\gamma}$ as $\nu \to 0$, with a spectral method to approximate the solution to the Poisson equation~\eqref{eq:introduction:poisson_equation}.
We come back to the case when $\cX = \torus$.
We employ a Galerkin method that is in general non-conformal,
in the sense that the finite-dimensional approximation space,
which we denote by $V_N$,
does not necessarily contain only mean-zero functions with respect to $\mu$.
Following the ideas developed in~\cite{roussel2018spectral},
we use a saddle point formulation to obtain an approximation of the solution to \cref{eq:introduction:poisson_equation}:
\begin{align}
  \label{eq:saddle_point_formutation}
  \left\{
    \begin{aligned}
       & - \Pi_N \, \mathcal L \, \Pi_N \Phi_N + \alpha_N u_N = \Pi_N p, \\
       & \Phi_N^\t u_N = 0,
    \end{aligned}
  \right.
\end{align}
where $\Pi_N$ is the $\lp{2}{\mu}$ projection operator on $V_N$,
% satisfying $\ip{\Pi_N u - u}{v_N} = 0$ for all functions $v_N \in V_N$,
$u_N = \Pi_N 1 / \norm{\Pi_N 1} \in V_N$
and $\alpha_N$ is a Lagrange multiplier.
As above,
$\ip{\cdot}{\cdot}$ and $\norm{\cdot}$ denote respectively the standard scalar product and norm of $\lp{2}{\mu}$.
We choose $V_N$ to be the subspace of $\lp{2}{\mu}$ spanned by tensor products of appropriate one-dimensional functions
constructed from trigonometric functions (in the $q$ direction) and Hermite polynomials (in the $p$ and $\vect z$ directions).
In the case of OU noise, for example,
we use the basis functions
\begin{equation*}
  \label{eq:basis_functions}
  e_{i,j,k} = Z^{1/2} \, \e^{\frac{\beta}{2} \left( H(q,p) + \frac{z^2}{2} \right)}
  \, G_i(q) \, H_j(p) \, H_k(z), \qquad 0 \leq i,j,k \leq N,
\end{equation*}
where $G_i$ are trigonometric functions,
\begin{equation}
  \label{eq:definition_trigonometric_functions}
  G_i(q) =
  \left\{ \begin{aligned}
    (2 \pi)^{-1/2}, \quad & \text{if}~i = 0, \\
    \pi^{-1/2} \sin\left(\frac{i + 1}{2}q\right), \quad & \text{if}~i~\text{is odd}, \\
    \pi^{-1/2} \cos\left(\frac{i}{2}q\right), \quad & \text{if}~i~\text{is even}, i > 0. \\
  \end{aligned} \right.
\end{equation}
and $H_j$ are rescaled normalized Hermite functions,
\begin{equation}
  \label{eq:definition_hermite_functions}
  H_j(p) = \frac{1}{\sqrt{\sigma}} \, \psi_j \left( \frac{p}{\sigma} \right),
  \qquad \psi_j (p) := (2 \pi)^{-\frac{1}{4}} \frac{(-1)^j}{\sqrt{j!}} \e^{\frac{p^2}{4}} \, \derivative*{j}{p^j} \, \left( \e^{- \frac{p^2}{2}} \right).
\end{equation}
The functions $(H_j)_{j\in \nat}$ are orthonormal in $\lp{2}{\real}$ regardless of the value of $\sigma$,
\revision{and $\sigma$ is} a scaling parameter that can be adjusted to better resolve $\Phi_N$.
We use the same number $N$ of basis functions in every direction because,
% It might appear inefficient to use the same number of basis functions in every direction,
although Hermite series converge much slower than Fourier series as $N \to \infty$ when $\sigma$ is fixed,
their spatial resolution is comparable to that of Fourier series when $\sigma$ is chosen appropriately,
as demonstrated in~\cite{tang1993hermite}.
% More specifically, while $\mathcal O(k^2)$ Hermite modes are required in the former case to resolve $k$ wavelengths of the Gaussian-windowed cosine function ($\cos(kx) \, \e^{-x^2/2}$),
% only $\mathcal O(k)$ are required when $\sigma \propto 1/\sqrt{N}$.

To solve the linear system associated with~\eqref{eq:saddle_point_formutation} and the basis functions~\eqref{eq:definition_trigonometric_functions},
we use either the \emph{SciPy}~\cite{scipy} function \texttt{scipy.sparse.linalg.spsolve},
which implements a direct method,
or, when the time or memory required to solve \cref{eq:saddle_point_formutation} with a direct method is prohibitive,
the function \texttt{scipy.sparse.linalg.gmres},
which \revision{implements} the generalized minimal residual method (GMRES)~\cite{SS86}.
% Once $\Phi_N$ has been calculated,
% the effective diffusion coefficient is approximated by:
% \begin{equation}
    % \label{eq:gle:approximate_diffusion_coefficient}
    % D_{N} \approx \ip{\Phi_N}{\Pi_N p},
% \end{equation}
% where the projection $\Pi_N p$ is calculated by numerical quadrature.

The numerical results presented below are for the one-dimensional periodic cosine potential $V(q) = \frac{1}{2} \, (1 - \cos q)$,
\revision{and they were all obtained with $\beta = 1$.}
We examine the variation of the diffusion coefficient with respect to $\gamma$ for fixed $\nu = 1$ in \cref{fig:numerics:diffusion_coefficient}.
The parameters used in the simulations are presented in \cref{tab:gle:numerical_method_for_computation_diffusion}.
The effective diffusion coefficient was computed for 100 values of $\gamma$ evenly spaced on a logarithmic scale,
and for each value of $\gamma$ the numerical error was approximated by carrying out the computation with half the \revision{number} of basis functions in each direction.
When the  relative error estimated in this manner was over 1\%,
which occurred roughly when $\gamma \leq 10^{-2}$ for the model GL1 and $\gamma \leq 10^{-1}$ for the model GL2,
the corresponding data points were considered inaccurate and were removed.
\begin{table}[htpb]
  \centering
  \begin{tabular}{c|c|c}
    Model  & Method when $\gamma < 1$ & Method when $\gamma > 1$ \\
    \hline
    L  & Direct ($N = 250$, $\sigma^{-2} = 16$)  & Direct ($N = 250$, $\sigma^{-2} = 16$) \\
    GL1  & GMRES ($N = 100$, $\sigma^{-2} = 9$, $\text{tol} = 10^{-3}$)  & Direct ($N = 40$, $\sigma^{-2} = 3$) \\
    GL2  & GMRES ($N = 40$, $\sigma^{-2} = 6$, $\text{tol} = 10^{-3}$)  & Direct ($N = 16$, $\sigma^{-2} = 2$) \\
  \end{tabular}
  \caption{%
    Numerical parameters used to generate the data presented in \cref{fig:numerics:diffusion_coefficient}.
    We employed the \emph{SciPy} function \texttt{scipy.sparse.linalg.spsolve} for the direct method,
    and the function \texttt{scipy.sparse.linalg.gmres} for GMRES.
  }
  \label{tab:gle:numerical_method_for_computation_diffusion}
\end{table}
\begin{figure}[ht]
  \centering
  \includegraphics[width=.49\linewidth]{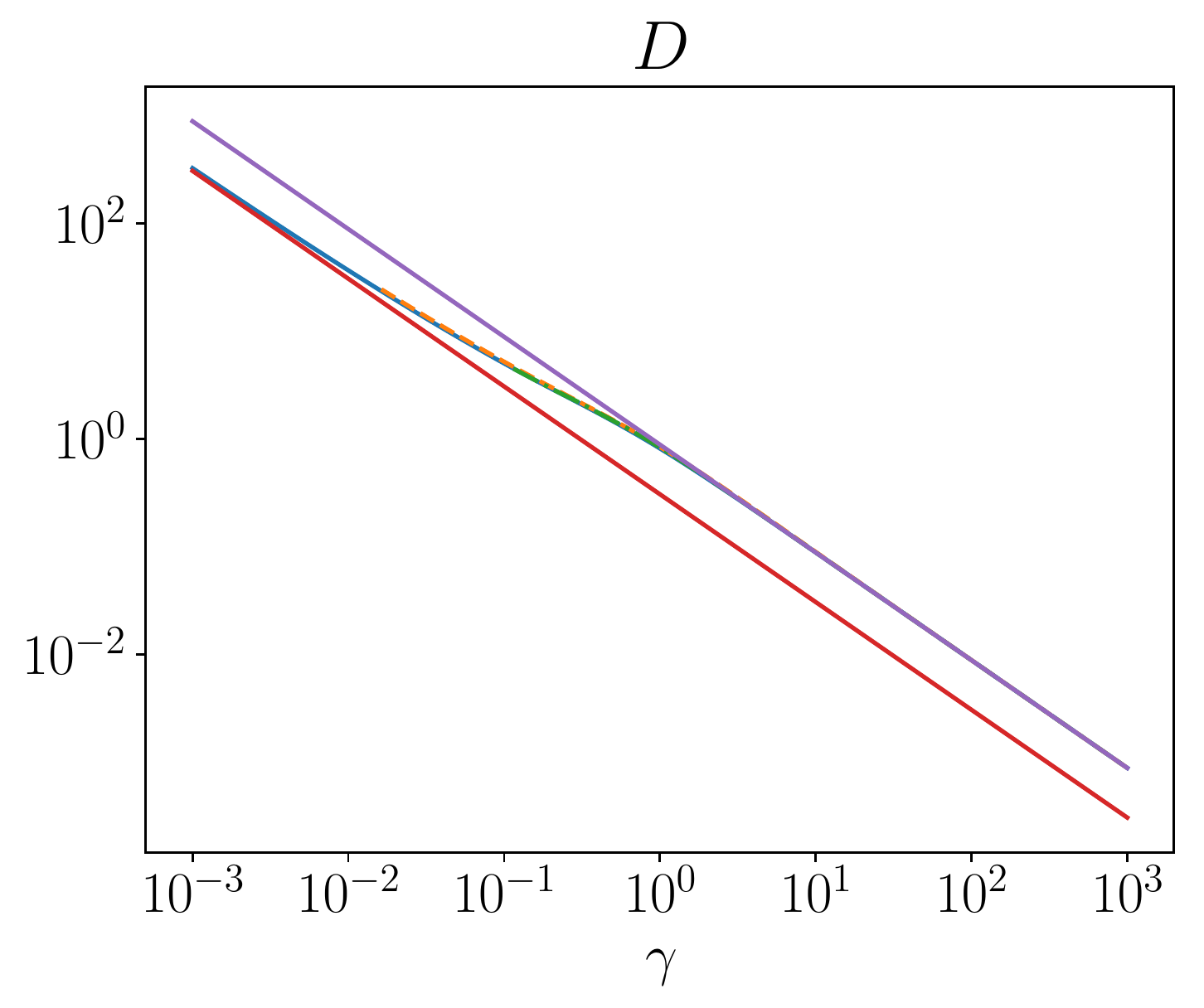}
  \includegraphics[width=.48\linewidth]{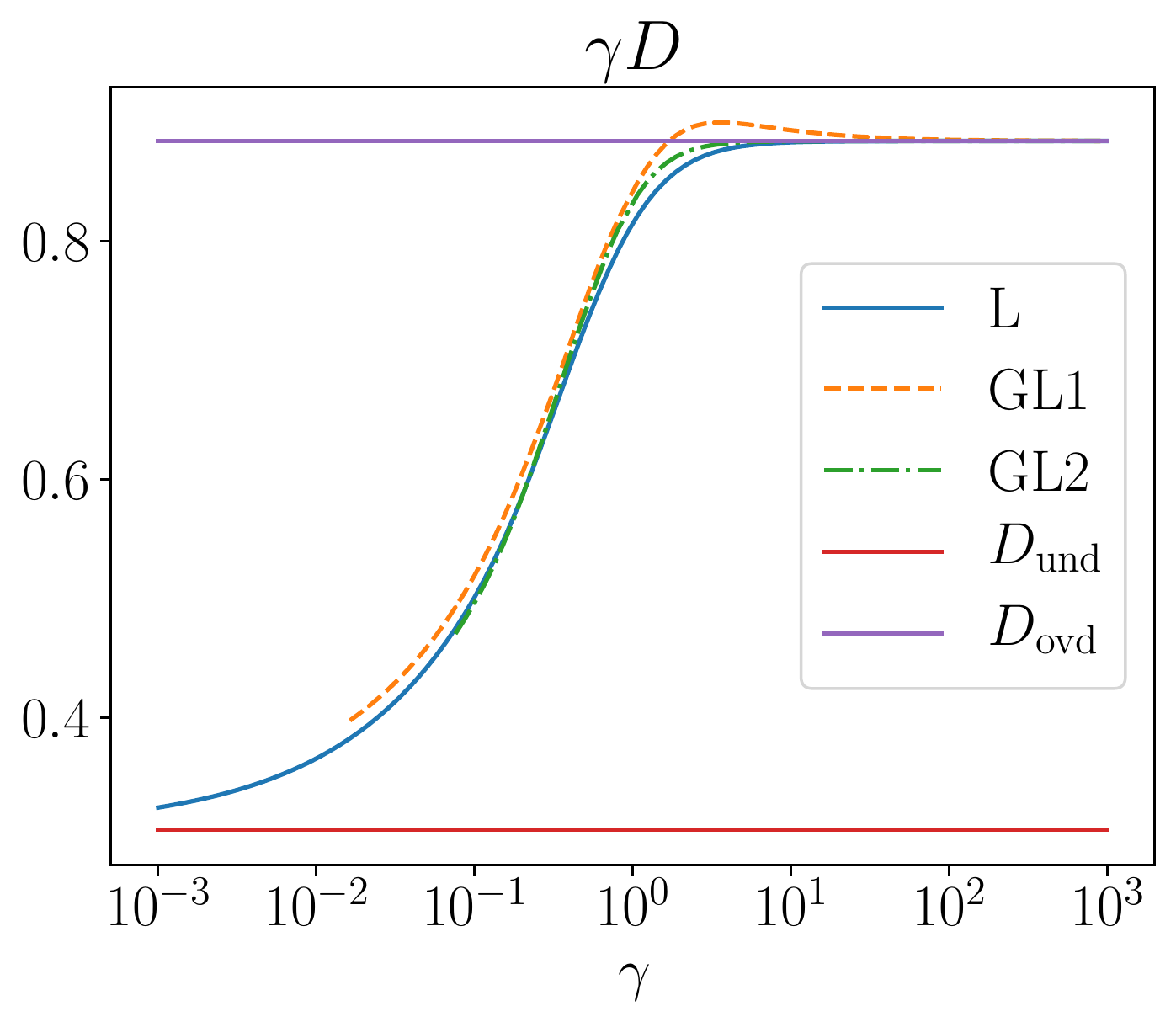}
  \caption{%
    Diffusion coefficient as a function of $\gamma$,
    for the parameters $\nu = 1$ (for the models GL1 and GL2)
    and $\alpha = 1$ (for GL2).
    We observe that,
    for values of $\gamma$ in the range $[1, 10]$,
    the GL1 diffusion coefficient is slightly larger than $D_{\rm ovd}/\gamma$.
  }%
  \label{fig:numerics:diffusion_coefficient}
\end{figure}

We observe from \cref{fig:numerics:diffusion_coefficient} that
the effective diffusion coefficient is of the same order of magnitude
for the three models across the whole range of $\gamma$,
and that $\gamma D \to D_{\rm ovd}$ for all models in the limit as $\gamma \to \infty$.
We also notice that the inequality $\gamma D \leq D_{\rm ovd}$,
which was proved to hold for the Langevin dynamics in~\cite{MR2394704},
is not satisfied for all values of $\gamma$ in the case of the GLE;
indeed it is clear from the figure that, for $\gamma$ close to 2,
the effective diffusion coefficient for the model GL1 is strictly greater than $D_{\rm ovd}/\gamma$.

To conclude this section,
we verify numerically that $D_{\gamma, \nu} \to D_{\gamma}$ in the limit as $\nu \to 0$ for fixed $\gamma$.
\Cref{fig:numerics:convergence_nu} presents the dependence on $\nu$ of the diffusion coefficient
for the models GL1 and GL2 for various values of~$\alpha$.
As expected,
we recover the effective diffusion coefficient corresponding to the model GL1 as $\alpha \to \infty$,
and that of the Langevin dynamics as $\nu \to 0$.
For $\alpha = 1$, the convergence to the limit as $\nu \to 0$ appears to be faster than for the other values of~$\alpha$.
In fact, it is possible to show that the deviation from the limiting effective diffusion coefficient is of order~$\nu^4$ in this case;
see~\cite{thesis_urbain}.
The convergence is illustrated in \cref{fig:numerics:convergence_nu_rate} in a log-log scale,
which confirms the expected rates.
When $\alpha = 1$, round off errors appear for small $\nu$,
which explains the deviation from the theoretical scaling~$\nu^4$.
\begin{figure}[ht]
  \centering
  \includegraphics[width=0.8\linewidth]{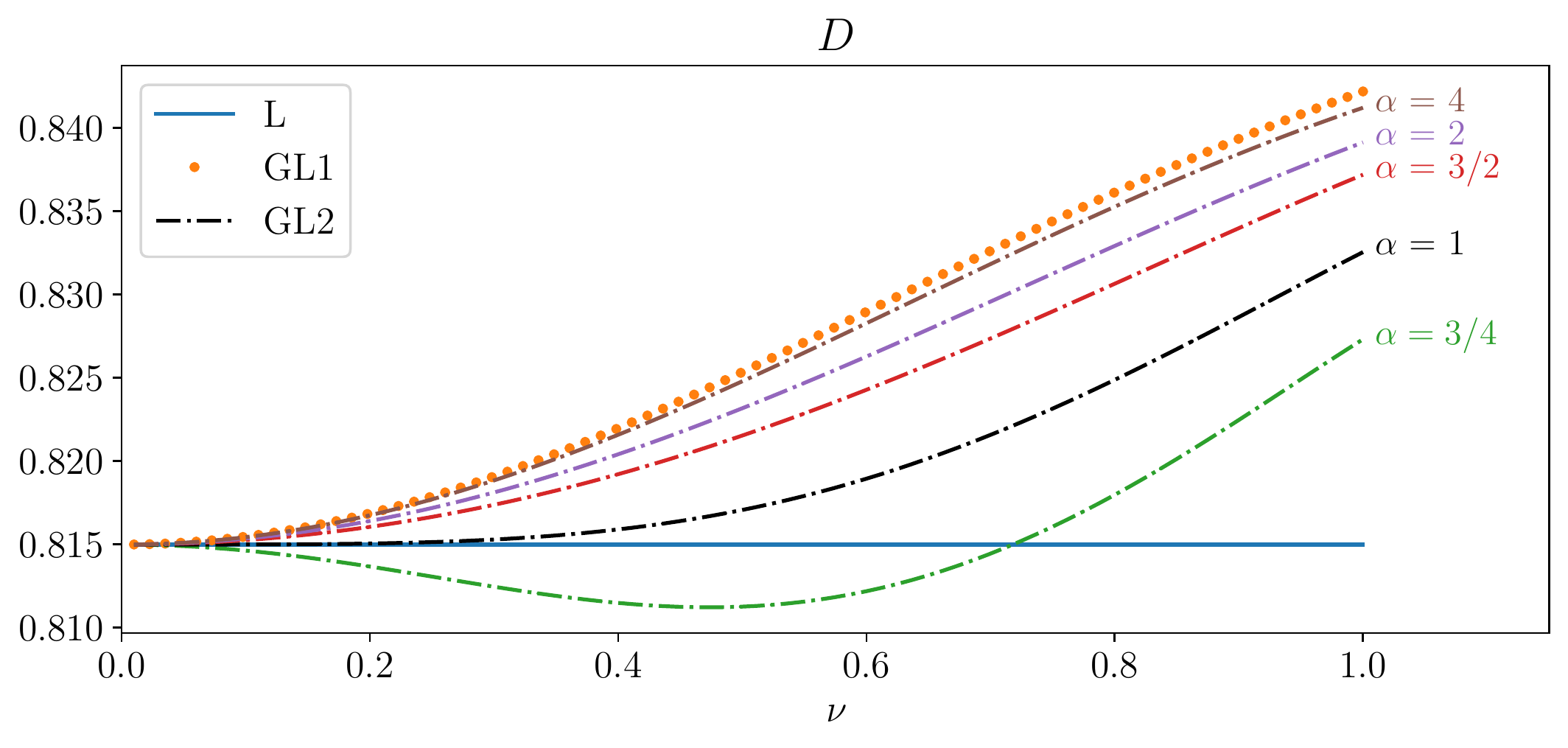}
  \caption{%
    Effective diffusion coefficient against $\nu$,
    whose square encodes the characteristic time of the autocorrelation function of the noise,
    for fixed values $\beta = \gamma = 1$.
  }
  \label{fig:numerics:convergence_nu}
\end{figure}

\begin{figure}[ht]
  \centering
  \includegraphics[width=0.8\linewidth]{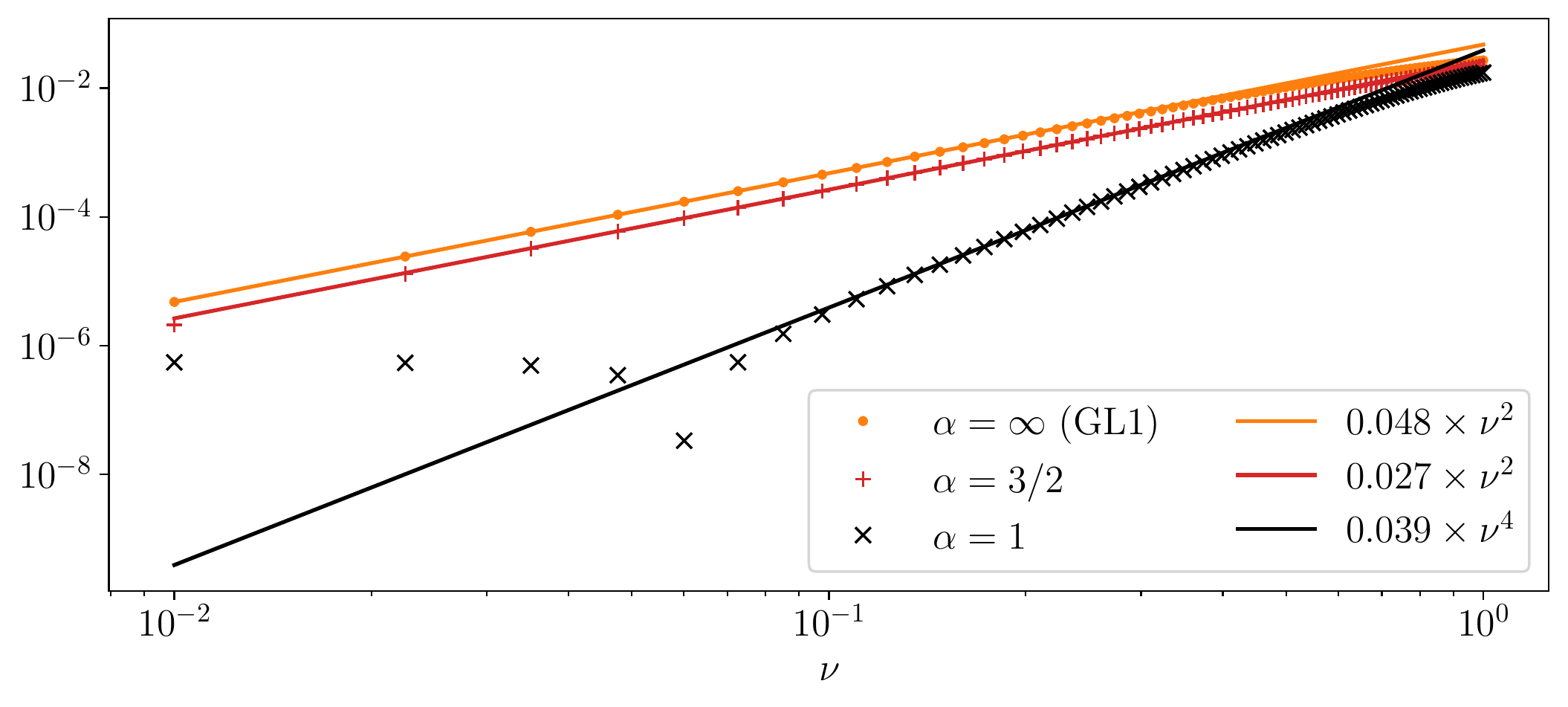}
  \caption{%
    Deviation of the effective diffusion coefficient from its limiting value as $\nu \to 0$.
    For $\alpha = 1$, we observe that the data is not aligned with the straight line for small values of $\nu$,
    which we attribute to round off errors.
  }
  \label{fig:numerics:convergence_nu_rate}
\end{figure}

\section{Longtime behavior for model GL1}%
\label{sec:convergence_of_the_gle_dynamics}

There are many results on the longtime convergence of the evolution semigroup~$\e^{t \mathcal{L}}$ of Langevin-like operators, as reviewed for instance in~\cite{BFLS20} (see also the recent review~\cite{Herau16}). Among the approaches allowing us to quantify the scaling of the convergence rate as a function of the parameters of the dynamics, one can quote:
\begin{itemize}
\item $H^1(\mu)$ hypocoercivity, pioneered in~\cite{MR1924934} and~\cite{MN06}, was later abstracted in~\cite{MR2562709}. The application of this theory to Langevin dynamics allows us to quantify the convergence rates in terms of the parameters of the dynamics; see for instance~\cite{MR2394704} for the Hamiltonian limit and~\cite{LMS16,MR3509213} for the overdamped limit. Moreover, the exponential convergence can be transferred to $L^2(\mu)$ by hypoelliptic regularization~\cite{Herau07}.
\item Entropic estimates, starting with~\cite{DV01}, have been abstracted in~\cite{MR2562709}, under conditions stronger than the ones for $H^1(\mu)$ hypocoercivit\revision{y.}
  Recently, it was shown how to remove the assumption that the Hessian of the potential is bounded~\cite{CGMZ19}.
\item A more direct route to prove the convergence in $L^2(\mu)$ was first proposed in~\cite{Herau06}, then extended in~\cite{MR2576899,MR3324910}, and revisited in~\cite{GS16} where domain issues of the operators at play are addressed. It is based on a modification of the $L^2(\mu)$ scalar product with some regularization operator. This more direct approach makes it even easier to quantify convergence rates; see~\cite{DKMS13,GS16,roussel2018spectral} for studies on the dependence of parameters such as the friction coefficient in Langevin dynamics, as well as~\cite{AAS15} for sharp estimates for equilibrium Langevin dynamics and a harmonic potential energy function.
\item Fully probabilistic techniques, based on clever coupling strategies, can also be used to obtain the exponential convergence of the law of Langevin processes to their stationary state~\cite{EGZ19}. One interest of this approach is that the drift needs not be gradient, in contrast to standard analytical approaches for which the analytical expression of the invariant measure should be known in order to separate the symmetric and antisymmetric parts of the generator under consideration.
\item Finally, it was recently shown how to directly obtain $L^2(\mu)$ estimates without changing the scalar product, relying on a space-time Poincar\'e inequality to conclude to an exponential convergence in time of the evolution semigroup~\cite{AM19,CLW19}.
\end{itemize}

Our focus in this work is on functional analytic estimates, in~$L^2(\mu)$ (where~$\mu$ defined in~\eqref{eq:mu} is the invariant measure of the dynamics), which is a natural framework for giving a meaning to quantities such as effective diffusion coefficients (which have the same form as asymptotic variances in central limit theorems for time averages). We were not able to work directly in~$L^2(\mu)$ by generalizing the approach from~\cite{MR2576899,MR3324910}, because of the hierarchical structure of the dynamics, where the noise in~$\vect z$ is first transferred to~$p$ and then to~$q$.
It is not so easy to construct a modified $L^2(\mu)$ scalar product in this case.
On the other hand, the $H^1(\mu)$ framework of~\cite{MR2562709} can be used directly,
as already done in~\cite{MR2793823}.
Our contribution, compared to the latter work,
is to carefully track the dependence of the convergence rate on the parameters of the dynamics.

In the calculations below we consider the model GL1 for simplicity,
but similar calculations can be carried out for other quasi-Markovian models.
Our results apply to both the periodic and confining settings.
Throughout this section, all operators are considered by default on the functional space~$L^2(\mu)$,
the adjoint of a closed unbounded operator~$T$ on this space being denoted by~$T^*$.

\revision{
We will prove the main results in an order different from that in which they are stated in \cref{sub:main_results}.
This is both more usual and more natural because, as will become clear later,
it is simpler to construct an inner product $\iip{\dummy}{\dummy}_{\gamma, \nu}$ such that~\eqref{eq:hypocoercivity_h1_auxiliary} holds than to construct one such that~\eqref{eq:hypoelliptic_regularization} holds.
}

\begin{remark}
  The approach taken in this section can be applied in particular to the model GL2,
  as discussed in Appendix~\ref{sec:longtime_behavior_for_model_gl2}.
  The computations are however algebraically more cumbersome,
  \revision{so} the scalings we obtain for the resolvent bound appear not to be sharp, at least in the limit $\gamma \to \infty$.
\end{remark}

\subsection{Proof of \texorpdfstring{\cref{thm:hypocoercivity_h1}}{Theorem 2.1}: decay in \texorpdfstring{$H^1(\mu)$}{H1}}%

We first introduce the adjoint operator $\partial_{z}^* = \beta z - \derivative{1}{z}$ and
rewrite the generator of the dynamics for the model GL1 in the standard form of the~$H^1(\mu)$ coercivity framework~\cite{MR2562709}:
\begin{align}
  \label{eq:hypocoercivity:decomposition_generator}
  \textstyle
  -\mathcal L = A^*A + B, \qquad A = \nu^{-1}\beta^{-1/2} \partial_{z},
  \qquad B = \sqrt{\gamma} \, \nu^{-1} ( p \, \derivative{1}{z} - z \, \derivative{1}{p} )
  + ( \derivative*{1}[V]{q}(q) \, \derivative{1}{p} - p \, \derivative{1}{q} ).
\end{align}
The relevant operators for the study of hypocoercivity are obtained from (iterated) commutators of~$B$ with~$A$:
\begin{subequations}
  \begin{align}
&\textstyle C_0 := \nu \, \beta^{1/2} \, A = \derivative{1}{z}, \nonumber \\
\label{eq:convergence:commutators_1}
&C_1 := \nu \, \gamma^{-1/2} \, \commut{C_0}{B} = \textstyle - \derivative{1}{p}, \\
\label{eq:convergence:commutators_2}
&C_2 := \commut{C_1}{B} + \nu^{-1} \, \gamma^{1/2}   \, C_0 = \textstyle \derivative{1}{q}, \\
\label{eq:convergence:commutators_3}
&\commut{C_2}{B} + \derivative*{2}[V]{q^2}(q) \, C_1 =  0.
  \end{align}
\end{subequations}
If we were interested only in showing that $-\mathcal L$ is hypocoercive,
it would be sufficient to invoke at this stage~\cite[Theorem 24]{MR2562709}, as done in~\cite{MR2793823}.
Here, however, we are interested not only in whether the dynamics converge to equilibrium
but also in the scaling of the rate of convergence with respect to $\nu$ and $\gamma$,
so a careful analysis is required.
Recalling that $\ip{\cdot}{\cdot}$ and $\norm{\cdot}$ denote respectively the standard scalar product and norm of $\lp{2}{\mu}$,
we denote by $\iip{\cdot}{\cdot}$ the inner product defined by polarization from the norm constructed with the operators~$C_0,C_1,C_2$ defined above:
\begin{align}
  \label{eq:hypocoercivity:norm}
  \textstyle
  \iip{h}{h}
  & = \textstyle \norm{h}^2
    + a_0 \norm{\derivative{1}[h]{z}}^2
    + a_1 \norm{\derivative{1}[h]{p}}^2
    + a_2 \norm{\derivative{1}[h]{q}}^2
  - 2 b_0 \ip{\derivative{1}[h]{z}}{\derivative{1}[h]{p}}
  - 2 b_1 \ip{\derivative{1}[h]{p}}{\derivative{1}[h]{q}},
\end{align}
that is
\begin{align}
  \textstyle
  \iip{h_1}{h_2}
  & =
  \textstyle \ip{h_1}{h_2}
  + a_0 \ip{\derivative{1}[h_1]{z}}{\derivative{1}[h_2]{z}}
  + a_1 \ip{\derivative{1}[h_1]{p}}{\derivative{1}[h_2]{p}}
  + a_2 \ip{\derivative{1}[h_1]{q}}{\derivative{1}[h_2]{q}}
  \nonumber \\
   & \ \ \ \textstyle
   - b_0  \ip{\derivative{1}[h_1]{z}}{\derivative{1}[h_2]{p}}
   - b_0 \ip{\derivative{1}[h_2]{z}}{\derivative{1}[h_1]{p}}
   - b_1 \ip{\derivative{1}[h_1]{p}}{\derivative{1}[h_2]{q}}
   - b_1 \ip{\derivative{1}[h_2]{p}}{\derivative{1}[h_1]{q}}.
   \label{eq:hypocoercivity:inner_product}
\end{align}
In \revision{these} expressions, the coefficients $a_0,a_1,a_2$ are positive, while $b_0,b_1$ are nonnegative (this sign convention is motivated by the computations performed in this section).

To prove hypocoercivity for the norm of $\sobolev{1}{\mu}$,
we must show that it is possible to find coefficients
$a_0, a_1, a_2>0$ and $b_0, b_1 \geq 0$ such that:
\begin{enumerate}[(i)]
\item the standard weighted Sobolev norm $\|\cdot\|_{H^1(\mu)}$ (which corresponds to~$a_0=a_1=a_2=1$ and~$b_0=b_1=0$ in~\eqref{eq:hypocoercivity:norm}) is equivalent to the norm~\eqref{eq:hypocoercivity:norm}. Let us however mention that degenerate norms not equivalent to~\eqref{eq:hypocoercivity:norm} can be considered, as initially done in~\cite{MR1924934}, and recently used in~\cite{Baudoin17,OL15,MR3925138};
\item coercivity holds for this modified norm, \emph{i.e.}~there exists $\lambda > 0$ such that $- \textstyle \iip{h}{\mathcal L h} \geq \lambda \, \iip{h}{h}$
  \revision{for all smooth, compactly supported and mean-zero $h \in C^{\infty}_c \cap L^2_0(\mu)$.}
\end{enumerate}
\revision{
  In order to be complete,
  we should in principle show that it is indeed sufficient to prove the coercivity inequality~(ii) only for $h \in C^{\infty}_c \cap L^2_0(\mu)$.
  As discussed in the proof of~\cite[Theorem A.8]{MR2562709} in a slightly different context,
  this requires an approximation argument,
  which we will omit here because this argument is both standard and \revision{somewhat} technical.
  For the same reason, we will not worry about the technical justification of the calculations in \cref{sub:hypoelliptic_regularization}.
}

By the Cauchy–Schwarz inequality and since $a_0,a_1,a_2,b_0,b_1 \geq 0$, it is clear that
\begin{equation}
  \label{eq:hypocoercivity:equivalence_with_sobolev_norm}
  \iip{h}{h} - \norm{h}^2 \geq
  \begin{pmatrix}
    \norm{\derivative{1}[h]{z}} \\
    \norm{\derivative{1}[h]{p}} \\
    \norm{\derivative{1}[h]{q}}
  \end{pmatrix}^\t
  \mat M_1
  \begin{pmatrix}
    \norm{\derivative{1}[h]{z}} \\
    \norm{\derivative{1}[h]{p}} \\
    \norm{\derivative{1}[h]{q}}
  \end{pmatrix},
  \qquad
  \mat M_1 =
    \begin{pmatrix}
      a_0 & -b_0 & 0 \\
      -b_0 & a_1 & -b_1 \\
      0 & - b_1 & a_2 \\
    \end{pmatrix},
\end{equation}
so it is sufficient that $\mat M_1$ be positive definite in order to meet the first condition.
For the second condition, we rely on the following auxiliary result.

\begin{lemma}
  \label{lem:coercivity_condition}
  Suppose that \cref{assumption:assumption_potential} holds.
  Then
  \label{lemma:auxiliary_result_hypocoercivity}
  \begin{equation}
    \label{eq:coercivity_auxiliary_norm}
    - \iip{h}{\mathcal L h} \geq
\frac{1}{\nu^2 \beta}
    \begin{pmatrix}
      \norm{C_0 \, C_0 \, h} \\
      \norm{C_0 \, C_1 \, h} \\
      \norm{C_0 \, C_2 \, h}
    \end{pmatrix}^\t
    \mat M_1
    \begin{pmatrix}
      \norm{C_0 \, C_0 \, h} \\
      \norm{C_0 \, C_1 \, h} \\
      \norm{C_0 \, C_2 \, h}
    \end{pmatrix}
    +
    \begin{pmatrix}
      \norm{C_0 \, h} \\
      \norm{C_1 \, h} \\
      \norm{C_2 \, h}
    \end{pmatrix}^\t
    \widetilde{\mat M}_2
    \begin{pmatrix}
      \norm{C_0 \, h} \\
      \norm{C_1 \, h} \\
      \norm{C_2 \, h}
    \end{pmatrix},
  \end{equation}
  where $\mat M_1$ is the same matrix as in~\cref{eq:hypocoercivity:equivalence_with_sobolev_norm}
  and $\widetilde{\mat M}_2$ is given by
  \begin{equation}
    \label{eq:definition_tilde_M2}
    \widetilde{\mat M}_2 =
    \begin{pmatrix}
      \dps \frac{1}{\beta \nu^2} + \frac{a_0}{\nu^{2}} - \frac{b_0 \sqrt{\gamma}}{\nu}
    & \dps - \left|{(a_0 - a_1)\frac{\sqrt{\gamma}}{\nu} + \frac{b_0}{\nu^{2}}}\right|
    & \dps - \left|{b_0 - \frac{b_1 \sqrt{\gamma}}{\nu}}\right|\\[10pt]
      0 & \dps \frac{b_0 \sqrt{\gamma}}{\nu} - b_1 \norm{\derivative*{2}[V]{x^2}}_{\infty} & - a_1 - a_2 \, \norm{\derivative*{2}[V]{x^2}}_{\infty} \\
      0 & 0 & b_1
    \end{pmatrix}.
  \end{equation}
\end{lemma}

It would be desirable to relax the condition $\norm*{V''}[\infty] < \infty$ in \cref{assumption:assumption_potential}
by following the approach presented in~\cite[Section~7]{MR3509213}, but this is not possible as such here; see Remark~\ref{rmk:V''bounded} for further precisions.

\begin{proof}
We calculate the action of the symmetric part of the generator on
the terms multiplying \revision{$a_0, a_1, a_2$} in~\eqref{eq:hypocoercivity:norm}:
\begin{subequations}
  \label{eq:subeq_coercivity_0}
  \begin{align}
    \label{eq:hypocoercivity:symmetric_on_pure_terms_0}
    & \ip{C_0 h}{C_0 (A^* A) h} = \nu^{-2} \, \beta^{-1} \, \norm{C_0^2 \, h}^2 + \nu^{-2} \, \norm{C_0 h}^2, \\
    \label{eq:hypocoercivity:symmetric_on_pure_terms_1}
    & \ip{C_1 h}{C_1 (A^* A) h} = \nu^{-2}  \, \beta^{-1} \, \norm{C_0 \, C_1 \, h}^2, \\
    \label{eq:hypocoercivity:symmetric_on_pure_terms_2}
    & \ip{C_2 h}{C_2 (A^* A) h} = \nu^{-2}  \, \beta^{-1} \, \norm{C_0 \, C_2 \, h}^2,
  \end{align}
\end{subequations}
where we took into account that $C_0$ commutes with $C_1,C_2$, while $\commut{C_0}{C_0^*} = \beta$.
The action of the antisymmetric part of the generator~$B$ on the the same terms is, in view of the commutator relations~\eqref{eq:convergence:commutators_1}-\eqref{eq:convergence:commutators_3},
\begin{subequations}
  \label{eq:subeq_coercivity_1}
  \begin{align}
    \label{eq:hypocoercivity:skew_on_pure_terms_0}
    & \ip{C_0 h}{C_0 B h} = \ip{C_0 h}{[C_0,B] h} = \gamma^{1/2} \, \nu^{-1} \, \ip{C_0 h}{C_1 h}, \\
    \label{eq:hypocoercivity:skew_on_pure_terms_1}
    & \ip{C_1 h}{C_1 B h} = \ip{C_1 h}{[C_1,B] h} = \ip{C_1 h}{C_2 h} - \gamma^{1/2} \, \nu^{-1} \ip{C_1 h}{C_0 h}, \\
    \label{eq:hypocoercivity:skew_on_pure_terms_2}
    & \ip{C_2 h}{C_2 B h} = \ip{C_2 h}{[C_2,B] h} = - \ip{C_2 h}{\derivative*{2}[V]{q^2}(q) \, C_1 \, h}.
  \end{align}
\end{subequations}
For the terms multiplying $b_0, b_1$ in~\eqref{eq:hypocoercivity:norm}, we have
\begin{subequations}
  \label{eq:subeq_coercivity_2}
  \begin{align}
    \label{eq:hypocoercivity:symmetric_on_mixed_terms_0}
    \ip{C_0 (A^* A)h}{C_1  h} + \ip{C_0 h}{C_1 (A^* A) h}  =& \, \nu^{-2} \,\left( \ip{C_0 h}{C_1 h} + 2 \, \beta^{-1} \, \ip{C_0^2 h}{C_0 C_1 h} \right), \\
    \label{eq:hypocoercivity:symmetric_on_mixed_terms_1}
    \ip{C_1 (A^* A) h}{C_2  h} + \ip{C_1 h}{C_2 (A^* A) h} =& \, 2 \, \nu^{-2} \, \beta^{-1} \, \ip{C_0 C_1 h}{C_0 C_2 h},
  \end{align}
\end{subequations}
and
\begin{subequations}
  \label{eq:subeq_coercivity_3}
  \begin{align}
    \label{eq:hypocoercivity:skew_on_mixed_terms_0}%
    \ip{C_0 B h}{C_1 h} + \ip{C_0 h}{C_1 B h} & = \ip{[C_0,B] h}{C_1 h} + \ip{C_0 h}{[C_1, B] h} \notag \\
                                              & = \gamma^{1/2} \, \nu^{-1} \, \norm{C_1 h}^2
                                                +  \ip{C_0 h}{C_2 h} - \gamma^{1/2} \, \nu^{-1} \, \norm{C_0 h}^2, \\
    \label{eq:hypocoercivity:skew_on_mixed_terms_1}%
    \ip{C_1 B h}{C_2 h} + \ip{C_1 h}{C_2 B h} & = \ip{[C_1,B] h}{C_2 h} + \ip{C_1 h}{[C_2,B] h} \notag \\
                                              & = \norm{C_2 h}^2 - \gamma^{1/2} \, \nu^{-1} \, \ip{C_0 h}{C_2 h}
                                                - \ip{\derivative*{2}[V]{q^2}(q) C_1 h}{C_1 h}.
  \end{align}
\end{subequations}
The inequality~\cref{eq:coercivity_auxiliary_norm} then follows by combining~\cref{eq:subeq_coercivity_0,eq:subeq_coercivity_1,eq:subeq_coercivity_2,eq:subeq_coercivity_3} and using~\cref{assumption:assumption_potential} as well as the Cauchy--Schwarz inequality.
\end{proof}

\begin{remark}
  \label{rmk:V''bounded}
  For underdamped Langevin dynamics, it is possible to relax the condition $\norm*{V''}[\infty] < \infty$ in \cref{assumption:assumption_potential} by following the approach of\revision{~\cite[Section~7]{MR2562709}};
  see also the presentation in the proof of~\cite[Theorem 2.15]{MR3509213}, which relies on an estimate provided by~\cite[Lemma~A.24]{MR2562709}. The latter result states that, if $V \in C^2(\real^n)$ satisfies the inequality
  \begin{equation}
    \label{eq:weaker_condition_hessian}
    \forall q \in \real^n, \qquad \abs{\hess V(q)} \leq c(1 + \abs{\nabla V(q)})
  \end{equation}
  for some constant $c >0$,
  then there exist nonnegative constants $A_V$ and $B_V$ such that
  \[
    \forall h \in \sobolev{1}{\e^{-\beta V}},
    \qquad \norm{h \, \hess V}[L^2(\e^{-\beta V})] \leq A_V \norm{h}[L^2(\e^{-\beta V})] + B_V \norm{\nabla h}[L^2(\e^{-\beta V})].
  \]
  Unfortunately,
  this approach does not enable to replace the condition of bounded Hessian by the weaker condition~\eqref{eq:weaker_condition_hessian} in the case of model GL1.
  In particular, it seems difficult to control the term on the right-hand side of~\eqref{eq:hypocoercivity:skew_on_pure_terms_2}.
  Indeed, quantities such as $\abs{\ip{C_1h}{V''(q) C_2h}}$ would be bounded by factors such as $\|C_2^2 h\|$ or $\|C_1C_2h\|$,
  which cannot be controlled with the first term \revision{on} the right-hand side of~\eqref{eq:coercivity_auxiliary_norm}.
  A similar issue arises with the last term on the \revision{right-hand} side of~\eqref{eq:hypocoercivity:skew_on_mixed_terms_1}.
\end{remark}

\Cref{lem:coercivity_condition} shows that the coercivity of~$-\mathcal L$ for the modified norm is ensured if we can find parameters $a_0, a_1, a_2, b_0, b_1$ such that the matrix~$\widetilde {\mat M}_2$ in~\eqref{eq:definition_tilde_M2} is positive definite. This is made precise in the following result (\Cref{proposition:rate_of_convergence_hypocoercitivy}),
which is a weaker version of \cref{proposition:rate_of_convergence} proved in the next section. In order to state it, we introduce the following notation: for matrices $\mat X,\mat Y \in \real^{d \times d}$,
\[
  \mat X \succcurlyeq_+ \mat Y \qquad \text{if} \qquad \vect v^\t \mat X \vect v \geq \vect v^\t \mat Y \vect v \qquad \forall \vect v \geq 0,
\]
where the notation $\vect v \geq 0$ for $\vect v =(v_1,\dots,v_d) \in \real^d$ means that $v_i \geq 0$ for all $1 \leq i \leq d$. We also define the minimum of the Rayleigh quotient under a positivity constraint:
\begin{equation}
  \label{eq:positive_rayleigh}%
  \lambda_{\min}^+(\mat X) := \min_{\vect v \neq 0, \vect v \geq 0} \, \frac{\vect v^\t {\mat X} \vect v}{\vect v^\t \vect v}.
\end{equation}
Note that $\mat X \succcurlyeq_+ \mat Y$ implies that $\lambda_{\min}^+(\mat X) \geq \lambda_{\min}^+(\mat Y)$.

\begin{remark}
The inequality $\mat X \succcurlyeq \mat Y$ for two symmetric matrices implies $\mat X \succcurlyeq_+ \mat Y$, but not conversely:
consider e.g.\ the 2-by-2 matrices with entries $\delta_{ij}$ and $(-1)^{i + j}$ for $1 \leq i,j \leq 2$. Remark also that,
for a matrix~$\mat X$ with nonpositive off-diagonal entries (such as~$\widetilde{\mat M}_2$),
it is equivalent to define $\lambda_{\min}^+\left(\mat X\right)$ as the smallest eigenvalue of the symmetrized matrix $(\mat X + \mat X^\t)/2$, since the minimum of $\vect v^\t \mat X \vect v$ on the sphere $|\vect v| = 1$ is achieved for some $\vect v$ with nonnegative elements.
\end{remark}

We easily deduce from~\eqref{eq:definition_tilde_M2} that
\begin{equation}
  \label{eq:simplified_matrix}%
  \widetilde{\mat M}_2 \succcurlyeq_+
  \begin{pmatrix}
    (\beta^{-1} + a_0) \nu^{-2}- b_0 r
    & - a_0 r - a_1 r - b_0 \nu^{-2}
    & - b_0 - b_1 r\\
    0 & b_0 r - b_1 \norm{\derivative*{2}[V]{x^2}}_{\infty} & - a_1 - a_2 \, \norm{\derivative*{2}[V]{x^2}}_{\infty} \\
    0 & 0 & b_1
  \end{pmatrix},
\end{equation}
where
\[
  r := \gamma^{1/2} \, \nu^{-1}.
\]
It is therefore sufficient to work with the matrix on the right-hand side of~\eqref{eq:simplified_matrix} to derive a lower bound for $\lambda_{\min}^+\left(\widetilde{\mat M}_2\right)$, and thus for the rate of convergence.

\begin{proposition}
  \label{proposition:rate_of_convergence_hypocoercitivy}%
  There exist parameters $(a_0, a_1, a_2, b_0, b_1)$, as well as a constant $C>0$ (independent of $\gamma,\nu$) and $\alpha(\gamma,\nu)>0$, such that $\alpha(\gamma,\nu) \mat I_3 \preccurlyeq \mat M_1 \preccurlyeq \mat I_3$ and
  \[
    \widetilde{\mat M}_2 \succcurlyeq_+ C \min \left(r^2 \nu^2, \frac{1}{r^2 \nu^2}, \frac{r^2}{\nu^2}\right)\mat I_3 = C \min \left(\gamma, \frac{1}{\gamma}, \frac{\gamma}{\nu^4}\right) \mat I_3.
  \]
\end{proposition}

We obtain from the previous proposition and~\eqref{eq:coercivity_auxiliary_norm} that
\[
  -\iip{h}{\mathcal L h}
  \geq \lambda_{\min}^+\left(\widetilde{\mat M}_2\right) \sum_{i=0}^{2} \norm{C_i h}^2.
\]
We rely at this stage on Poincar\'e's inequality to control~$\norm{h}^2$ with the right-hand side \revision{of} the above inequality.
Since $\mu$ is a product of probability measures that satisfy \revision{a} Poincar\'e's inequality (the marginals in~$p,z$ being Gaussian distributions of variance~$\beta^{-1}$),
it itself satisfies \revision{a Poincar\'e} inequality, see for instance~\cite[Proposition 2.6]{MR3509213}:
\[
\forall h \in H^1_0(\mu), \qquad  \norm{h}^2 \leq \frac{1}{\min (R_{\beta}, \beta)}\sum_{i=0}^{2} \norm{C_i h}^2.
\]
where $R_{\beta}$ is the Poincar\'e constant from~\cref{assumption:assumption_potential}. Denoting by $K_1$ is the largest eigenvalue \revision{of} $\mat M_1$ and by $\kappa = \max(R_{\beta}^{-1}, \beta^{-1})$, this implies that, for any $\zeta \in (0, 1)$ and any $h \in H^1_0(\mu)$,
\[
-\iip{h}{\mathcal L h} \geq \lambda_{\min}^+\left(\widetilde{\mat M}_2\right) \left( \zeta \kappa \norm{h}^2 + (1 - \zeta) \sum_{i=0}^{2} \norm{C_i h}^2\right) \geq \lambda_{\min}^+\left(\widetilde{\mat M}_2\right) \min \left(\zeta \kappa, \frac{1-\zeta}{K_1} \right) \iip{h}{h}.
\]
The optimal choice for~$\zeta$ is $\zeta = (1+K_1 \kappa)^{-1}$, which leads to the following inequality for $h \in H^1_0(\mu)$ (noting that $\e^{t \mathcal{L}}h \in H^1_0(\mu)$ for all $t \geq 0$ \revision{by Theorem~\ref{thm:hypoelliptic_regularization}}):
\[
\frac{1}{2} \, \derivative*{1}{t} \iip{\e^{t \mathcal{L}}h}{\e^{t \mathcal{L}}h} \revision{= \iip{\mathcal{L}\e^{t \mathcal{L}}h}{\e^{t \mathcal{L}}h}} \leq - \frac{\lambda_{\min}^+\left(\widetilde{\mat M}_2\right)}{K_1+\kappa^{-1}} \iip{\e^{t \mathcal{L}}h}{\e^{t \mathcal{L}}h}.
\]
\Cref{thm:hypocoercivity_h1} then follows from Gronwall's inequality.

\subsection{Proof of \texorpdfstring{\cref{thm:hypoelliptic_regularization}}{Theorem 2.2}: hypoelliptic regularization}%
\label{sub:hypoelliptic_regularization}

In this section, we prove the hypoelliptic regularization estimate~\eqref{eq:hypoelliptic_regularization}.
% for an initial condition~$h \in H^1_0(\mu)$.
% By density of $H^1_0(\mu)$ in $L^2_0(\mu)$, the estimate can then be extended to any initial condition in $L^2_0(\mu)$.
For any $h\in \revision{L^2_0(\mu)}$, we define, analogously to~\cite{Herau07,MR2394704,MR2793823},
\begin{align}
  \textstyle
  N_{h}(t)
  =
  \norm{\e^{t \mathcal{L}}h}^2
  & \textstyle + a_0 t \norm{ \derivative{1}[\e^{t \mathcal{L}}h]{z} }^2
  + a_1 t^3 \norm{\derivative{1}[\e^{t \mathcal{L}}h]{p}}^2
  + a_2 t^5 \norm{\derivative{1}[\e^{t \mathcal{L}}h]{q}}^2
  \nonumber \\
   & \textstyle
   - 2 b_0 t^2 \ip{\derivative{1}[\e^{t \mathcal{L}}h]{z}}{\derivative{1}[\e^{t \mathcal{L}}h]{p}}
   - 2 b_1 t^4 \ip{\derivative{1}[\e^{t \mathcal{L}}h]{p}}{\derivative{1}[\e^{t \mathcal{L}}h]{q}},
   \label{eq:regularization:norm}
\end{align}
where $a_0, a_1, a_2>0$ and $b_0, b_1 \geq 0$ are small parameters.
We calculate, by computations similar to the ones performed in the proof of \Cref{lem:coercivity_condition},
\begin{equation}
\label{eq:hypoelliptic_regularization_inequality}
  \frac{1}{2} \, \derivative*{1}[N_h]{t} (t) \leq
  - \frac{1}{\beta \nu^2}
  \begin{pmatrix}
    \norm{C_0 \, C_0 \, h} \\
    \norm{C_0 \, C_1 \, h} \\
    \norm{C_0 \, C_2 \, h}
  \end{pmatrix}^\t
  \widetilde{\mat M}_1(t)
  \begin{pmatrix}
    \norm{C_0 \, C_0 \, h} \\
    \norm{C_0 \, C_1 \, h} \\
    \norm{C_0 \, C_2 \, h}
  \end{pmatrix}
  -
  \begin{pmatrix}
    \norm{C_0 \, h} \\
    \norm{C_1 \, h} \\
    \norm{C_2 \, h}
  \end{pmatrix}^\t
  \mat M_2(t)
  \begin{pmatrix}
    \norm{C_0 \, h} \\
    \norm{C_1 \, h} \\
    \norm{C_2 \, h}
  \end{pmatrix},
\end{equation}
where
% Dominant terms
\newcommand{\dom}[1]{#1}
\[
  \widetilde{\mat M}_1(t) =
  \begin{pmatrix}
    a_0 \, t & -b_0 \, t^2 & 0 \\
    - b_0 \, t^2 & a_1 \, t^3 & -b_1 \, t^4 \\
    0 & - b_1 \, t^4 & a_2 \, t^5
  \end{pmatrix},
\]
and,
employing again the notation $r = \gamma^{1/2} \, \nu^{-1}$,
\begin{align}
    \mat M_2(t)
    & =
    \notag%
    \scriptscriptstyle{%
    \begin{pmatrix}
      \left(\beta^{-1}+a_0 t\right) \nu^{-2} - b_0r t^2
      &  - a_0 r t - a_1 r t^3 - b_0 \nu^{-2} t^2
      & - b_0 t^2 - b_1 r t^4\\
      0 &  b_0 r t^2 - b_1 t^4 \norm{\derivative*{2}[V]{x^2}}_{\infty}
        &  - a_1 t^3 - a_2 t^5 \norm{\derivative*{2}[V]{x^2}}_{\infty} \\
      0 & 0 &  b_1 t^4
    \end{pmatrix}} \\
    & \ \ +
    \label{eq:decomposition_M2_t}
    \begin{pmatrix}
      - \dps \frac{a_0}{2}
      & - 2 b_0 t
      & 0\\
      0 & \dps - \frac{3a_1}{2} t^2
        & - 4 b_1  t^3 \\
      0 & 0 & \dps - \frac{5a_2}{2} t^4
    \end{pmatrix}%
\end{align}
Note that $\widetilde{\mat M}_1(1) = \mat M_1$ and, for $t=1$ also,
the first matrix on the right-hand side of~\eqref{eq:decomposition_M2_t} coincides with the matrix on the right-hand side of~\eqref{eq:simplified_matrix}.
The matrix $\widetilde{\mat M}_1(t)$ is clearly positive semidefinite for any $t \in [0, 1]$ if $\mat M_1$ is positive semidefinite,
which can be viewed from the factorization
\[
  \widetilde{\mat M}_1(t) =
  \begin{pmatrix}
    t^{1/2} & 0 & 0 \\
    0 & t^{3/2} & 0 \\
    0 & 0 & t^{5/2}
  \end{pmatrix}
  \mat M_1
  \begin{pmatrix}
    t^{1/2} & 0 & 0 \\
    0 & t^{3/2} & 0 \\
    0 & 0 & t^{5/2}
  \end{pmatrix}.
\]
We also notice that, for any~$t \in [0, 1]$,
\[
  \begin{pmatrix}
    1 & 0 & 0 \\
    0 & t^{-1} & 0 \\
    0 & 0 & t^{-2}
  \end{pmatrix}
  \mat M_2(t)
  \begin{pmatrix}
    1 & 0 & 0 \\
    0 & t^{-1} & 0 \\
    0 & 0 & t^{-2}
  \end{pmatrix}
  \succcurlyeq_+ \mat M_2,
\]
with
\begin{equation}
  \label{eq:matrix_regularization}
  \mat M_2 =
  \begin{pmatrix}
      \dps \frac{1}{\beta\nu^2} - b_0r-\frac{a_0}{2}
      &  - a_0 r - a_1 r - \left(2+\nu^{-2}\right)b_0
      & - b_0 - b_1 r\\
      0 &  \dps b_0 r - \frac{3a_1}{2} - b_1 \norm{\derivative*{2}[V]{x^2}}_{\infty}
        &  -4b_1 - a_1 - a_2 \norm{\derivative*{2}[V]{x^2}}_{\infty} \\
      0 & 0 & \dps b_1 - \frac{5a_2}{2}
    \end{pmatrix}.
\end{equation}
The following key result shows that,
for an appropriate choice of the parameters $(a_0, a_1, a_2, b_0, b_1)$,
the matrix $\mat M_2$ is bounded from below,
in the sense of $\succcurlyeq_+$, by the identity matrix multiplied by a positive prefactor.
% Consequently, since $\widetilde{\mat M}_1(t)$ and $\mat M_2(t)$ are positive semidefinite for all $t \in [0, 1]$,
% we deduce from~\eqref{eq:hypoelliptic_regularization_inequality} that $N_h(1) \leq N_h(0)$,
% which is precisely~\eqref{eq:hypoelliptic_regularization}.

\begin{proposition}
  \label{proposition:rate_of_convergence}%
  There exist parameters $(a_0, a_1, a_2, b_0, b_1)$,
  as well as a constant $C > 0 $ (independent of $\gamma$ and~$\nu$) and $\alpha(\gamma,\nu)>0$,
  such that $\alpha(\gamma,\nu) \mat I_3 \preccurlyeq \mat M_1 \preccurlyeq \mat I_3$ and
  \[
    \mat M_2 \succcurlyeq_+ C \min \left(r^2 \nu^2, \frac{1}{r^2 \nu^2}, \frac{r^2}{\nu^2}\right) \mat I_3 = C \min \left(\gamma, \frac{1}{\gamma}, \frac{\gamma}{\nu^4}\right) \mat I_3.
  \]
\end{proposition}

Observe that $\widetilde{\mat M}_2 \succcurlyeq_+ \mat M_2$, where $\widetilde{\mat M}_2$ is the matrix defined in~\eqref{eq:definition_tilde_M2}, so the lower bound on~$\mat M_2$ in \Cref{proposition:rate_of_convergence} implies~\cref{proposition:rate_of_convergence_hypocoercitivy} as a byproduct.
\revision{\Cref{proposition:rate_of_convergence} also implies that $ N'_h(t) \leq 0$ for any~$t \in [0,1]$, by~\eqref{eq:hypoelliptic_regularization_inequality}.
This leads to the inequality
\[
    N_h(1) \leq N_h(0),
\]
which is precisely the hypoelliptic regularization inequality~\eqref{eq:hypoelliptic_regularization}.
}

\begin{proof}%[Proof of \Cref{proposition:rate_of_convergence}]
Inspecting the entries of $\mat M_2$,
we notice that, since $\revision{a_0,a_1,a_2} \geq 0$ always appear with a negative sign,
\begin{equation}
  \mat M_2^+ := \begin{pmatrix}
    \dps \frac{1}{\beta\nu^2}
& - \left(2+\nu^{-2}\right)b_0
& - b_0 - b_1 r\\
    0 & b_0 r & - 4 b_1 \\
    0 & 0 & b_1
  \end{pmatrix}
  \succcurlyeq_+
  \mat M_2.
\end{equation}
Therefore, any bound from below for $\mat M_2$ is necessarily a bound from below also for~$\mat M_2^+$.
By examining the latter matrix, we obtain tentative scalings for the coefficients $b_0$ and $b_1$.
In a second step, we show that these scalings are in fact also suitable for $\mat M_2$.

\paragraph{Step 1: bound from below on $\mat M_2^+$.}%
In order to obtain a bound from below on~$\lambda_{\rm min}^+\left(\mat M_2^+\right)$,
we consider vectors~$\mat v$ in~\eqref{eq:positive_rayleigh} which have two non-zero elements.
%The positivity for vectors with a single non-zero element is equivalent to the positivity of the diagonal entries of $\mat M_2^+$. This is the case provided the coefficients $b_0$ and $b_1$ satisfy
%\begin{equation}
%  \label{eq:conditions_diagonal}
%  0 < b_0 < \frac{1}{\beta r \nu^2},
%  \qquad
%  0 < b_1 < \frac{b_0 r}{\norm{\derivative*{2}[V]{x^2}}_{\infty}}.
%\end{equation}
A necessary condition for the positivity of~$\lambda_{\rm min}^+\left(\mat M_2^+\right)$ is that the determinants of the following $2 \times 2$ symmetrized submatrices are positive:
\[
\mat M_2^{+,i,j} := \begin{pmatrix} \left[M_2^+\right]_{i,i} & \dps \frac12 \left[M_2^+\right]_{i,j} \\ \dps \frac12\left[M_2^+\right]_{j,i} & \left[M_2^+\right]_{j,j} \end{pmatrix}, \qquad 1 \leq i < j \leq 3.
\]
This leads to the following conditions:
\begin{subequations}
  \label{eq:simple_inequalities}%
  \begin{align}
    \label{eq:positivity_det_first_comatrix}
\frac{b_0 r}{\beta\nu^2} - b_0^2 \left(1 + \frac{1}{2\nu^2}\right)^2 > 0
\quad & \Longrightarrow \quad b_0
< \frac{r}{\beta} \, \min(4\nu^2, \nu^{-2}), \\
\label{eq:positivity_det_second_comatrix}
\frac{b_1}{\beta\nu^2} - \frac14\left(b_0 + b_1 r\right)^2 > 0
\quad & \Longrightarrow \quad
\left\{
\begin{aligned}
          & b_1 < \frac{1}{\beta r^2 \nu^2}, \\
          & b_0^2 < \frac{4b_1}{\beta\nu^2},
\end{aligned}
\right. \\
\label{eq:positivity_det_third_comatrix}
b_0 b_1 r - 4 b_1^2 > 0  \quad & \Longleftrightarrow \quad 0 < b_1 < \frac{b_0 r}{4}.
  \end{align}
\end{subequations}
Equation~\cref{eq:positivity_det_first_comatrix} shows that~$b_0$ is at most of order~$\min(r \nu^2,r \nu^{-2})$, so that, from~\cref{eq:positivity_det_second_comatrix} and with~\cref{eq:positivity_det_third_comatrix}, $b_1$ is at most of order
\[
  m(r,\nu) := \min(r^2 \nu^2, r^{-2} \nu^{-2}, r^2 \nu^{-2}) = \min\left(\gamma,\frac1\gamma,\frac{\gamma}{\nu^4}\right).
\]
Note the following inequalities, which will prove useful later on:
\begin{subequations}
  \begin{align}
    \label{eq:first_inequality_m}
&m(r, \nu) \leq \min (r^2 \nu^2, r^{-2} \nu^{-2}) \leq 1, \\
\label{eq:second_inequality_m}
&m(r, \nu) \leq \min (r^2 \nu^2, r^2 \nu^{-2}) \leq r^2, \\
\label{eq:third_inequality_m}
&m(r, \nu) \leq \min (r^{-2} \nu^{-2}, r^{2} \nu^{-2}) \leq \nu^{-2}.
  \end{align}
\end{subequations}
Condition~\cref{eq:positivity_det_third_comatrix} suggests that $b_0$ is of order~$r^{-1} m(r,\nu)$.
We therefore consider the choice
\begin{equation}
  \label{eq:scaling_b0_b1}
  b_0 = A r^{-1} m(r, \nu), \qquad b_1 = B m(r, \nu),
\end{equation}
with $A,B>0$ yet to be chosen. The matrix $\mat M_2^+$ then reads
\begin{equation*}
  \mat M_2^+ = m(r, \nu)
  \begin{pmatrix}
    \beta^{-1}\nu^{-2} m(r, \nu)^{-1}
    & - \left(2+\nu^{-2}\right) A r^{-1}
    & - A r^{-1} - B r \\
    0 & A
      & - 4 B \\
    0 & 0 & B
  \end{pmatrix}.
\end{equation*}
Now let
\begin{equation}
  \label{eq:definition_U}
  \mat U :=
  \begin{pmatrix}
    \beta^{1/2}\nu & 0 & 0 \\
    0 & A^{-1/2} m^{-1/2} & 0 \\
    0 & 0 & B^{-1/2} m^{-1/2}
  \end{pmatrix}
\end{equation}
and observe that
\begin{align*}
  \mat U
  \mat M_2^+
  \mat U
  =
  \begin{pmatrix}
    1
& - \left(2\nu+\nu^{-1}\right) r^{-1} \sqrt{\beta m} \sqrt{A}
& - \sqrt{\beta}\left[ AB^{-1/2} (m^{1/2} r^{-1} \nu) + B^{1/2} (m^{1/2} r\nu)\right]\\
    0 & 1
      & - 4 A^{-1/2}B^{1/2} \\
    0 & 0 & 1
  \end{pmatrix},
\end{align*}
where, to simplify the notation, we omitted the dependence of $m$ on $r$, $\nu$.
By definition of $m(r, \nu)$, it holds $m^{1/2} r^{-1} \nu \leq 1$, $m^{1/2} r\nu \leq 1$ and $\left(2\nu+\nu^{-1}\right) r^{-1} \sqrt{m} \leq 3$. We moreover choose $B = A^\delta$, which leads to
\begin{align*}
  \mat U
  \mat M_2^+
  \mat U
  \succcurlyeq_+
  \begin{pmatrix}
    1 & - 3\sqrt{\beta A}
& - \sqrt{\beta}\left[ A^{1-\delta/2} + A^{\delta/2} \right] \\
    0 & 1
      & - 4 A^{(\delta-1)/2} \\
    0 & 0 & 1
  \end{pmatrix}.
\end{align*}
This shows that it is possible to choose~$1 < \delta < 2$ and $A \in (0,1]$ sufficiently small so that $\lambda_{\rm min}^+\left(\mat U \mat M_2^+ \mat U\right) \geq 1/2$. To conclude this step, notice that, for any $\vect x \geq 0$,
\begin{align*}
  \vect x^\t \mat M_2^+ \vect x
  &= \vect x^\t \mat U^{-1}(\mat U \mat M_2^+ \mat U) \mat U^{-1} \vect x \geq \frac12 \norm{\mat U^{-1} \vect x}^2 \geq \frac{A^\delta}{2} m(r, \nu) \norm{\vect x}^2.
\end{align*}
which shows that $\mat M_2^+ \succcurlyeq_+ \frac{A^\delta}{2} m(r, \nu) \mat I_3$.

\paragraph{Step 2: Bound from below on $\mat M_2$.}%
\label{par:step_2_bound_on_mat_m_2}
We now consider the matrix $\mat M_1$ in~\cref{eq:hypocoercivity:equivalence_with_sobolev_norm} to be of the form
\[
  \mat M_1 = x \,
  \begin{pmatrix}
    2y^{-1} & - 1 & 0\\
    - 1 & y & 0 \\
    0 & 0 & 0
  \end{pmatrix}
  +
  w \,
  \begin{pmatrix}
    0 & 0 & 0 \\
    0 & 2 z^{-1} & - 1 \\
    0 & - 1 & z
  \end{pmatrix},
\]
where $x,y,w,z$ are new positive parameters. This corresponds to setting $b_0 = x$, $b_1 = w$ and $a_0 = 2x/y$, $a_1=xy+2w/z$, $a_2 = zw$. The interest of this parametrization is to bring the number of parameters down from 5 to 4 and to directly ensure that $\mat M_1$ is positive definite. Motivated by the scalings~\eqref{eq:scaling_b0_b1}, we choose~$x$ to be of order~$r^{-1} \, m(r, \nu)$ and~$w$ to be of order~$m(r, \nu)$. To guess the scalings of $y$ and $z$ with respect to $r$ and $\nu$, we rely on the following observations:
\begin{itemize}
  \item to ensure that the $(2, 2)$ entry of $\mat M_2$, which reads $b_0 r -\frac{3a_1}{2} - b_1\norm{\derivative*{2}[V]{x^2}}_{\infty}$, is positive for all values of $r$ and $\nu$, it is necessary that $a_1$ is at most of order~$b_0 \, r$,
    which suggests that~$y$ scales as~$r$;
  \item the coefficient $a_2$ appears only in matrix entries where $b_1$ is also present
    and, in these entries,
    both coefficients appear with prefactors that scale identically with respect to~$r$ and~$\nu$.
    This suggests that~$z$ is of order~1.
\end{itemize}
Guided by these observations, we consider the following choice:
\[
  x = A \, r^{-1} \, m(r, \nu),
  \qquad
  y = A^\eta \, r,
  \qquad
  w = A^{\delta} \, m(r, \nu),
  \qquad
  z = A^{\rho},
\]
where $\eta, \delta, \rho$ are exponents independent of~$r$ and~$\nu$ yet to be determined, while $A>0$ is a small parameter.
With the same matrix~$\mat U$ as in~\eqref{eq:definition_U} with $B=A^\delta$,
we obtain $\mat U \mat M_{2} \mat U \succcurlyeq_+ \mat I_3 + \mat R$, where $\mat R$ is an upper diagonal matrix with entries
\[
  \begin{aligned}
    \left[\mat R\right]_{1,1} & = - \beta\left(A + A^{1-\eta}\right), \\
    \left[\mat R\right]_{1,2} & = -\sqrt{\beta}\left( 2A^{1/2-\eta}+A^{1/2+\eta} +2A^{\delta-\rho-1/2}+3A^{1/2}\right),\\
    \left[\mat R\right]_{1,3} & = - \sqrt{\beta}\left(A^{1-\delta/2} + A^{\delta/2}\right),\\
    \left[\mat R\right]_{2,2} & = - \frac{3A^\eta}{2} - 3 A^{\delta-\rho-1} - A^{\delta-1}\norm{\derivative*{2}[V]{x^2}}_{\infty}, \\
    \left[\mat R\right]_{2,3} & = -4 A^{(\delta-1)/2} - A^{\eta+(1-\delta)/2} - 2A^{(\delta-1)/2-\rho} - A^{\rho +(\delta-1)/2}\norm{\derivative*{2}[V]{x^2}}_{\infty}, \\
    \left[\mat R\right]_{3,3} & = -\frac{5A^\rho}{2} .
  \end{aligned}
\]
In order for all entries of~$\mat R$ to converge to~0 as~$A \to 0$,
we require $0 < \eta < 1/2$, $1 < \delta < 1 + 2 \eta$ and $0 < \rho < (\delta-1)/2$.
By the same reasoning at the one allowing to conclude Step~1,
we can show that there exist~$A>0$ sufficiently small and a constant~$C > 0$ (which depends on~$A$) for which $\mat M_2 \succcurlyeq_+ C m(r,\nu) \mat I_3$.
%\begin{align*}
%  &a_0 = 2 A^{1+} r^{-2} m(r, \nu),
%  \qquad
%  &&a_1 = 3 \, \varepsilon^{5} \, m(r, \nu),
%  \qquad
%  &&a_2 = \varepsilon^{7} \, m(r, \nu), \\
%  &b_0 = \varepsilon^{4} \, r^{-1} \, m(r, \nu),
%  &&b_1 = \varepsilon^{6} \, m(r, \nu).
% \end{align*}

Finally, it is easy to see that the smallest eigenvalue~$\alpha(\gamma,\nu)$ of the real symmetric matrix~$\mat M_1$ is positive and scales as $\min(1,r^{-2})m(r,\nu)$. Upon decreasing~$A$ if necessary, we can further ensure that $\mat M_1 \preccurlyeq \mat I_3$.
\end{proof}

\begin{remark}
  \revision{%
    Since it holds $\widetilde{\mat M}_2 \succcurlyeq_+ \mat M_2$,
    one might wonder whether a sharper lower bound in~\cref{proposition:rate_of_convergence_hypocoercitivy} could have been obtained by working directly with $\widetilde{\mat M}_2$.
    While it is possible that a larger constant $C$ on the right-hand side could have been obtained,
    the dependence of the lower bound on $\gamma$ and $\nu$ in~\cref{proposition:rate_of_convergence_hypocoercitivy} is in fact sharp.
    Indeed, better scalings with respect to $\gamma$ and $\nu$ in~\cref{proposition:rate_of_convergence_hypocoercitivy} would lead to better scalings in~\cref{thm:hypocoercivity_h1}
    and therefore also in~\eqref{eq:exponential_growth_bound},
    but we show in~\cref{sec:confirmation_of_the_rate_of_convergence_in_the_quadratic_case} that the scalings of the decay rate in~\eqref{eq:exponential_growth_bound} are optimal.
  }
\end{remark}

\section{Scaling limits of the effective diffusion coefficient}%
\label{sec:multiscale_analysis}

We study in this section various limits for the effective diffusion coefficient~\eqref{eq:introduction:effective_diffusion} of GLE,
namely the short memory limit $\nu \to 0$ in \cref{sub:gle:short_memory_limit}
(for which we expect to recover the behavior of standard Langevin dynamics),
the overdamped limit $\gamma \to \infty$ in \cref{sub:gle:the_overdamped_limit}
(for which we expect to recover the behavior of overdamped Langevin dynamics),
and finally the underdamped limit $\gamma \to 0$ in \cref{sub:gle:the_underdamped_limit}
(for which we expect the effective diffusion coefficient to scale as $\gamma^{-1}$,
as \revision{for} standard Langevin dynamics in the same limit).

All the analysis in this section is done for GLE in a \emph{periodic} potential on the domain $\mathcal{X} = \torus$. The general strategy is the following:
\begin{enumerate}[(i)]
\item we first formally approximate the solution to the Poisson equation~\eqref{eq:introduction:poisson_equation} by some function~$\widehat{\phi}$ obtained by an asymptotic analysis where the solution~$\phi$ is expanded in powers of a small parameter;
\item we next rely on the the resolvent estimates~\eqref{eq:resolvent_bound_GL1} to provide bounds on~$\left\|\phi-\widehat{\phi}\right\|$;
\item we finally deduce the leading order behavior of the diffusion coefficient by replacing~$\phi$ by~$\widehat{\phi}$ in~\eqref{eq:introduction:effective_diffusion}.
\end{enumerate}
Let us also already emphasize that,
while the results presented in \cref{sub:gle:short_memory_limit,sub:gle:the_overdamped_limit} are mathematically rigorous,
the discussion in \cref{sub:gle:the_underdamped_limit} is only formal since the asymptotic analysis is quite cumbersome in the underdamped setting,
where the leading part of the dynamics is a Hamiltonian evolution.

\subsection{The short memory limit}%
\label{sub:gle:short_memory_limit}

In this section,
we show rigorously that, in the limit as $\nu \to 0$ and for
\begin{equation}
  \label{eq:gamma_GLE_nu0}
  \gamma = \lambda^\t \mat A^{-1} \vect \lambda = \int_{0}^{\infty} \gamma(t) \, \d t > 0
\end{equation}
fixed, the effective diffusion coefficient $D_{\gamma,\nu}$ associated with the GLE converges to
that associated with the Langevin dynamics for the same value of~$\gamma$, denoted by $D_{\gamma}$. The latter
diffusion coefficient is defined in terms of the solution~$\phi_{\rm Lang}$ of the Poisson equation $-\mathcal{L}_{\rm Lang}\phi_{\rm Lang}\revision{ = p}$, where
\[
  \mathcal L_{\rm Lang} = p \derivative{1}{q} - \derivative*{1}[V]{q}(q) \derivative{1}{p} - \gamma p \derivative{1}{p} + \frac{\gamma}{\beta} \derivative{2}{p},
\]
the generator of the Langevin dynamics, acts on functions of~$(q,p)$.
More precisely, denoting by~$\mu_{\rm Lang}$ the marginal of the invariant probability measure~$\mu$ in the~$(q,p)$ variables,
\begin{equation}
  \label{eq:Dgamma_Lang}
  D_{\gamma} = \int_{\torus \times \real} \phi_{\rm Lang}(q,p) \, p \, \mu_{\rm Lang}(\d q \, \d p), \qquad
  \mu_{\rm Lang}(\d q \, \d p) = \frac{1}{Z_{\beta}}\e^{- \beta H(q,p)} \, \d q \, \d p,
\end{equation}
where $Z_{\beta} = \int_{\torus \times \real} \e^{- \beta H(q,p)} \, \d q \, \d p$ is the normalization constant.
Let us recall that, by the results of~\cite{MR1924934,Kopec2015}, the solution~$\phi_{\rm Lang}$ is a smooth function which, together with all its derivatives, grows at most polynomially as $|p| \to \infty$. It particular it belongs to~$L^2(\mu_{\rm Lang})$, so that~$D_\gamma$ is well defined by a Cauchy--Schwarz inequality.

We present the analysis for a general quasi-Markovian approximation of the noise of the form~\eqref{eq:markovian_approximation},
with parameters $\vect \lambda$ and $\mat A$ rescaled in such a way that
the correlation time of the noise appears explicitly as a parameter, while keeping~$\gamma$ fixed.
More precisely, we rewrite the generator of GLE as (indicating explicitly the dependence on~$\nu$)
\[
\mathcal{L}_\nu = \mathcal{A}_0 + \frac{1}{\nu} \mathcal{A}_1 + \frac{1}{\nu^2} \mathcal{A}_2,
\]
with
\[
  \mathcal{A}_0 = p \derivative{1}{q} - \derivative*{1}[V]{q}(q) \derivative{1}{p},
  \qquad \mathcal{A}_1 = \vect \lambda^\t \vect z \derivative{1}{p} - p \vect \lambda^\t \nabla_{\vect z},
  \qquad \mathcal{A}_2 = - \vect z^\t \mat A^\t \nabla_{\vect z} + \frac1\beta \mat A : \nabla^2_{\vect z}.
\]
Note that models GL1 and GL2 are already in this rescaled form.
Recall also that, denoting by $(q_t^\nu, p_t^\nu, \vect z_t^\nu)$ the solution to~\eqref{eq:markovian_approximation},
$(q_t^\nu,p_t^\nu)$ converges in the short memory limit~$\nu \to 0$,
in the sense of weak convergence of probability measures on $C([0, T], \torus \times \real)$ for some \revision{fixed} final time~$T>0$,
to the solution of the Langevin equation~\eqref{eq:model:langevin} with friction coefficient~\eqref{eq:gamma_GLE_nu0};
see for example~\cite[Theorem 2.6]{MR2793823} or~\cite[Result 8.4]{pavliotis2011applied}.

We have the following result, which can be obtained constructively by formal asymptotics; see~\cite{thesis_urbain} for details.
\begin{lemma}
  \label{lemma:auxiliary_short_memory}
  The function
  \[
    \widehat \phi(q, p, \vect z) := \phi_{\rm Lang}(q,p)
    + \nu \vect \lambda^\t \mat A^{-1} \vect z \, \derivative{1}[\phi_{\rm Lang}]{p}(q,p) + \nu^2 \, \psi_2(q, p, \vect z) + \nu^3 \, \psi_3(q, p, \vect z),
  \]
  where
  \begin{align*}
    \psi_2(q, p, \vect z) & =  \frac{1}{2} \left( \left( \vect \lambda^\t \mat A^{-1} \vect z\right)^2 \revision{- \frac{1}{\beta} \vect \lambda^\t \mat A^{-1} \mat A^{-\t} \vect \lambda} \right) \, \derivative{2}[\phi_{\rm Lang}]{p^2}(q,p), \\
    \psi_3(q, p, \vect z) & = \left( \frac{1}{6} \left( \vect \lambda^\t \mat A^{-1} \vect z\right)^3 + \frac{\gamma}{\beta} \vect \lambda^\t \mat A^{-2} \vect z \revision{ - \frac{1}{2\beta}(\vect \lambda^\t \mat A^{-1} \vect z) (\vect \lambda^\t \mat A^{-1} \mat A^{-\t} \vect \lambda) }\right) \derivative{3}[\phi_{\rm Lang}]{p^3}(q, p) \\
                           & \ \  - \left(\gamma p + \derivative*{1}[V]{q}(q)\right) \vect \lambda^\t \mat A^{-2} \vect z \derivative{2}[\phi_{\rm Lang}]{p^2}(q, p) + p \vect \lambda^\t \mat A^{-2} \vect z \derivative{2}[\phi_{\rm Lang}]{q, p}(q, p),
  \end{align*}
  \revision{belongs to~$L^2_0(\mu)$ and} satisfies
  \begin{equation}
    \label{eq:result_formal_asymptotics_n}
    - \mathcal L_{\nu} \widehat \phi = p - \nu^2 (\mathcal A_0 \psi_2+ \mathcal A_1 \psi_3) - \nu^3 \mathcal A_0\psi_3.
  \end{equation}
\end{lemma}

\begin{proof}
  We first compute $\mathcal{A}_2\left(\vect \lambda^\t \mat A^{-1} \vect z\right) = - \vect \lambda^\t \vect z$ and
  $\mathcal{A}_2\left(\vect \lambda^\t \mat A^{-2} \vect z\right) = - \vect \lambda^\t \mat A^{-1} \vect z$,
  as well as, for $\alpha=2,3$,
  \[
    \mathcal{A}_2\left[\left(\vect \lambda^\t \mat A^{-1} \vect z\right)^\alpha\right] = -\alpha \left(\vect \lambda^\t \mat A^{-1} \vect z\right)^{\alpha-1}\vect \lambda^\t \vect z + \frac{\alpha(\alpha-1)\gamma}{\beta} \left(\vect \lambda^\t \mat A^{-1} \vect z\right)^{\alpha-2}.
  \]
  The result can then be verified directly by calculating $- \mathcal L_{\nu} \widehat \phi$ and gathering terms with the same powers of~$\nu$.
\end{proof}

We can then provide the convergence result on the effective diffusion coefficient.

\begin{proposition}
  Fix $\gamma>0$ and assume that there exists $C>0$ such that
  \begin{equation}
    \label{eq:resolvent_estimate_nu0}
    \forall \nu \in (0,1], \qquad \left\|\mathcal{L}_\nu^{-1}\right\|_{\mathcal{B}(L^2_0(\mu))} \leq C.
  \end{equation}
  Then there exists~$R>0$ such that, for any $0 < \nu \leq 1$,
  \[
    \abs{ D_{\gamma,\nu} - D_{\gamma} } \leq R \nu^2.
  \]
\end{proposition}

Note that the condition~\eqref{eq:resolvent_estimate_nu0} follows for GL1 from the resolvent estimate~\eqref{eq:resolvent_bound_GL1},
and for GL2 from the resolvent estimate~\eqref{eq:resolvent_bound_GL2} in Appendix~\ref{sec:longtime_behavior_for_model_gl2}.
It would be possible to weaken this condition by allowing some power law growth with respect to~$\nu$ on the right-hand side of~\eqref{eq:resolvent_estimate_nu0} upon further continuing the asymptotic expansion of \cref{lemma:auxiliary_short_memory} in order to have higher order terms in~\eqref{eq:result_formal_asymptotics_n}.

\begin{proof}
  By the result from~\cite{Kopec2015} mentioned previously,
  $\phi_{\rm Lang}$ and all its derivatives are smooth and grow at most polynomially as $|p| \to \infty$.
  Given the definitions of $\mathcal A_1$ and $\mathcal A_2$,
  this implies that the coefficients of~$\nu^2$ and~$\nu^3$ on the right-hand side of~\cref{eq:result_formal_asymptotics_n}
  are smooth functions in~$\lp{2}{\mu}$. Since these functions are independent of~$\nu$, there exists~$K>0$ such that
  \[
    \norm{\mathcal L_{\nu} \left(\phi_{\nu} - \widehat \phi\right)} \leq K \nu^2,
  \]
  where~$\phi_{\nu}$ denotes the solution to the Poisson equation~$- \mathcal L_{\nu} \phi_{\nu} = p$.
  We therefore obtain \revision{with~\eqref{eq:resolvent_estimate_nu0}} that $\norm{\phi_{\nu} - \widehat \phi} \leq CK \nu^2$ for all~$\nu \leq 1$. Since the function $(q,p,\vect z) \mapsto \vect \lambda^\t \mat A^{-1} \vect z \, \derivative{1}[\phi_{\rm Lang}]{p}(q,p)$ has average~0 with respect to~$\mu$, the desired estimate follows by substituting~$\phi_\nu$ by~$\widehat \phi$ in~\eqref{eq:introduction:effective_diffusion}, integrating over~$\vect z$, and comparing with~\eqref{eq:Dgamma_Lang}.
\end{proof}

\begin{remark}
  For the special case of the model GL1, if we had wanted to show only that $|D_{\gamma,\nu} - D_{\gamma}| \to 0$ in the limit as $\nu \to 0$ without making precise a convergence rate,
  we could have proceeded more directly from~\cite[Theorem 2.6]{MR2793823},
  which can be leveraged by using a reasoning similar to that in the proof of~\cite[Proposition 3.3]{MR2394704}.
  With the same notation as above,
  \cite[Theorem~2.6]{MR2793823} implies that
  the random variable $(q_t^{\nu}, p_t^{\nu})$ converges weakly, in the limit as $\nu \to 0$,
  to the solution of~\eqref{eq:model:langevin} evaluated at $t$,
  for any $t \geq 0$ and any initial condition with finite moments of all orders.
  Consequently, for any bounded and continuous function $f(q,p)$, it holds as $\nu \to 0$
  \[
    \forall (q, p, z) \in \torus \times \real \times \real,
    \qquad \e^{t \mathcal L_{\nu}} f(q, p, z) \to  \e^{t \mathcal L_{\rm Lang}} f(q, p),
  \]
  and so also $\e^{t \mathcal L_{\nu}} f \to  \e^{t \mathcal L_{\rm Lang}} f$ in $\lp{2}{\mu}$ by dominated convergence.
  % (To make sense of $\e^{t\mathcal L_{\nu}} f$, we use the obvious embedding of $L^2(\mu_{\rm Lang})$ into $L^2(\mu)$.)
  By density of the bounded and continuous functions in $\lp{2}{\mu_{\rm Lang}}$ and the continuity of the propagators on $\lp{2}{\mu_{\rm Lang}}$,
  this limit holds in fact for any $f \in \lp{2}{\mu_{\rm Lang}}$.
  Therefore, for any $f \in  L^2_0(\mu_{\rm Lang})$, it holds, as $\nu \to 0$,
  \begin{align*}
    \left\| \mathcal L_{\nu}^{-1}f - \mathcal L_{\rm Lang}^{-1}f \right\|
    &= \norm{\int_{0}^{\infty} \e^{t\mathcal L_{\nu}} f \, \d t - \int_{0}^{\infty} \e^{t\mathcal L_{\rm Lang}} f \, \d t} \\
    &\leq \int_{0}^{\infty} \norm{\e^{t\mathcal L_{\nu}} f - \e^{t\mathcal L_{\rm Lang}} f} \d t \xrightarrow[\nu \to 0]{} 0.
  \end{align*}
  The last limit is justified by dominated convergence because,
  by~\eqref{eq:exponential_growth_bound} and the corresponding result for the Langevin equation,
  we have the bound $\norm{\e^{t\mathcal L_{\nu}} f - \e^{t\mathcal L_{\rm Lang}} f}
  \leq \norm{\e^{t\mathcal L_{\nu}} f} + \norm {\e^{t\mathcal L_{\rm Lang}} f} \leq C \e^{-\lambda t}$,
  for some positive constants $C$ and $\lambda$ independent of $\nu$.
  This concludes the proof since
  \[
    D_{\gamma,\nu} = \ip{-\mathcal L_{\nu}^{-1}p}{p} \xrightarrow[\nu \to 0]{} \ip{-\mathcal L_{\rm Lang}^{-1}p}{p} = D_{\gamma}.
  \]
\end{remark}
% \begin{remark}
%   We notice, by substitution of~\eqref{eq:solution_psi2_when_alpha_is_1} in~\eqref{eq:asymptotics:effective_diffusion},
%   that the coefficient of $\nu^2$ in the asymptotic expansion is zero when $\alpha = 1$.
%   In this case it is possible to show that the first nonzero correction is of order $\mathcal O(\nu^4)$.
% \end{remark}

\subsection{The overdamped limit}
\label{sub:gle:the_overdamped_limit}

We prove in this section that the effective diffusion coefficient~$D_{\gamma,\nu}$ associated with the model GL1 converges,
as the effective friction~\eqref{eq:gamma_GLE_nu0} goes to infinity,
to the effective diffusion coefficient~$D_{\rm ovd}$ associated with the overdamped Langevin dynamics~\eqref{eq:model:overdamped}.
It is certainly possible to extend our analysis to more general models of noise than GL1,
but the algebra involved in the asymptotic analysis of \cref{lem:asymptotic_expansion_ovd} below becomes more cumbersome,
so we refrain from doing so.

Denoting by~$\mu_{\rm ovd}(\d q)$ the marginal of~$\mu(\d q \, \d p \, \d z)$
(which has a density proportional to~$\e^{-\beta V(q)}\,\d q$),
the effective diffusion coefficient for overdamped dynamics is defined from the unique solution~$\phi_{\rm ovd} \in L^2(\mu_{\rm ovd})$ to the Poisson equation
\begin{equation}
  \label{eq:poisson_equation_overdamped}%
  - \mathcal{L}_{\rm ovd} \phi_{\rm ovd} = - V', \qquad \int_\torus \phi_{\rm ovd} \, \d\mu_{\rm ovd} = 0,
\end{equation}
where $\mathcal{L}_{\rm ovd}$ acts on functions of~$q$ as
\[
  \mathcal{L}_{\rm ovd} = - \derivative*{1}[V]{q}(q)  \derivative*{1}{q} + \beta^{-1} \derivative*{2}{q^2}.
\]
By elliptic regularity, the solution~$\phi_{\rm ovd}$ belongs to~$C^\infty(\torus)$.
It can then be shown,
using the tools from~\cite[Chapter~3]{BLP} (see for instance~\cite[Section~1.2]{FHS14} and~\cite[Chapter 13]{pavliotis2008multiscale}) that
\begin{equation}
  \label{eq:D_ovd}
  D_{\rm ovd} = \beta^{-1} + \int_\torus \phi_{\rm ovd} V' \, \d\mu_{\rm ovd}.
\end{equation}
The overdamped limit of the effective diffusion coefficient~\eqref{eq:Dgamma_Lang} for the Langevin dynamics was already studied in~\cite{MR2394704}
(see also~\cite[Section~3.1.1]{LMS16}),
where it is shown that $D_\gamma = D_{\rm ovd}/\gamma + \bigo{\gamma^{-2}}$.
We provide the counterpart of this estimate for GL1 in the following result.

\begin{proposition}
  \label{prop:ovd_diff}
  Consider the model GL1 and recall the definition~\eqref{eq:gamma_GLE_nu0} of the effective friction~$\gamma$. There exists $R>0$ such that
  \[
    \forall \gamma \geq 1, \qquad \left| D_{\gamma, \nu} - \frac{1}{\gamma} D_{\rm ovd} \right| \leq \frac{R}{\gamma^{3/2}}.
  \]
\end{proposition}

In fact, using a more involved asymptotic analysis (see \cref{rmk:order_2_ovd}), it would be possible to show that the difference $D_{\gamma,\nu}-D_{\rm ovd}/\gamma$ is of order~$\gamma^{-2}$. In order to perform the asymptotic analysis, we rewrite the generator of the GL1 model as (indicating explicitly the dependence on the friction~$\gamma>0$)
\[
  \mathcal L_{\gamma} = \sqrt{\gamma} \mathcal{A}_0 + \mathcal{A}_1,
  \qquad
  \mathcal{A}_0 = \frac1\nu \left( z \derivative{1}{p} - p \derivative{1}{z}\right),
  \qquad
  \mathcal{A}_1 = p \derivative{1}{q} - \derivative*{1}[V]{q}(q) \derivative{1}{p}
  + \frac{1}{\nu^2}\left( - z \derivative{1}{z} +  \frac1\beta \derivative{2}{z^2}\right).
\]
The resolvent estimates~\eqref{eq:resolvent_bound_GL1} suggest that the solution to the Poisson equation $- \mathcal L_{\gamma} \phi_{\gamma} = p$ is of order~$\gamma$ as~$\gamma \to \infty$. However, by analogy with asymptotic calculations for the Langevin equation~\cite{MR2427108} where the leading order term in the series expansion in inverse powers of~$\gamma$ of the solution to the Poisson equation~$-\mathcal{L}_{\rm Lang}\phi_{\rm Lang} = p$ scales as $\mathcal O(1)$ (which can also be seen from the expressions of~$\mathcal{L}_{\rm Lang}$ provided in~\cite[Theorem~2.5]{LMS16} and~\cite{BFLS20}), we formally expand~$\phi_\gamma$ as
\[
  \phi_{\gamma} = \overline{\phi}_0 + \gamma^{-1/2} \overline{\phi}_1 + \gamma^{-1} \, \overline{\phi}_2 + \dotsb.
\]
The correctness of this assumption can be checked \emph{a posteriori}, using formal asymptotics (the details of which are presented in~\cite{thesis_urbain}) to calculate the functions $\overline{\phi}_i$ in a systematic way. This is made precise in the following technical result, where $L^2_0(\mu_{\rm ovd})$ is the subset of functions of~$L^2(\mu_{\rm ovd})$ with mean~0 with respect to~$\mu_{\rm ovd}(\d q)$.
For simplicity of notation we assume that $\nu = 1$.

\begin{lemma}
\label{lem:asymptotic_expansion_ovd}
Denote by $\phi_{\gamma}$ the solution to the Poisson equation~\eqref{eq:introduction:poisson_equation} for GL1, and by
$\psi$ the unique solution in~$L^2_0(\mu_{\rm ovd})$ of the Poisson equation
\begin{equation}
  \label{eq:auxiliary_poisson_overdamped}%
  \mathcal{L}_{\rm ovd} \psi =
  \left|\derivative*{1}[V]{q}\right|^{2} \derivative*{2}[\phi_{\rm ovd}]{q^2}
  - \frac{3}{2 \beta} \, \derivative*{1}[V]{q} \derivative*{3}[\phi_{\rm ovd}]{q^3}
  - \frac{1}{\beta} \, \derivative*{2}[V]{q^2} \derivative*{2}[\phi_{\rm ovd}]{q^2}
  + \frac{1}{2 \beta^2} \derivative*{4}[\phi_{\rm ovd}]{q^4}.
\end{equation}
Define the function
\begin{align*}
  \widehat \phi =
  \phi_{\rm ovd}
  & + \gamma^{- \frac{1}{2}} \, z(\derivative*{1}[\phi_{\rm ovd}]{q} + 1)
  + \gamma^{-1} \left( \frac{z^2}{2} \, \derivative*{2}[\phi_{\rm ovd}]{q^2} + p \, ( \derivative*{1}[\phi_{\rm ovd}]{q} + 1) + \psi \right) + \gamma^{- \frac{3}{2}} \Phi_{3/2} + \gamma^{-2} \Phi_2 + \gamma^{ - \frac{5}{2} } \Phi_{5/2},
\end{align*}
with
\[
  \begin{aligned}
    \Phi_{3/2} & = p z \, \derivative*{2}[\phi_{\rm ovd}]{q^2} + \frac{z^3}{6} \derivative*{3}[\phi_{\rm ovd}]{q^3} + z \derivative*{1}[\psi]{q}, \\
    \Phi_2 & = \left( p \derivative*{1}[V]{q} + \frac{p^2}{2}
      + \frac{p^{2}}{2} \derivative*{2}[V]{q^2}+ \frac{z^{2}}{2} \derivative*{2}[V]{q^2}
    \right) \derivative*{2}[\phi_{\rm ovd}]{q^2}
    + \left(\frac{p z^{2}}{2} - \frac{p}{\beta}
    + \frac{p^{2}}{4} \derivative*{1}[V]{q}+ \frac{z^{2}}{4} \derivative*{1}[V]{q} \right) \derivative*{3}[\phi_{\rm ovd}]{q^3} \\
           & \ \ + \left( \frac{z^{4}}{24} - \frac{p^{2}}{4 \beta} - \frac{z^{2}}{4 \beta}  \right)  \, \derivative*{4}[\phi_{\rm ovd}]{q^4} + p \derivative*{1}[\psi]{q} + \frac{z^{2}}{2}  \derivative*{2}[\psi]{q^2}, \\
    \end{aligned}
  \]
  and
  \[
    \begin{aligned}
      \Phi_{5/2} & =
      \left( p z \, \derivative*{2}[\psi]{q^2} + \frac{z^{3}}{6} \derivative*{3}[\psi]{q^3} \right)
      + \left(
        \frac{p^{2} z}{2} \, \derivative*{3}[V]{q^3}
        + p z \derivative*{2}[V]{q^2}
        + \frac{z^{3}}{2} \, \derivative*{3}[V]{q^3}
        - z \, \derivative*{1}[V]{q} \derivative*{2}[V]{q^2}
        - z \, \derivative*{1}[V]{q}
      \right)
      \derivative*{2}[\phi_{\rm ovd}]{q^2} \\
      & \ \ +
    \left(
      \frac{3 p^{2} z}{4} \, \derivative*{2}[V]{q^2}
      + \frac{p^{2} z}{2}
      + p z \, \derivative*{1}[V]{q}
      + \frac{3 z^{3}}{4} \derivative*{2}[V]{q^2}
    - \frac{z}{2} \, \left| \derivative*{1}[V]{q} \right|^{2} + \frac{z}{\beta}
  \right)
  \derivative*{3}[\phi_{\rm ovd}]{q^3} \\
  & \ \  +
  \left(
    \frac{p^{2} z}{4} \, \derivative*{1}[V]{q}
    + \frac{p z^{3}}{6} - \frac{p z}{\beta}
    + \frac{z^{3}}{4} \,  \derivative*{1}[V]{q}
    + \frac{z \derivative*{1}[V]{q}}{2\beta}
  \right)
  \derivative*{4}[\phi_{\rm ovd}]{q^4}
  +
  \left(
    - \frac{p^{2} z}{4 \beta} + \frac{z^{5}}{120} - \frac{z^{3}}{4 \beta}
  \right)
  \derivative*{5}[\phi_{\rm ovd}]{q^5}.
  \end{aligned}
  \]
  \revision{Let also $\widehat \phi_0 := \widehat \phi - \int_{\torus \times \real \times \real} \widehat \phi \, \d \mu$, so that $\widehat \phi_0 \in L^2_0(\mu)$.}
  Then $\mathcal{R} = \gamma^{\frac{5}{2}} \mathcal L_{\gamma} \left(\revision{\widehat \phi_0} - \phi_{\gamma}\right)$ is a well defined function of~$L^2_0(\mu)$ which is independent of~$\gamma>0$.
\end{lemma}

Including the term multiplying $\gamma^{\frac{5}{2}}$ adds a significant amount of complexity,
but this is required for rigorously proving the convergence of $\gamma D_{\gamma, \nu}$ to $D_{\rm ovd}$ in the limit as $\gamma \to \infty$.

\begin{proof}
  Note first that~\eqref{eq:auxiliary_poisson_overdamped} admits a solution because the right-hand side belongs to $L^2_0(\mu_{\rm ovd})$,
  which can be shown by using integration by part on the last term:
  \begin{align*}
    &\int_{\torus} \frac{\derivative*{4}[\phi_{\rm ovd}]{q^4}}{2 \beta^2} \, \d \mu_{\rm ovd}
      = \int_{\torus} \frac{\derivative*{1}[V]{q} \, \derivative*{3}[\phi_{\rm ovd}]{q^3}}{2 \beta} \, \d \mu_{\rm ovd}
      = \int_{\torus} \frac{3}{2 \beta} \, \derivative*{1}[V]{q} \, \derivative*{3}[\phi_{\rm ovd}]{q^3} \, \d \mu_{\rm ovd}
      - \int_{\torus} \left( \abs{\derivative*{1}[V]{q}}^2 - \frac{\derivative*{2}[V]{q^2}}{\beta}  \right) \, \derivative*{2}[\phi_{\rm ovd}]{q^2} \, \d \mu_{\rm ovd}.
  \end{align*}
  A straightforward computation shows that
  $\dps \mathcal L_{\gamma} \left(\revision{\widehat \phi_0} - \phi_{\gamma}\right) = \gamma^{-5/2} \mathcal{A}_1 \Phi_{5/2}$.
  The function $\mathcal{R}$ is in~$L^2(\mu)$ by direct inspection,
  since $\phi_{\rm ovd}$ and~$\psi$ are smooth and defined on a compact domain.
  Finally, $\mathcal{R}$ has average~0 with respect to~$\mu$ since
  it is in the image of~$\mathcal{A}_1$,
  which is the sum of two generators of stochastic dynamics leaving~$\mu$ invariant.
\end{proof}

 We can conclude this section with the proof of \cref{prop:ovd_diff}.

 \begin{proof}[Proof of \cref{prop:ovd_diff}]
   Note first that,
   by an integration by parts in~\eqref{eq:D_ovd},
   \begin{equation}
     \label{eq:reformulation_D_ovd}
     D_{\rm ovd} = \beta^{-1} \left( 1 + \int \phi_{\rm ovd}'  \, \d\mu \right).
   \end{equation}
   The resolvent estimate~\eqref{eq:resolvent_bound_GL1} and Lemma~\ref{lem:asymptotic_expansion_ovd} next imply that
  \begin{equation}
    \label{eq:norm_phi_gamma_hat_phi}
    \norm{\revision{\widehat \phi_0} - \phi_{\gamma}}
    \leq C\left\| \mathcal{R} \right\| \gamma^{- \frac{3}{2}}.
  \end{equation}
  Moreover, using that~$\phi_{\rm ovd}$ and~$\psi$ have average~0 with respect to~$\mu$,
  and that functions with odd powers of~$p$ and~$z$ also have vanishing averages with respect to~$\mu$,
  \begin{align}
    \notag
    \revision{\int_{\torus \times \real \times \real} \widehat{\phi}_0 p \, \d\mu}
    &= \int_{\torus \times \real \times \real} \widehat{\phi} p \, \d\mu \\
    \label{eq:diff_hat_phi}
    &= \frac{1}{\beta\gamma} \int_{\torus \times \real \times \real} \left(1+\phi_{\rm ovd}'\right) \d\mu + \frac{1}{\gamma^2} \int_{\torus \times \real \times \real} \left( \Phi_2 + \gamma^{-1/2} \Phi_{5/2}\right)p \, \d\mu
  \end{align}
  The result then follows by combining the previous equality with~\eqref{eq:reformulation_D_ovd} and~\eqref{eq:norm_phi_gamma_hat_phi}.
\end{proof}

\begin{remark}
  \label{rmk:order_2_ovd}
  Note that the term of order~$\gamma^{-3/2}$ vanishes in the effective diffusion coefficient~\eqref{eq:diff_hat_phi}. We therefore expect that the bound in \cref{prop:ovd_diff} can be improved to~$\gamma^{-2}$, as for Langevin dynamics. There is no conceptual obstruction to this end, but this would require going to an extra order in \cref{lem:asymptotic_expansion_ovd} by making explicit the term $\gamma^{-3}\Phi_3$ and changing~$\Phi_{5/2}$ in the definition of~$\widehat{\phi}$, which is algebraically cumbersome.
\end{remark}

\subsection{The underdamped limit}%
\label{sub:gle:the_underdamped_limit}
We consider in this section the underdamped limit~$\gamma\to 0$ for the GL1 model.
\revision{We will assume without loss of generality that $\min_{q \in [-\pi, \pi]} V(q) = 0$.}
The underdamped limit of Langevin dynamics in one dimension was carefully studied in~\cite{MR2394704}, see also~\cite{MR2427108},
where it is shown that the effective diffusion coefficient associated with the Langevin dynamics~\eqref{eq:model:langevin},
multiplied by the factor $\gamma$,
converges in the limit as $\gamma \to 0$ to
\begin{equation}
  \label{eq:effective_underdamped_Langevin}
  D_{\rm und} = \frac{8 \pi^2}{\beta Z_\beta} \int_{E_0}^{\infty} \frac{\e^{- \beta E}}{S_{\rm und}(E)} \, \d E,
  \qquad Z_\beta = \int_{\torus \times \real} \e^{-\beta H} = \sqrt{\frac{2\pi}{\beta}} \int_\torus \e^{-\beta V},
\end{equation}
where~$E_0 =  \max_{q \in [-\pi, \pi]} V(q)$, and
\begin{equation}
  \label{eq:def_S_0}
  S_{\rm und}(E) = \int_{-\pi}^{\pi} P(q, E) \, \d q, \qquad P(q,E) := \sqrt{2(E - V(q))}.
\end{equation}
In this section, we show that a similar result holds for the underdamped limit of GL1.
In particular, we motivate the following result.
\begin{result}
\label{result:underdamped}
In the limit as $\gamma \to 0$,
the effective diffusion coefficient~$D_{\gamma,\nu}$ for GL1 scales as $D^*_{\nu}/\gamma$ in the limit~$\gamma \to 0$,
for some limiting coefficient $D^*_{\nu}$.
More precisely, it holds
\[
  \abs{\gamma D_{\gamma,\nu} - D^*_{\nu}} \xrightarrow[\nu \to 0]{} 0.
\]
The effective diffusion coefficient is given by
\begin{align}
  \label{eq:effective_diffusion_GLE_underdamped}
  D_{\nu}^*
  = \frac{8 \pi^2}{\beta Z_\beta} \int_{E_0}^{\infty} \frac{\nu^2 \e^{-\beta \mathcal E}}{S_{\nu}(\mathcal E)} \, \d \mathcal E.
\end{align}
Here
\begin{equation}
  \label{eq:S_nu_integral}
  S_\nu(E) = \int_\torus s_\nu(q,E) \, \d q,
\end{equation}
with~$s_\nu(\dummy, E)$, for $E > E_0$,
the unique smooth periodic solution to the following first order differential equation (see \cref{lemma:auxiliary_underdamped} in \cref{sec:estimates_underdamped_limit}):
\begin{align}
  \label{eq:ode_for_s_nu}%
  \forall q \in \torus, \qquad P(q, E) \derivative{1}[s_{\nu}]{q} (q, E) - \frac{1}{\nu^2} s_{\nu}(q, E) = - P(q, E).
\end{align}
\end{result}
\begin{remark}
  We use here the environment \emph{Result}, as opposed to \emph{Proposition} or \emph{Theorem},
  because our proof of the result relies on formal asymptotics,
  and additional work would be required to turn these asymptotics into a rigorous proof.
\end{remark}

The derivation of this result using formal asymptotics is presented at the end of this section.
Note that the integral on the right-hand side of~\eqref{eq:effective_diffusion_GLE_underdamped} is well defined since
\begin{equation}
    \label{eq:proof_bounds_underdamped}
    \forall  q \in \torus, \qquad \nu^2 \sqrt{2(E - E_0)} \leq \nu^2 \inf_{q \in \torus} P(q,E) \leq s_{\nu}(q, E) \leq \nu^2 \sup_{q \in \torus} P(q,E) = \nu^2 \, \sqrt{2E}.
\end{equation}
The latter inequalities are obtained from the ordinary differential equation (ODE)~\eqref{eq:ode_for_s_nu} satisfied by~$s_\nu$, since $\partial_q s_\nu(q,E)$ vanishes at the extrema of $q \mapsto s_{\nu}(q, E)$.
The relationship between~$D_{\nu}^*$ and the diffusion coefficient~$D_{\rm und}$ given by~\cref{eq:effective_underdamped_Langevin} is made precise in the following result.

\begin{proposition}
  There exists $C > 0$ depending only on $V$ such that
  \begin{equation}
    \label{eq:limit_underdamped_GL1}
      \forall \nu > 0, \qquad
      \left| D^*_{\nu} - D_{\rm und} \right| \leq C \nu^4 \, (1 + \nu^2).
  \end{equation}
\end{proposition}

\begin{proof}
If the potential $V$ is constant, then $s_{\nu}(q, E) = \nu^2 P(q, E)$, so $\nu^{-2}S_{\nu}(E) = S_{\rm und}(E)$ and,
consequently, $D^*_{\nu} = D_{\rm und}$.
If $V$ is not constant, then by~\eqref{eq:underdamped_estimate_ode} below, it holds
\begin{align*}
  & \norm{\nu^{-2} s_{\nu}(\dummy, E) - P(\dummy, E) - \nu^2 \, P(\dummy, E) \partial_q P (\dummy, E)}_{\infty} \\
  & \qquad \qquad \leq \nu^4 \left\| P(\dummy,E) \partial_q\left[ P(\dummy,E) \partial_q P(\dummy,E) \right] \right\|_\infty = \nu^4 \left\|P(\dummy, E) V'' \right\|_\infty.
\end{align*}
Since $P(\dummy, E) \partial_q P (\dummy, E)$ integrates to zero over~$\torus$ and $0 \leq P(\dummy, E) \leq \sqrt{2E}$ \revision{by~\eqref{eq:def_S_0}}, there exists $K > 0$ independent of $\nu$ and $E$ such that
\[
  \forall E > E_0, \quad \forall \nu > 0, \qquad \abs{\nu^{-2} S_\nu(E) - S_{\rm und}(E)} \leq K \nu^4 \, \sqrt{E},
\]
in view of the definition~\eqref{eq:def_S_0} of~$S_{\rm und}$.
By~\eqref{eq:bounds_sn} in the appendix (where we use the assumption that $V$ is not constant)
and the fact that $S_{\rm und}(E)$ is an increasing function of $E$ with $S_{\rm und}(E_0) > 0$,
it holds
\begin{align*}
  \forall E > E_0, \quad \forall \nu > 0,\qquad
  \abs{\frac{\nu^2}{S_{\nu}(E)} - \frac{1}{S_{\rm und}(E)}}
  &= \abs{\frac{\nu^2}{S_{\nu}(E) \, S_{\rm und}(E)}} \, \abs{\nu^{-2} S_\nu(E) - S_{\rm und}(E)}  \\
  &\leq \frac{1 + M \nu^2}{\abs{S_{\rm und}(E_0)}^2} K \nu^4 \, \sqrt{E},
\end{align*}
which, given the definitions~\cref{eq:effective_underdamped_Langevin,eq:effective_diffusion_GLE_underdamped}, shows~\eqref{eq:limit_underdamped_GL1}.
\end{proof}

\medskip

We conclude this section by first presenting some numerical results
confirming that~\eqref{eq:limit_underdamped_GL1} \revision{holds and} then motivating \cref{result:underdamped}.

\subsubsection{Numerical illustration}
In order to estimate $D^*_{\nu}$ for $\nu > 0$,
the last integral in~\eqref{eq:effective_diffusion_GLE_underdamped} can be truncated and approximated by numerical quadrature.
This requires to numerically approximate the integrand,
in particular the term $S_{\nu}(E)$, for discrete values of $E$.
To this end,
we employ the function \verb?solve_bvp? from the \emph{SciPy} module \verb?scipy.integrate?.
This function implements a solver for boundary value problems (BVP) using the approach proposed in~\cite{MR1918120},
and we employ it in order to calculate an approximation $\hat{\vect y} = (\hat y_1, \hat y_2)^\t$ of the solution
to the following first order system of ODEs:
\[
  \derivative*{1}[\vect y]{q}(q) =
   \begin{pmatrix}
     \dps \frac{y_1(q)}{\nu^2 \, P(q, E)} - 1 \\
     y_1(q)
   \end{pmatrix} =: \vect f(q), \qquad -\pi \leq q \leq \pi,
\]
subject to the boundary condition
\[
   \vect g(\vect y(-\pi), \vect y(\pi)) :=
   \begin{pmatrix}
     y_1(-\pi) - y_1(\pi) \\
     y_2(-\pi)
   \end{pmatrix}
   = \vect 0.
 \]
The first line in the ODE for $\vect y$ is~\eqref{eq:ode_for_s_nu} divided by~$P$, while the second one corresponds to~\eqref{eq:S_nu_integral}. Since, given $E > E_0$,
the unique exact solution of this problem is given by
\[
  \vect y(q) =
  \begin{pmatrix}
  s_{\nu}(q, E) \\
  \dps \int_{-\pi}^{q} s_{\nu}(Q, E) \, \d Q
  \end{pmatrix},
\]
an approximation of $S_{\nu}(E)$ is obtained by simply evaluating $\hat y_2(\pi)$.

Results from the numerical simulation are illustrated in \cref{fig:underdamped:effective_diffusion}
for the case of the cosine potential $V(q) = \frac{1}{2} \, (1 - \cos q)$,
which was also employed in \cref{par:numerical_results}. Recall that it is difficult to compare $D^*_\nu$ with~\eqref{eq:introduction:effective_diffusion} since the spectral method we use in Section~\ref{par:numerical_results} cannot tackle values of $\gamma$ smaller than~$0.01$, and such values are not sufficiently small to see the asymptotic regime~$\gamma \to 0$ (as evidenced by Figure~\ref{fig:numerics:diffusion_coefficient}, Right).

To generate the results in \cref{fig:underdamped:effective_diffusion},
the integral in~\eqref{eq:effective_diffusion_GLE_underdamped} was truncated at $E = 25$
and approximated using the \emph{SciPy} function \verb?scipy.integrate.quad? with a relative tolerance equal to $10^{-12}$. The tolerance used in \verb?scipy.integrate.solve_bvp? was equal to $10^{-11}$.
The limiting coefficient $D_{\rm und}$ was computed by truncating and approximating the integral in~\eqref{eq:effective_underdamped_Langevin} with the same parameters, and by using the explicit expression of $S(E)$ derived in~\cite{MR2427108}.

We observe that, although $D_{\nu}^*$ does vary with $\nu$,
the relative variation is very small ($< 5\%$ over the interval $\nu \in [0, 1]$).
We also notice that $D_{\nu}^* \to D_{\rm und}$ as $\nu \to 0$, as expected from~\eqref{eq:limit_underdamped_GL1}.
In fact, the difference $D_{\nu}^* -D_{\rm und}$ clearly scales as~$\nu^4$ in the limit $\nu \to 0$,
consistently with~\eqref{eq:limit_underdamped_GL1}.
\begin{figure}[ht]
  \centering
  \includegraphics[width=.8\linewidth]{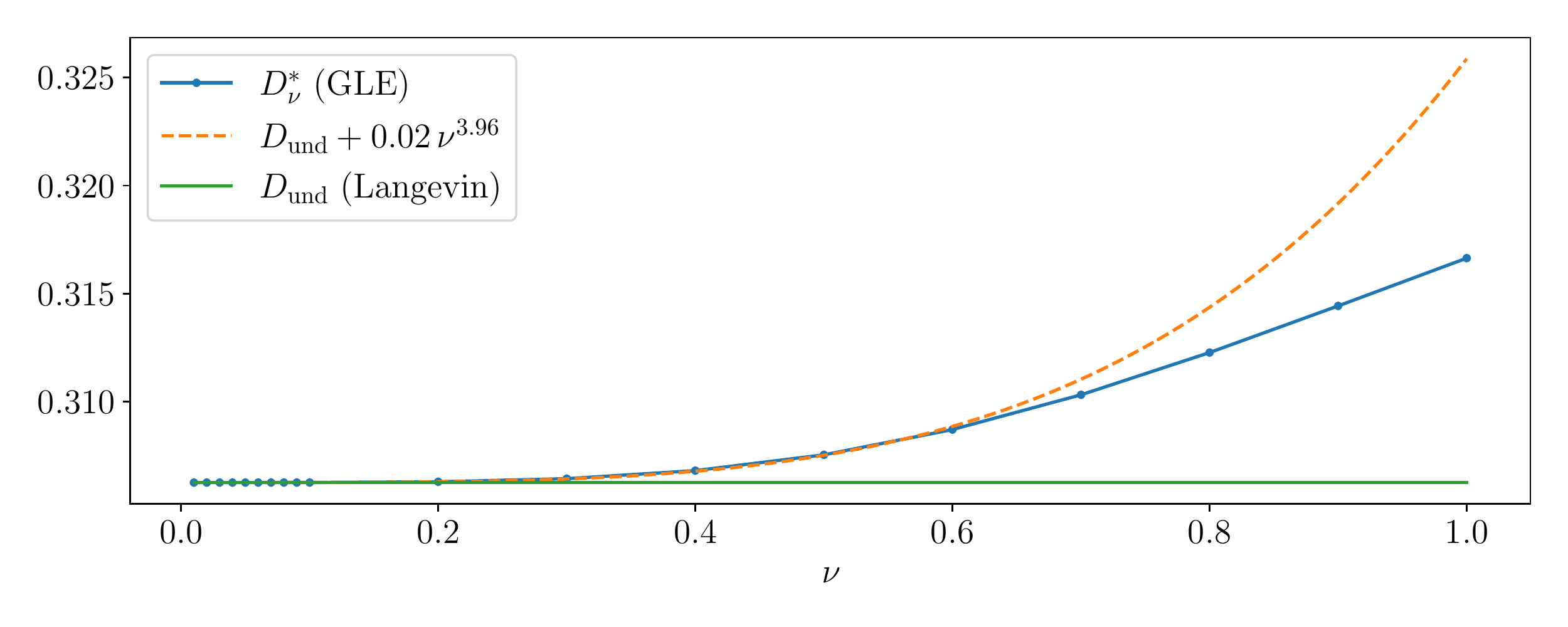}
  \includegraphics[width=.8\linewidth]{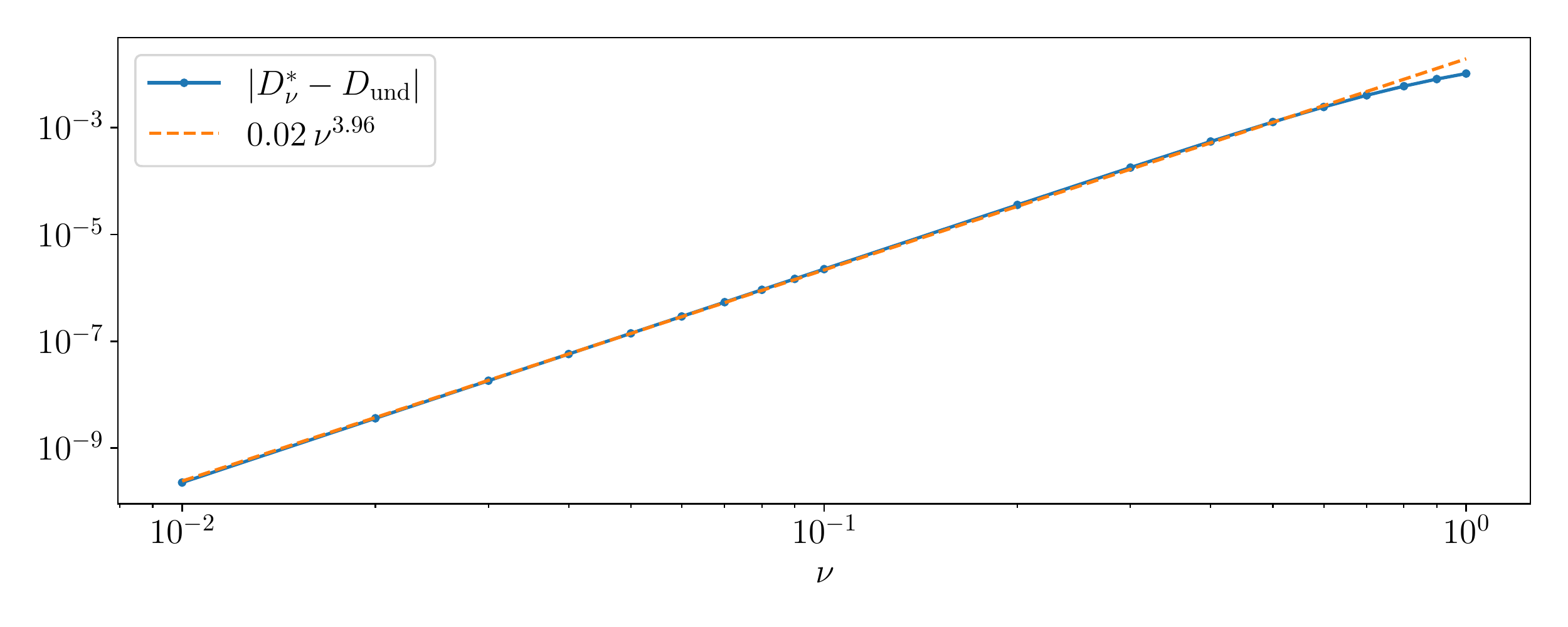}
  \caption{%
    Comparison between the effective diffusion coefficients (multiplied by $\gamma$) for the Langevin and the generalized Langevin dynamics in the underdamped limit, in linear (top) and logarithmic (bottom) scales.
  }%
  \label{fig:underdamped:effective_diffusion}
\end{figure}

%--------------------------------------------------
\subsubsection{Formal derivation of \texorpdfstring{\cref{result:underdamped}}{Result 4.1}}

In order to formally obtain \cref{result:underdamped},
we rewrite the generator as $  \mathcal{L}_\gamma = \mathcal A_{0} + \sqrt{\gamma} \mathcal A_{1}$
(indicating explicitly the dependence on the friction~$\gamma > 0$), with
\[
  \mathcal{A}_0 = p \derivative{1}{q} - \derivative*{1}[V]{q}(q)\derivative{1}{p} + \frac{1}{\nu^2} \left(- z \derivative{1}{z} +  \frac{1}{\beta} \, \derivative{2}{z^2} \right),
  \qquad
  \mathcal{A}_1 = \frac{1}{\nu} \, \left(z \, \derivative{1}{p} - p \, \derivative{1}{z}\right).
\]
We next consider the following ansatz on the solution~$\phi_\gamma$ to the Poisson equation~\eqref{eq:introduction:poisson_equation}, motivated by the fact that the leading order of the resolvent is~$\gamma^{-1}$ by~\eqref{eq:resolvent_bound_GL1}:
\[
  \phi_{\gamma} = \frac{1}{\gamma} \, \overline{\phi}_0 + \frac{1}{\sqrt{\gamma}} \, \overline{\phi}_1 + \overline{\phi}_2 + \dotsb
\]
By substituting into the Poisson equation~\eqref{eq:introduction:poisson_equation} and successively identifying terms of order~$\gamma^{-1},\gamma^{-1/2},1,\dotsb$,
we obtain
\begin{subequations}
\begin{align}
  \label{eq:underdamped:equation_1}
  \mathcal A_0 \overline{\phi}_0 = 0, \\
  \label{eq:underdamped:equation_2}
  \mathcal A_0 \overline{\phi}_1 + \mathcal A_1 \overline{\phi}_0 = 0, \\
  \label{eq:underdamped:equation_3}
  \mathcal A_0 \overline{\phi}_2 + \mathcal A_1 \overline{\phi}_1 = - p;
\end{align}
\end{subequations}
while $\mathcal A_0 \overline{\phi}_{i+1} + \mathcal A_1 \overline{\phi}_i = 0$ for $i \geq 2$.
We motivate below that
\begin{equation}
  \label{eq:phi_und}
  \overline{\phi}_0(q,p,z) = \sign(p) \, \psi_0(H(q,p)),
\end{equation}
where $\sign(p) = \mathbbm 1_{[0, \infty)}(p) - \mathbbm 1_{(-\infty, 0]}(p)$ and
\begin{equation}
  \label{eq:expression_psi0}
  \psi_0(E) = \mathbbm 1_{[E_0, \infty)}(E) \, 2 \pi \nu^2 \int_{E_0}^{E} \frac{1}{S_{\nu}(\mathcal{E})} \, \d \mathcal{E}.
\end{equation}
Unfortunately, $\derivative*{1}[\psi_0]{E}$ is discontinuous at $E = E_0$,
so $\mathcal L_{\gamma} \overline{\phi}_1$ does not make sense as a function.
This invalidates a posteriori the assumed asymptotic expansion~\eqref{eq:proof_bounds_underdamped}. Despite the breakdown of the naive expansion,
which was also noted in~\cite{MR2427108} and~\cite{MR2394704} for the Langevin equation in the underdamped regime,
we assume that $\gamma^{-1} \overline{\phi}_0$ still captures $\phi_{\gamma}$ at dominant order.
% but with a correction possibly scaling differently from~$\gamma^{-1/2}$.
For the Langevin equation,
the correctness of the dominant term in the naive expansion can be shown rigorously based on results by Freidlin and Weber~\cite{MR1722275},
but showing this for GL1 is an open problem that would probably require considerable additional work.

From the above discussion, the effective diffusion coefficient~$D_{\gamma,\nu}$ should scale at dominant order as~$\gamma^{-1} D^*_{\nu}$ with
\begin{align}
  \notag%
  D_{\nu}^*
  & = \int_{\torus \times \real \times \real} \overline{\phi}_0 \, p \, \d \mu = \frac{1}{Z_\beta}\int_{\torus \times \real} \overline{\phi}_0(q,p) p \, \e^{-\beta H(q,p)} \, \d q \, \d p
  = \frac{2}{Z_\beta} \int_{E_0}^{\infty} \int_{\torus} \psi_0(E) \e^{- \beta E} \, \d q \, \d E \\
  \notag%
  & = \frac{4 \pi}{Z_\beta} \int_{E_0}^{\infty} \psi_0(E) \e^{-\beta E} \, \d E
  = \frac{8 \pi^2}{Z_\beta} \int_{E_0}^{\infty} \int_{E_0}^E \frac{\nu^2 \e^{-\beta E}}{S_{\nu}(\mathcal{E})} \,  \d \mathcal{E} \, \d E \\
  \notag%
  & = \frac{8 \pi^2}{Z_\beta} \int_{E_0}^{\infty} \frac{\nu^2}{S_{\nu}(\mathcal{E})} \left( \int_{\mathcal{E}}^{\infty}  \e^{-\beta E}\,  \d E \right) \d \mathcal{E}
  = \frac{8 \pi^2}{\beta Z_\beta} \int_{E_0}^{\infty} \frac{\nu^2 \e^{-\beta \mathcal E}}{S_{\nu}(\mathcal E)} \, \d \mathcal E,
\end{align}
where we used in the first line that the change of variables~$(q,p) \mapsto (q,E)$ has Jacobian~$P(q,E)$.
This concludes the derivation of \cref{result:underdamped}.

\paragraph{Motivation for~\eqref{eq:phi_und}.}
We assume for simplicity,
in addition to \cref{assumption:assumption_potential},
that $V(q)$ is an even function around $q = 0$.
This is not required to obtain the final result,
but it leads to a simplified ansatz for $\overline{\phi}_0$,
which for general potentials can be justified only \emph{a posteriori}.
Under this additional assumption,
one can check by substitution that,
if $\phi_{\gamma}$ denotes the solution to the Poisson equation $-\mathcal L_{\gamma} \phi_{\gamma} = p$,
then $\psi_{\gamma}(q, p, z) = - \phi_{\gamma}(-q, -p, -z)$ is also a solution to the equation,
so $\psi_{\gamma} = \phi_{\gamma}$ by uniqueness of the solution.
Therefore, $\phi_{\gamma}$ and all the summands $\{\overline{\phi}_i\}_{i=0, 1, \dotsc}$ must satisfy the symmetry relation
\begin{equation}
  \label{eq:symmetry}
  u(q, p, z) = -u(-q, -p, -z).
\end{equation}

Multiplying \cref{eq:underdamped:equation_1} by $\overline{\phi}_0$, integrating with respect to~$\mu$,
and taking into account that the contribution of the antisymmetric part of $\mathcal A_0$ vanishes,
we obtain
\begin{equation}
  \label{eq:dz_vanish}
  \frac{1}{\beta}\int_{\torus \times \real \times \real} \left(\derivative{1}[\overline{\phi}_0]{z} \right)^2 \, \mu (\d q \, \d p \, \d z) = 0,
\end{equation}
which shows that $\overline{\phi}_0 = \overline{\phi}_0(q, p)$.
Substituting again in \cref{eq:underdamped:equation_1},
we deduce that $\overline{\phi}_0$ must lie in the kernel of the operator $p \, \derivative{1}{q} - \derivative*{1}[V]{q}(q) \, \derivative{1}{p}$,
which consists of functions that are constant on the contour lines of the Hamiltonian $H(q,p)$.
Together with the symmetry relation~\eqref{eq:symmetry}, and distinguishing closed and open orbits (corresponding respectively to $H(q,p)<E_0$ and $H(q,p)>E_0$), this motivates that
\[
  \overline{\phi}_0(q, p) =
  \begin{cases}
    \psi_0 (H(q,p)), & \text{if } H(q,p) \geq E_0 \text{ and } p \geq 0, \\
    0 & \text{if } H(q,p) < E_0, \\
    - \psi_0 (H(q,p)), & \text{if } H(q,p) \geq E_0 \text{ and } p \leq 0,
  \end{cases}
\]
for some function $\psi_0$ still to be determined.

For the next order we use the ansatz $\overline{\phi}_1(q, p, z) = z \, \psi_1(\sign(p) \, q,H(q,p))$.
We remark that any function in the kernel of $\mathcal A_0$,
\emph{i.e.} any function of only $(q,p)$ that is constant on the contour lines of the Hamiltonian,
could in principle be added to this ansatz.
However, this will not be necessary for our purposes,
because any constant-in-$z$ part of $\overline{\phi}_1$ cancels out in the equation~\eqref{eq:eq_phi0_times_Snu}
for~$\psi_0$ derived below. Restricting our attention first to the region where $H(q,p) \geq E_0$ and $p \geq 0$,
we obtain the following equation for the function~$(q,E)\mapsto\psi_1(q,E)$ from~\cref{eq:underdamped:equation_2}:
\[
  p \derivative{1}[\psi_1]{q} \left(q, V(q) + \frac{p^2}{2}\right)
  - \frac{1}{\nu^2} \psi_1 \left(q, V(q) + \frac{p^2}{2}\right)
  + \frac{1}{\nu} p \, \derivative*{1}[\psi_0]{E}\left(V(q) + \frac{p^2}{2}\right) = 0.
\]
This equation is satisfied pointwise provided
\begin{equation}
  \label{eq:underdamped:equation_for_psi1}
  \forall (q, E) \in \torus \times (E_0, + \infty),
  \qquad P(q, E) \, \derivative{1}[\psi_1]{q} (q, E) - \frac{1}{\nu^2} \psi_1(q, E) + \frac{1}{\nu} \, P(q, E) \, \derivative*{1}[\psi_0]{E}(E) = 0.
\end{equation}
In view of the definition~\eqref{eq:ode_for_s_nu} of $s_{\nu}(q, E)$, it holds
\begin{equation}
  \label{eq:solution_psi1}
  \psi_1 (q, E) = \frac{1}{\nu} s_{\nu}(q, E) \, \derivative*{1}[\psi_0]{E}(E).
\end{equation}
In the region $E < E_0$,
equation \cref{eq:underdamped:equation_2} simplifies to $\mathcal A_0 \overline{\phi}_1 = 0$. We can follow the treatment of~\eqref{eq:dz_vanish}, by taking into account the domain in the~$(q,p)$ variables. This is done by multiplying~\eqref{eq:solution_psi1} by~$\overline{\phi}_1$, integrating over the set $A := \{(q, p, z) \in  \torus \times \real \times \real: H(q,p) \leq E_0\}$
with respect to the Gaussian weight~$g(z) := \sqrt{\frac{\beta}{2\pi}} \, \e^{-\beta z^2/2}$,
and employing the formal antisymmetry of $p \, \derivative{1}{q} - \derivative*{1}[V]{q}(q) \, \derivative{1}{p}$,
which leads to
\begin{equation*}
  \frac{1}{\beta}\int_{A} \left(\derivative{1}[\overline{\phi}_1]{z} \right)^2 \, g(z) \, \d q \, \d p \, \d z = 0,
\end{equation*}
so necessarily $\overline{\phi_1} = \overline{\phi_1}(q,p)$ in that region.
Consequently, $\overline{\phi_1}(q, p)$ must lie in the kernel of the operator $p \, \derivative{1}{q} - \derivative*{1}[V]{q}(q) \, \derivative{1}{p}$,
so it is a function of~$H(q,p)$ only.
By the symmetry relation~\eqref{eq:symmetry}, we obtain that $\revision{\psi_1}(q,E) = 0$ \revision{in that region}.

Substituting the expression of $\overline{\phi_1}$ in~\eqref{eq:underdamped:equation_3}
and integrating in the $z$ direction with respect to the Gaussian weight~$g(z)$,
we obtain
\begin{equation}
  \label{eq:equation_for_phi2_underdamped}
  \left(p \derivative{1}{q} - \derivative*{1}[V]{q}(q) \, \derivative{1}{p} \right) \int_{\real}\overline{\phi}_2(q,p,z) \, g(z) \, \d z
  +  \frac{p}{\nu}\left[\left( \frac{1}{\beta}\derivative{1}{E} - 1 \right) \psi_1\right]\left(q, H(q,p) \right) = - p,
\end{equation}
which can be viewed as a PDE on $\torus \times \real$ for $(q,p) \mapsto \int_{\real}\overline{\phi}_2(q,p,z) \, g(z) \, \d z$.
Since the operator \revision{acting on} this function is formally antisymmetric in $\lp{2}{B}$,
with $B := \{(q,p) \in \torus \times \real: p \geq 0~\text{and}~H(q,p) > E_0\}$,
the solvability condition associated with this equation is
\[
  \int_{B} \left( \frac{p}{\nu} \left[\left( \frac{1}{\beta} \derivative{1}{E} - 1 \right) \psi_1\right]\left( q, H(q,p) \right) + p \right) \, f(H(q,p)) \, \d q \, \d p = 0
\]
for any smooth function $E \mapsto f(E)$ with compact support \revision{in $(E_0,\infty)$}.
Using again the change of variables~$(q,p) \mapsto (q,E)$,
\revision{the latter} equation reads
\[
   \int_{\torus \times (E_0, \infty)} \left( \frac{1}{\nu} \left[\left( \frac{1}{\beta} \derivative{1}{E} - 1 \right) \psi_1\right](q, E) + 1 \right) \, f(E) \, \d q \, \d E = 0.
\]
In order for this equation to be satisfied for any choice of $f$,
it is necessary that
\[
  \forall E > E_0, \qquad
  \int_{\torus} \left[ \frac{1}{\nu} \left( \frac{1}{\beta} \derivative{1}[\psi_1]{E}(q, E) - \psi_1(q, E) \right) + 1 \right] \d q = 0,
\]
which, by substituting the expression of $\psi_1$ given by~\eqref{eq:solution_psi1},
gives
\begin{align}
  0 &= \int_{\torus} \left[\frac{1}{\nu^2} \left( \frac{1}{\beta} \derivative{1}{E} \left(s_{\nu}(q, E) \, \derivative*{1}[\psi_0]{E}(E) \right) - s_{\nu}(q, E) \, \derivative*{1}[\psi_0]{E}(E) \right) + 1 \right] \d q \notag \\
    &= \frac{1}{\nu^2} \left( \frac{1}{\beta}\derivative*{1}{E} \left( \derivative*{1}[\psi_0]{E}(E)  S_{\nu}(E)  \right) - \derivative*{1}[\psi_0]{E}(E) S_{\nu}(E) \right) + 2 \pi.
      \label{eq:eq_phi0_times_Snu}
\end{align}
The latter equation is similar to that obtained for the Langevin dynamics in~\cite{MR2427108}.
Viewed as an ODE for $\derivative*{1}[\psi_0]{E}(E)  S_{\nu}(E)$,
equation \cref{eq:eq_phi0_times_Snu} admits the general solution
\begin{align}
  \label{eq:expression_derivative}
  \derivative*{1}[\psi_0]{E}(E)  S_{\nu}(E) = 2\pi \nu^2 + C \e^{\beta E},
\end{align}
for any constant $C$.
A \revision{necessary} condition for $\overline{\phi}_0$ to belong to~$H^1(\mu)$
is that $C = 0$ and that $\overline{\phi}_0$ be continuous at the homoclinic orbit $H(q, p) = E_0$,
which leads to~\eqref{eq:expression_psi0}. It is in fact possible to prove that $\overline{\phi}_0$ is in~$H^1(\mu)$, see
Appendix~\ref{sec:estimates_underdamped_limit}.

\begin{remark}
The above formal calculations,
which could be replicated at the level of the backward Kolmogorov equation associated with the GL1 dynamics,
suggest that the Hamiltonian of the rescaled process~$H(q^{\gamma}(t/\gamma), p^{\gamma}(t/\gamma))$,
where $(q^{\gamma}, p^{\gamma}, z^{\gamma})$ is the solution to~\eqref{eq:markovian_approximation},
converges weakly to a Markov process on a graph,
with a generator similar to that in the case of underdamped Langevin dynamics; see~\cite{MR1722275,MR2394704,pavliotis2011applied}.
\end{remark}

\section{Conclusions}%
\label{sec:conclusions}
In this work,
we studied quasi-Markovian approximations of the GLE,
and we scrutinized in particular two finite-dimensional models of the noise:
the scalar Ornstein--\revision{Uhlenbeck} noise and the harmonic noise.
For the former model,
we obtained decay estimates with explicit scalings with respect to the parameters,
and we investigated the asymptotic behavior of the associated effective diffusion coefficient in several limits of physical relevance.
We also employed an efficient Fourier/Hermite spectral method to verify most of our findings numerically,
thereby complementing previous numerical works~\cite{igarashi1988non,igarashi1992velocity} on the subject.

Exciting questions remain open for future work.
\begin{itemize}
\item On the theoretical front,
it is not clear whether a direct $\lp{2}{\mu}$ hypocoercivity approach of the type introduced in~\cite{MR2576899,MR3324910} can be applied to the GLE.
If this was the case,
we are hopeful that
the approach could be replicated at the discrete level to obtain a hypocoercivity estimate for the Fourier/Hermite numerical method,
which would enable the calculation of bounds on the numerical error, as done in~\cite{roussel2018spectral} for Langevin dynamics.
It would also be interesting to investigate whether the approach of~\cite{BFLS20},
based on Schur complements,
could be generalized in order to more directly obtain the resolvent bounds~\cref{eq:resolvent_bound_GL1,eq:resolvent_bound_GL2}.
% Finally, it would be interesting to investigate the extent to which our results carry over to higher-dimensional models of the noise.
Finally, it would be interesting to investigate the longtime behavior of more general generalized Langevin equations and,
in particular, the application of hypocoercivity techniques to the case of arbitrary stationary Gaussian noise processes.
In principle, this would require taking into account an infinite number of auxiliary Ornstein-Uhlenbeck processes
and it is related to the problem of stochastic realization theory~\cite{LindquistPicci1985}.

\item On the numerical front,
it would be interesting to investigate carefully the underdamped limit of systems in dimension greater or equal to~2 with Monte Carlo simulations,
which could be made more efficient with the variance reduction technique based on control variates recently developed in~\cite{roussel2017perturbative}.
\end{itemize}

\paragraph{Acknowledgements.}
The authors are grateful to the anonymous referees for their careful reading of our work and their very useful suggestions.
The work of G.S.\ was partially funded by the European Research Council (ERC) under the European Union's Horizon 2020 research and innovation programme (grant agreement No 810367), and by the Agence Nationale de la Recherche under grant ANR-19-CE40-0010-01 (QuAMProcs).
The work of G.P.\ and U.V. was partially funded by EPSRC through grants number EP/P031587/1, EP/L024926/1, EP/L020564/1 and EP/K034154/1.
\revision{The work of U.V.\ was partially funded by the Fondation Sciences Mathématique de Paris (FSMP),
through a postdoctoral fellowship in the ``mathematical interactions'' program.}
The work of G.P.\ was partially funded by JPMorgan Chase \& Co.
Any views or opinions expressed herein are solely those of the authors listed,
and may differ from the views and opinions expressed by JPMorgan Chase \& Co.\ or its affiliates.
This material is not a product of the Research Department of J.P.\ Morgan Securities LLC.
This material does not constitute a solicitation or offer in any jurisdiction.

\appendix
\setcounter{secnumdepth}{4}
\renewcommand*\theparagraph{(\roman{paragraph})}

\section{Confirmation of the rate of convergence in the quadratic case}%
\label{sec:confirmation_of_the_rate_of_convergence_in_the_quadratic_case}

To assess the sharpness of the lower bounds on the convergence rate provided by~\cref{thm:hypocoercivity_h1},
we consider the generalized Langevin dynamics confined by the quadratic potential $V(q) = k q^2/2$ (with $k>0$) in $\cX = \real$.
In this case, \cref{eq:markovian_approximation} can be written as~\eqref{eq:OU_form_GLE} with
\begin{equation}
  \label{eq:matrix_D}
  \mat D = \begin{pmatrix} 0 & 1 & 0 \\ -k & 0 & \dps \frac{\sqrt{\gamma}}{\nu} \\ 0 & \dps -\frac{\sqrt{\gamma}}{\nu} & \dps -\frac{1}{\nu^2} \\\end{pmatrix} = \mat A + \frac{\sqrt{\gamma}}{\nu} \mat B + \frac{1}{\nu^2} \mat C,
\end{equation}
with
\[
\mat A = \begin{pmatrix} 0 & 1 & 0 \\ -k & 0 & 0 \\ 0 & 0 & 0 \\\end{pmatrix}, \qquad \mat B = \begin{pmatrix} 0 & 0 & 0 \\ 0 & 0 & 1 \\ 0 & -1 & 0 \\\end{pmatrix}, \qquad \mat C = \begin{pmatrix} 0 & 0 & 0 \\ 0 & 0 & 0 \\ 0 & 0 & -1 \\\end{pmatrix}.
\]
In order to determine the elements in~\eqref{eq:sigma_OU_form_GLE}, it suffices to compute the three eigenvalues of~$\mat D$, and combine them linearly with positive coefficients. This means that the spectral bound corresponds to the smallest absolute value of the real parts of the eigenvalues of~$\mat D$. For the model GL1, the eigenvalues $x_j$ (for $1 \leq j \leq 3$) are the roots of the polynomial $p(x) = \mathrm{det}(x \mat I - \mat D)$, which reads
\begin{equation}
  \label{eq:charact_pol}
  p(x) = x^3 + \frac{x^2}{\nu^2} + \left(\frac{\gamma}{\nu^2} + k \right)x + \frac{k}{\nu^2}.
\end{equation}

The spectral gap is easily seen to be a continuous function of~$(\gamma,\nu) \in (0,+\infty) \times (0,+\infty)$, with positive values.
To obtain the scaling of the spectral gap as a function of the parameters, it therefore suffices to consider the various limiting regimes where at least one of the parameters goes to~0 or~$+\infty$. The expression~\eqref{eq:matrix_D} suggests that the eigenvalues can be expanded in series of (inverse or fractional) powers of~$\gamma$ and~$\nu$. The leading order term depends on the asymptotic regime which is considered. We start by investigating regimes where only one of the parameters goes to~0 or~$+\infty$, and then discuss regimes where both parameters are varied.
The organization of this appendix is illustrated in~\cref{figure:limits_appendix}.
\begin{figure}[ht!]
  \centering
  \resizebox{.94\textwidth}{!}{%
  \begin{tikzpicture}[xscale=14.5,yscale=8]
    \draw[black,thick] (.7,.7) -- (1,1);
    \draw[black,thick] (.7,.3) -- (1,0);
    \node[rotate=35,black] at (.8, .76) {$\nu^2 = \gamma$};
    \node[rotate=-30,black] at (.8, .25) {$\nu^2 = 1/\gamma$};

    \draw[thick,->] (0,0) -- (1.05,0);
    \draw[-] (0,0) -- (1,0);
    \draw[thick,->] (0,0) -- (0,1.1);
    \draw[dashed,thick] (1,0) -- (1,1);
    \draw[dashed,thick] (0,1) -- (1,1);
    \node[below left] at (0, 0) {$0$};
    \node[below] at (1, 0) {$1$};
    \node[left] at (0, 1) {$1$};
    \node[below] at (1.05, -.01) {$\dps \frac{\gamma}{1 + \gamma}$};
    \node[left] at (0, 1.1) {$\dps \frac{\nu^2}{1 + \nu^2}$};

    \node at (.5, .94) {\cref{par:limit_nu_to_infty_with_gamma_fixed_}: $\dps \lambda \approx \frac{\gamma}{2k\nu^4}$};
    \node[rotate=90] at (.96, .5) {\cref{par:limit_gamma_to_infty_for_nu_0_fixed_}: $\dps \lambda \approx \frac{k}{\gamma}$};
    \node[rotate=90] at (.04, .5) {\cref{par:limit_gamma_to_0_with_nu_fixed_}: $\dps \lambda \approx \frac{\gamma}{2(k\nu^4 + 1)}$};
    \node at (.5, .08) {\cref{par:limit_nu_to_0_for_gamma_fixed_}: $\dps \lambda \approx \frac{\gamma}{2}\left(1-\mathbbm{1}_{\gamma^2 > 4k}\sqrt{1-\frac{4k}{\gamma^2}}\right)$};
    \node at (.03, .95) {\cref{par:joint_limits_gammato0_and_nuto_infty_or_nu_to_infty_with_gamma_nu_2_to_0_}};
    \node at (.9, .96) {\cref{par:joint_limits_gammato0_and_nuto_infty_or_nu_to_infty_with_gamma_nu_2_to_0_}};
    \node at (.96, .1) {\cref{par:joint_limits_gammato_infty_with_gamma_nu_2_to_infty_and_gamma_nu_2_to_infty_}};
    \node at (.96, .9) {\cref{par:joint_limits_gammato_infty_with_gamma_nu_2_to_infty_and_gamma_nu_2_to_infty_}};
    \node at (.88, .04) {\cref{par:joint_limits_nu_to_0_and_gamma_to_infty_with_gamma_nu_2_to_0_}};
    \node at (.04, .05) {\cref{par:limit_g_to_0_and_nu_to_0_}};
  \end{tikzpicture}}
  \caption{%
      Illustration of the different limits considered in this appendix.
      The joint limit $\gamma \to \infty$ and $\nu \to \infty$ with $\gamma / \nu^2$ bounded from above and below,
      as well as the joint limit $\gamma \to \infty$ and $\nu \to 0$ with $\gamma \nu^2$ bounded from above and below,
      are omitted for conciseness.
  }%
  \label{figure:limits_appendix}
\end{figure}
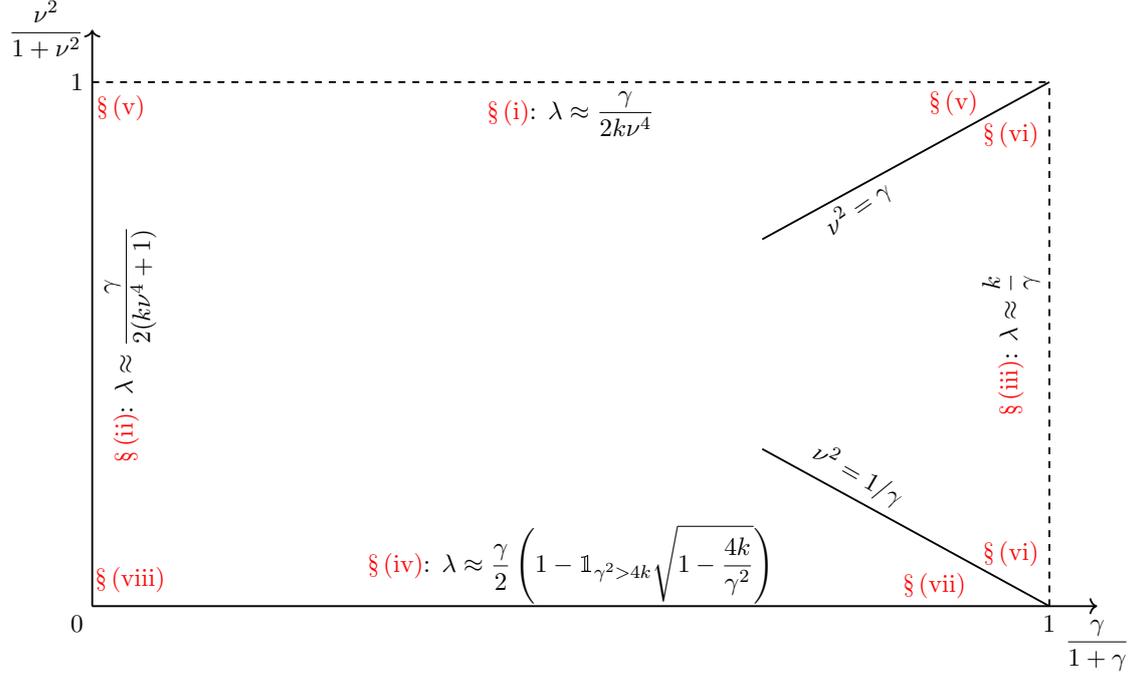

\paragraph{Limit $\nu \to +\infty$ with $\gamma$ fixed.}%
\label{par:limit_nu_to_infty_with_gamma_fixed_}
We denote by~$\varepsilon = \nu^{-1}$ and rewrite~\eqref{eq:matrix_D} as $\mat A + \eps \sqrt{\gamma}\mat B + \eps^2 \mat C$. The matrix~$\mat A$ is diagonalizable and its eigenvalues~0 and~$\pm \ri \sqrt{k}$ are isolated and non-degenerate. Perturbation theory~\cite[Chapter~II]{Kato} then shows that the eigenvalues~$x_j(\eps)$ are analytic functions of~$\eps$ for~$\eps$ sufficiently small. We write
\begin{equation}
  \label{eq:analytic_expansion_x_j}
x_j(\eps) = x_j^0 + x_j^1 \eps + x_j^2 \eps^2 + \dots
\end{equation}
To identify the coefficients~$x_j^k$, we plug the above expansion in the characteristic polynomial, which reads
\begin{equation}
  \label{eq:charact_pol_reference}
  p(x) = x^3 + \eps^2 x^2 + \left(\gamma\eps^2 + k \right)x + k \eps^2,
\end{equation}
and identify terms with the same powers of~$\eps$ in~$p(x_j(\eps) ) = 0 $. Some straightforward computations show that~$x_j^{2n+1} = 0$ (in fact, the expansion in~\eqref{eq:analytic_expansion_x_j} could be immediately restricted to even powers of~$\eps$ by~\cite[Theorem~XII.2]{ReedSimon4} since the coefficients of the polynomial are analytic in~$\eps^2$), and
\[
  x_j^2 = - \frac{k+\gamma x_j^0+ (x_j^0)^2}{3(x_j^0)^2 + k},
  \qquad
  x_j^4 = - \frac{(\gamma+2x_j^0)x_j^2+ 3x_j^0 (x_j^2)^2}{3(x_j^0)^2 + k},
\]
\revision{so}
\begin{equation}
  \label{eq:x_j_nu_infty}
  x_1(\eps) = -\eps^2 + \bigo{\eps^4},
  \qquad
  x_2(\eps) = \ri \sqrt{k}\left(1 + \frac{\gamma}{2k} \eps^2 - \frac{\gamma^2}{8k^2} \eps^4\right) - \frac{\gamma}{2k}\eps^4 + \bigo{\eps^6},
\end{equation}
and a similar expression for~$x_3(\eps)$ with an imaginary part of opposite sign. This shows that the spectral gap~$\lambda$ is $\gamma/(2k\nu^4) + \bigo{\nu^{-6}}$ as $\nu \to +\infty$ with~$\gamma>0$ fixed.

\paragraph{Limit $\gamma \to 0$ with $\nu$ fixed.}%
\label{par:limit_gamma_to_0_with_nu_fixed_}
We denote by~$\varepsilon = \sqrt{\gamma}$ and rewrite~\eqref{eq:matrix_D} as $\mat A + \nu^{-2}\mat C + \eps \nu^{-1}\mat B$. The dominant part~$\mat A + \nu^{-2} \mat C$ is diagonalizable with isolated and non-degenerate eigenvalues~$x_1^0 = -\nu^{-2}$, $x_2^0 = \ri \sqrt{k}$ and~$x_3^0 = -\ri \sqrt{k}$. The eigenvalues~$x_j(\eps)$ of~\eqref{eq:matrix_D} are therefore analytic in~$\eps$ for~$\eps$ small enough, and an expansion of the form~\eqref{eq:analytic_expansion_x_j} also holds in this case. By an analysis similar to the one leading to~\eqref{eq:x_j_nu_infty}, we obtain $x_j^1 = x_j^3 = 0$ and $x_j^2 = -x_j^0/(3\nu^2(x_j^0)^2+2x_j^0+k\nu^2)$, so
\begin{equation}
  \label{eq:x_gamma_0}
  x_1(\eps) = -\frac{1}{\nu^2} + \bigo{\gamma}, % -\frac{1}{1+k\nu^4} \gamma
  \qquad
  x_2(\eps) = \ri\sqrt{k}-\frac{1-\ri \sqrt{k} \nu^2}{2(1+k\nu^4)} \gamma + \bigo{\gamma^2},
\end{equation}
and a similar expression for~$x_3(\eps)$ with an imaginary part of opposite sign. This shows that the spectral gap~$\lambda$ is $\gamma/(2(1+k\nu^4)) + \bigo{\gamma^{2}}$ as $\gamma \to 0$ with~$\nu>0$ fixed.

\paragraph{Limit $\gamma \to +\infty$ for $\nu>0$ fixed.}%
\label{par:limit_gamma_to_infty_for_nu_0_fixed_}
In order to treat the situation when $\gamma$ diverges, we introduce $\eps = \gamma^{-1/2}$ and rewrite~\eqref{eq:matrix_D} as $\eps^{-1} \left[ \nu^{-1}\mat B + \eps \left(\mat A + \nu^{-2}\mat C\right) \right]$. The eigenvalues~$x_j(\eps)$ of~$\mat D$ are then obtained by rescaling the eigenvalues~$y_j(\eps)$ of $\nu^{-1}\mat B + \eps \left(\mat A + \nu^{-2}\mat C\right)$ as $x_j(\eps) = y_j(\eps)/\eps$. The eigenproblem associated with $\nu^{-1}\mat B + \eps \left(\mat A + \nu^{-2}\mat C\right)$ has the form of a standard perturbation problem, with associated characteristic polynomial %P(y) = \eps^3 p(y/eps)
\[
  P(y) = y^3 + \eps\nu^{-2}y^2 + (\nu^{-2}+k\eps^2)y + k\nu^{-2}\eps^3.
\]
The dominant part~$\nu^{-1} \mat B$ is diagonalizable with isolated and non-degenerate eigenvalues~$y_1^0 = 0$, $y_2^0 = \ri/\nu$ and~$y_3^0 = -\ri/\nu$. The eigenvalues~$y_j(\eps)$ of~ $\nu^{-1}\mat B + \eps \left(\mat A + \nu^{-2}\mat C\right)$ are therefore again analytic in~$\eps$ for~$\eps$ small enough, and an expansion of the form~\eqref{eq:analytic_expansion_x_j} still holds, with~$x$ replaced by~$y$. By gathering terms with the same powers of~$\eps$ in~$P(y_j(\eps)) = 0$, we find after some computations that
\[
y_1(\eps) = -k \gamma^{-3/2}+ \bigo{\frac{1}{\gamma^{5/2}}}, \qquad y_2(\eps) = \frac{\ri}{\nu} - \frac{1}{2\nu^2 \sqrt{\gamma}} + \bigo{\frac1\gamma},
\]
and a similar expression for~$y_3(\eps)$ with an imaginary part of opposite sign. This shows that the spectral gap~$\lambda$ is $k/\gamma + \bigo{\gamma^{-2}}$ as $\gamma \to +\infty$ with~$\nu>0$ fixed.

\paragraph{Limit $\nu \to 0$ for $\gamma$ fixed.}%
\label{par:limit_nu_to_0_for_gamma_fixed_}
We introduce here $\eps = \nu$ and rewrite~\eqref{eq:matrix_D} as $\eps^{-2} \left[ \mat C + \eps \sqrt{\gamma} \mat B + \eps^2 \mat A \right]$. As in the previous situation, we look for the eigenvalues~$y_j(\eps)$ of~$\mat C + \eps \sqrt{\gamma} \mat B + \eps^2 \mat A$, whose characteristic polynomial reads % P(y) = \eps^6 p(y/\eps^2)
\begin{equation}
  \label{eq:P_y_nu_0}
P(y) = y^3 + y^2 + (\gamma\eps^2 + k\eps^4)y + k\eps^4.
\end{equation}
The dominant part~$\mat C$ of the eigenvalue problem is diagonalizable with eigenvalues~$y_1^0 = -1$ and~$y_2^0 = y_3^0 = 0$. Since the latter eigenvalue is doubly degenerate, the results of~\cite[Chapter~II]{Kato} show, by reducing the eigenvalue problem to the subspace associated with~$y_2(\eps)$ and~$y_3(\eps)$, that
% Section II.2.3 in particular
$y_1(\eps)$ is analytic in~$\eps$ for~$\eps$ sufficiently small, while~$y_2(\eps)$ and~$y_3(\eps)$ are analytic in~$\sqrt{\eps}$:
\begin{equation}
  \label{eq:y_2_nu_0}
y_2(\eps) = y_2^{1/2} \sqrt{\eps} + y_2^1 \eps + y_2^{3/2} \eps^{3/2} + \dots,
\end{equation}
and a similar expansion for~$y_3(\eps)$. Note that $y_1(\eps) = -1 + \bigo{\eps}$, so the spectral gap is determined by the lowest order terms in the expansions of~$\lambda_2(\eps)$ and~$\lambda_3(\eps)$. By plugging~\eqref{eq:y_2_nu_0} into~\eqref{eq:P_y_nu_0} we find that the first non-zero terms in the expansions satisfy
\[
(y_j^2)^2 + \gamma y_j^2 + k = 0, \qquad y_j^{5/2}(2y_j^2+\gamma) = 0.
\]
We need to distinguish three situations at this stage:
\begin{enumerate}[(i)]
\item When $\gamma^2 > 4k$, one finds $y_2^2 = -\frac\gamma2 \left(1+\sqrt{1-\frac{4k}{\gamma^2}}\right)$ and $y_3^2 = -\frac\gamma2 \left(1-\sqrt{1-\frac{4k}{\gamma^2}}\right)$, so~$y_j^{5/2} = 0$. In this case the spectral gap of the rescaled matrix $\mat C + \eps \sqrt{\gamma} \mat B + \eps^2 \mat A$ scales as~$\frac\gamma2 \left(1-\sqrt{1-\frac{4k}{\gamma^2}}\right) + \bigo{\nu}$ as $\nu \to 0$ with~$\gamma>0$ fixed.
\item When $\gamma^2 < 4k$, one finds $y_2^2 = -\frac\gamma2 \left(1-\ri \sqrt{\frac{4k}{\gamma^2}-1}\right)$ and $y_3^2 = -\frac\gamma2 \left(1+\ri \sqrt{\frac{4k}{\gamma^2}-1}\right)$, so~$y_j^{5/2} = 0$. In this case the spectral gap of the rescaled matrix scales as~$\frac{\gamma}{2}  + \bigo{\nu}$ as $\nu \to 0$ with~$\gamma>0$ fixed.
\item When $\gamma^2 = 4k$, one finds $y_2^2 = y_3^2 = -\frac\gamma2$, and actually $y_j^{5/2} = 0$ (because of the next order condition in~$\eps$ which reads $(y_j^{5/2})^2+(\gamma+2y_j^2)y_j^3 = 0$). The spectral gap of the rescaled matrix scales as in the previous case as~$\frac{\gamma}{2} + \bigo{\nu}$ as $\nu \to 0$.
\end{enumerate}
In conclusion, the spectral gap~$\lambda$ of~$\mat D$ scales as $\frac{\gamma}{2}\left(1-\mathbbm{1}_{\gamma^2 > 4k}\sqrt{1-\frac{4k}{\gamma^2}}\right)+ \bigo{\nu^{3}}$.

\paragraph{Joint limits $\gamma\to0$ and $\nu\to+\infty$, or $\nu \to +\infty$ with $\gamma/\nu^2 \to 0$.}%
\label{par:joint_limits_gammato0_and_nuto_infty_or_nu_to_infty_with_gamma_nu_2_to_0_}
We denote by $\eta = \frac{\sqrt{\gamma}}{\nu}$ and $\varepsilon = \frac{1}{\nu^2}$. These two parameters are small in the asymptotic regime which is considered. The eigenvalues of the leading order matrix~$\mat A$ in~\eqref{eq:matrix_D} are isolated,
so, by perturbation theory (upon writing the spectral projector using contour integrals and expanding the resolvents which appear), we can write the eigenvalues of~$\mat D = \mat A + \eta \mat B + \eps \mat C$ as
\begin{equation}
  \label{eq:x_j_double_limit}
x_j(\eta,\eps) = x_j^0 + x_j^{1,0} \eta + x_j^{0,1} \eps + x_j^{2,0} \eta^2 + x_j^{1,1} \eta\eps + x_j^{0,2} \eps^2 + \dots,
\end{equation}
with $x_1^0 = 0$, $x_2^0 = \ri \sqrt{k}$ and $x_3^0 = -\ri \sqrt{k}$. We next identify terms with the same powers of~$\eta,\eps$ in~$p(x_j(\eta,\eps) )=0$ with~$p(x) = x^3 + \eps x^2 + (\eta^2+k)x+k\eps$. Some straightforward computations show that
\[
  x_1(\eta,\eps) = -\eps + \frac{1}{k} \eta^2 \eps + \bigo{\eta^2\left(\eps^2 + \eta^2\right)},
  \qquad
  x_2(\eta,\eps) = \ri \left(\sqrt{k} + \frac{\eta^2}{2\sqrt{k}}\right) - \frac{1}{2k} \eta^2\eps + \bigo{\eps^4 + \eta^4},
\]
and a similar expansion for~$x_3(\eta,\eps)$ with imaginary parts of opposite signs. Note that the expressions of~$x_1,x_2$ coincide with~\eqref{eq:x_j_nu_infty}, as well with~\eqref{eq:x_gamma_0} in the limit~$\nu \to +\infty$. It can be shown that~$x_j^{2m,0} \in \ri \mathbb{R}$ and~$x_j^{2m+1,0} = 0$ for any~$m \geq 1$ and $1 \leq j \leq 3$ (it suffices in fact to set~$\eps=0$ to identify these coefficients, which amounts to considering the polynomial~$x(x^2+\eta^2+k)$, whose roots are~0 and $\pm\ri \sqrt{k+\eta^2}$). Similarly, $x_j^{0,m} = 0$ for $m \geq 2$ and $1 \leq j \leq 3$, as seen by setting $\eta = 0$ and factorizing~$p(x)$ as $(x+\varepsilon)(x^2+k)$. The spectral gap is therefore $\dps \frac{\gamma}{2k \nu^4} + \bigo{\eps^2 \eta^2}$.

\paragraph{Joint limits $\gamma\to +\infty$ with $\gamma/\nu^2 \to +\infty$ and $\gamma \nu^2 \to +\infty$.}%
\label{par:joint_limits_gammato_infty_with_gamma_nu_2_to_infty_and_gamma_nu_2_to_infty_}
We denote by $\eta = \nu/\sqrt{\gamma}$ and $\varepsilon = 1/(\nu\sqrt{\gamma})$ the two small parameters in the asymptotic regime considered here. We write~$\mat D$ as~$\eta^{-1}\left(\mat B + \eta \mat A + \eps \mat C\right)$, the characteristic polynomial associated with~$\mat B + \eta \mat A + \eps \mat C$ being $P(y) = y^3 + \eps y^2 + (k\eta^2 +1)y + k\eps\eta^2$. We can use an argument similar to the one used to write~\eqref{eq:x_j_double_limit} with $y_1^0 = 0$, $y_2^0 = \ri$ and~$y_3^0 = -\ri$ to obtain the following expansion for the zeros of the polynomial~$P(y)$:
\[
  y_1(\eta,\eps) = - k\eta^2\eps + \bigo{\eta^4 + \eta^4},
  \qquad
  y_2(\eta,\eps) = \ri - \frac{\eps}{2} + \bigo{\eps^2 + \eta^2},
\]
and a similar expansion for~$y_3(\eta,\eps)$ with imaginary parts of opposite signs.
Note that it can be shown that the remainders for~$y_1(\eta,\eps)$ involve only powers of~$\eta^2$ (since the polynomial itself is analytic in~$\eta^2$, see~\cite[Theorem~XII.2]{ReedSimon4})
and there are no remainders of the form~$\eps^n$ or~$\eta^n$ (as seen by setting respectively $\eta=0$ and $\varepsilon=0$ in the expression of~$P(y)$).
For~$y_2$, it can similarly be shown that~$y_2^{2m,0} \in \ri \mathbb{R}$ and~$y_2^{2m+1,0} = 0$ for any~$m \geq 1$.
The spectral gap is therefore $\dps k\eta\eps + \bigo{\eta \eps^2}$ with $k\eta\eps = k/\gamma$,
which coincides with the limit obtained for~$\gamma\to+\infty$ with~$\nu$ fixed.

\paragraph{Joint limits $\nu \to 0$ and $\gamma \to +\infty$ with $\gamma \nu^2 \to 0$.}%
\label{par:joint_limits_nu_to_0_and_gamma_to_infty_with_gamma_nu_2_to_0_}
It is difficult to rely on a perturbative approach here. Indeed, denoting by $\eta = \nu\sqrt{\gamma}$ and $\varepsilon = \nu^2$ the two small parameters in the asymptotic regime considered here, we write~$\mat D$ as~$\eps^{-1}\left(\mat C + \eta \mat B + \eps \mat A\right)$. The dominant term~$\mat C$ however has degenerate eigenvalues,
so it is not clear whether the eigenvalues can be expanded as in~\eqref{eq:x_j_double_limit} (in fact, it is in general not true that this is the case, see~\cite[Section~II.5.7]{Kato}). %,
%\noindent
%{\bf There is a difficulty here since one eigenvalue is degenerate... Maybe one needs to formally obtain better approx of eigenvalues and eigenvectors to start with in order to do a smaller perturbation around already splitted eigenvalues... i.e. find matrices~$U_{\eta,\eps}$ (by formal expansions of the eigenvectors and eigenvalues) such that $U_{\eta,\eps}^{-1}\left(\mat C + \eta \mat B + \eps \mat A\right)U_{\eta,\eps}$ is diagonal with entries $\widehat{x}_j(\eta,\eps)$ up to terms of higher orders in $\eps,\eta$ -- sufficiently high so that standard perturbation theory applies. This is doable but will require some work... Do we want to go into this trouble, or do we simply mention that it would be possible?}
%{\bf Other possibility: au lieu de travailler sur les matrices, de travailler sur le polynome, en l'ecrivant $P(y) = (y - Y_1(\eps,\eta))(y-Y_2)(y-Y_3)$ + something small in powers of $\eta^2,\eps^2$}
We use for this case another argument, based on a localization of the zeros of the characteristic polynomial~$P(y) = y^3 + y^2 + (k\eps^2 + \eta^2)y + k\eps^2$ associated with~$\mat C + \eta \mat B + \eps \mat A$. We first note that the discriminant of the third order polynomial~$P$ is
\[
  % \left(k\eps^2 + \eta^2\right)^2 \left[1- 4 \left(k\eps^2 + \eta^2\right)\right] - k\eps^2 \left(4 + 27 k\eps^2\right) + 18\left(k\eps^2 + \eta^2\right) k\eps^2
  \left(k\eps^2 + \eta^2\right)^2 \left[1- 4 \left(k\eps^2 + \eta^2\right)\right] - \left(4 + 9 k\eps^2 - 18 \eta^2\right) k\eps^2
  = -4k\eps^2 + \eta^4 + \bigo{\eps^4 + \eta^2 \eps^2 + \eta^6}.
\]
Since $\eta^4/\eps^2 = \gamma^2 \to +\infty$, the discriminant is positive in the limiting regime we consider,
so that~$P$ admits three real zeros. The polynomial~$p$ in~\eqref{eq:charact_pol_reference} therefore also admits three real zeros. We next compute
\[
p\left(-\frac{1}{\nu^2} + \delta \gamma \right) = \frac{(\delta-1)\gamma}{\nu^4} + \delta\gamma \left( \frac{(1-2\delta)\gamma}{\nu^2} + \delta^2\gamma^2 + k\right).
\]
When $\delta = 1$, the right-hand side of the above equality is $-\gamma^2/\nu^2<0$ at dominant order, while, for $\delta = 1+\eps$ with a fixed small parameter~$\eps > 0$, the right-hand side scales as $\eps\gamma/\nu^4 > 0$. This shows that one of the roots is in the interval~$[-\nu^{-2}+\gamma,-\nu^{-2}+(1+\eps)\gamma]$ for~$\gamma$ large enough and $\nu,\gamma\nu^2$ sufficiently small. To localize the second root, we consider
\[
\frac{\nu^2}{\gamma^2} p\left(-(1 + \delta)\gamma\right) = \delta(1+\delta) - \gamma\nu^2(1+\delta)^3-\frac{k\nu^2}{\gamma}(1+\delta) + \frac{k}{\gamma^2}.
\]
The right-hand side converges to~$\delta(1+\delta)$ as $\gamma\to+\infty$ and~$\nu\to 0$ with $\gamma\nu^2\to 0$, which shows that, for any $\varepsilon\in (0,1]$, the second root is in the interval~$[-(1+\eps)\gamma,-(1-\eps)\gamma]$ for~$\gamma$ large enough and $\nu,\gamma\nu^2$ sufficiently small .
Finally,
\[
  \begin{aligned}
    \nu^2 p\left(-\frac{k}{\gamma}-\frac{(1+\delta)k^2}{\gamma^3}\right)
    = -\delta\frac{k^2}{\gamma^2} & - \frac{k\nu^2}{\gamma}\left(k +\frac{(1+\delta)k^2}{\gamma^2}\right)
    - \frac{\nu^2}{\gamma^3}\left(k +\frac{(1+\delta)k^2}{\gamma^2}\right)^3 \\
    & + \frac{(1+\delta)k^3}{\gamma^4}\left(2 +\frac{(1+\delta)k}{\gamma^2}\right).
\end{aligned}
\]
The right-hand side scales as $-\delta k^2/\gamma^2$ as $\gamma\to+\infty$ and~$\nu\to 0$ with $\gamma\nu^2\to 0$, which shows that, for any $\varepsilon\in (0,1]$, the third root is in the interval~$[-k/\gamma-(1+\eps)k^2/\gamma^2,-k/\gamma-(1-\eps)k^2/\gamma^2]$ for~$\gamma$ large enough and $\nu,\gamma\nu^2$ sufficiently small. The spectral gap is dictated by the location of this eigenvalue, and scales therefore as $k/\gamma$ in the limit which is considered here.

\paragraph{Limit $\gamma \to 0$ and $\nu \to 0$.}%
\label{par:limit_g_to_0_and_nu_to_0_}
Here also, it is difficult to rely on a perturbative approach,
so we use an alternative method.
Employing the same \revision{reasoning} as in \cref{par:joint_limits_nu_to_0_and_gamma_to_infty_with_gamma_nu_2_to_0_},
it is possible to show
that the \revision{characteristic} polynomial~$p$ admits only one real root in this limit,
\revision{denoted by $x_1$}, and that this root is, for $\eps \in (0,1]$ fixed, in the interval $[- \nu^{-2} + (1-\eps)\gamma, - \nu^{-2} + (1+\eps) \gamma]$ provided that $\gamma \nu^2$ is sufficiently small.
Inspired by~\eqref{eq:x_gamma_0},
we show that the other two, complex conjugate roots,
which we denote by $x^{\pm}$, with the superscript indicating the sign of the imaginary part,
are close to $\hat x^{\pm} := - \gamma/2 \pm \ri \sqrt{k}$.
To this end, we calculate
\[
  p(\hat x^{\pm}) = k \gamma \pm \frac{3}{4} \ri \sqrt{k} \gamma^{2} - \frac{\gamma^{3}}{8} - \frac{\gamma^{2}}{4 \nu^{2}}.
\]
Given the factorization $p(x) = (x - x_1)(x - x^+)(x - x^-)$,
this implies
\begin{equation}
    \label{eq:bound_xplus_xminus}
     (\hat x^{\pm} - x^{+}) (\hat x^{\pm} - x^{-}) = \frac{\dps k \gamma \pm \frac{3}{4} \ri \sqrt{k} \gamma^{2} - \frac{\gamma^{3}}{8} - \frac{\gamma^{2}}{4 \nu^{2}}}{\hat x^{\pm} - x_1}.
\end{equation}
Since $|\hat x^{\pm} - x_1|$ scales as $\nu^{-2}$,
the modulus of the right-hand side scales as $k \gamma \nu^2 - \frac{\gamma^2}{4}$ in the limit as $\gamma \to 0$ and $\nu \to 0$.
This implies
\[
    |\hat x^{\pm} - x^{\pm}|^2 \leq |\hat x^{\pm} - x^{+}| \, |\hat x^{\pm} - x^{-}| = \bigo {\gamma \nu^2 + \gamma^2},
\]
so $|\hat x^{\pm} - x^{\pm}| = \bigo {\sqrt{\gamma} \nu + \gamma}$.
Therefore, $|\hat x^{\pm} - x^{\mp}|$ converges to~$2\sqrt{k}$ in the limit~$\gamma,\nu \to 0$,
so~\eqref{eq:bound_xplus_xminus} implies that in fact $|\hat x^{\pm} - x^{\pm}| = \bigo {\gamma \nu^2 + \gamma^2}$, which is negligible in front of~$\gamma/2 = -\mathrm{Re}(\hat x^{\pm})$.
In conclusion,
the spectral gap scales as $\gamma/2$ in the limit considered here.

\begin{remark}
\revision{
The scaling of the exponential decay rate can also be obtained from explicit expressions for the roots of the characteristic polynomial,
as we illustrate below in the particular case of the limit~$\nu \to \infty$ with fixed $\gamma$.
For simplicity, we \revision{consider} $k = \gamma = 1$.
In this case, the characteristic polynomial is
\begin{equation*}
  p(x) = \frac{1}{\nu^2} \left( \nu^2 x^3 + x^2 + \left(1 +\nu^2 \right)x + 1 \right).
\end{equation*}
Letting $\lambda = \nu^2$,
we define
\[
  q(x) = \lambda x^3 + x^2 + \left(1 + \lambda \right)x + 1
  =: ax^3 + b x^2 + cx + d.
\]
Using the standard notation for expressing the roots of a cubic polynomial,
we define
\begin{align*}
  \Delta_0 &= b^2 - 3 ac = 1 - 3 \lambda ( 1 + \lambda) = -3 \lambda^2 - 3 \lambda + 1, \\
  \Delta_1 &= 2 b^3 - 9 abc + 27 a^2 d = 2 - 9 \lambda (1 + \lambda) + 27 \lambda^2 = 18 \lambda^2 - 9 \lambda + 2,
\end{align*}
and
\begin{align*}
  C = \sqrt[3]{\frac{\Delta_1 + \sqrt{\Delta_1^2 - 4 \Delta_0^3}}{2}}
  &= \sqrt[3]{\frac{18 \lambda^2 - 9 \lambda + 2 + \sqrt{108 \lambda^{6} + 324 \lambda^{5} + 540 \lambda^{4} - 432 \lambda^{3} + 81 \lambda^{2}}}{2}}.
\end{align*}
By a Taylor expansion, obtained by symbolic calculation, we calculate
\begin{align*}
  \frac{C}{\lambda}
  &= \sqrt{3}
  + \left( \frac{\sqrt{3}}{2} + 1\right) \frac{1}{\lambda}
  - \left( 3 + \frac{\sqrt{3}}{4} \right) \frac{1}{2\lambda^2}
  % + \left( 9 - \frac{21 \sqrt{3}}{8} \right) \frac{1}{6 \lambda^3}
  + \bigo{\frac{1}{\lambda^3}} \\
  &=: \hat C(\lambda) + \bigo{\frac{1}{\lambda^3}}.
\end{align*}
The roots of $q$ (and therefore also of $p$) are given by
\begin{align*}
  x_k &= - \frac{1}{3a}\left(b+\xi^kC+\frac{\Delta_0}{\xi^kC}\right), \\
      &= - \frac{1}{3}\left(\frac{1}{\lambda}+\xi^k \hat C+\frac{\Delta_0}{\xi^k \hat C \lambda^2} + \bigo{\frac{1}{\lambda^3}}\right),
    \qquad \xi = \frac{-1 + i \sqrt{3}}{2},
    \qquad k \in \{0,1,2\}.
\end{align*}
Carrying out the calculation of the leading order terms using symbolic calculations,
we obtain
\begin{align*}
  \mathrm {Re} (x_0) = - \frac{1}{\lambda} + \bigo{\frac{1}{\lambda^2}}, \quad
  \mathrm {Re} (x_1) = - \frac{1}{2\lambda^2} + \bigo{\frac{1}{\lambda^3}}, \quad
  \mathrm {Re} (x_2) = - \frac{1}{2\lambda^2} + \bigo{\frac{1}{\lambda^3}},
\end{align*}
so the decay rate scales as $-1/2 \nu^4$ in the limit,
which coincides with the result found in \cref{par:limit_nu_to_infty_with_gamma_fixed_}.
This example shows that, while feasible,
obtaining the scaling of the spectral gap from the explicit expressions for the roots of the characteristic polynomial
presents a level of difficulty similar to the one encountered above.
One advantage of our approach is that it can easily be extended to higher dimensions,
which is relevant when more auxiliary processes are considered in the GLE.
}
\end{remark}

\section{Longtime behavior for model GL2}%
\label{sec:longtime_behavior_for_model_gl2}

We present here elements on how to obtain results for GL2 similar to the ones proved for GL1 in Section~\ref{sec:convergence_of_the_gle_dynamics},
with choices of coefficients providing a modified $H^1(\mu)$ inner product and guaranteeing hypocoercivity and hypoelliptic regularization. We omit the details of the calculations, which are more cumbersome than for GL1.
Defining the generator associated with GL2 as
\begin{align*}
    \textstyle
    -\mathcal L
    & = \textstyle \nu^{-2} \, \alpha^2 \, \beta^{-1} \, \partial_{z_2}^* \, \derivative{1}{ {z_2} } \\
    & \ \ \textstyle + \nu^{-2} \, \alpha ( z_1 \, \derivative{1}{ {z_2} } - z_2 \, \derivative{1}{ {z_1} } )
    + \sqrt{\gamma} \, \nu^{-1} ( p \, \derivative{1}{ {z_1} } - z_1 \, \derivative{1}{p} )
    + ( \derivative*{1}[V]{q} \, \derivative{1}{p} - p \, \derivative{1}{q} ),
\end{align*}
where $\partial^*_{z_2} := \beta z_2 - \partial_{z_2}$, the approach outlined below leads to the resolvent bound
\begin{align}
  \label{eq:resolvent_bound_GL2}
  \norm{\mathcal L^{-1}}[\mathcal B\left(L^2_0(\mu)\right)] \leq C \max  \left( \gamma, \  \frac{1}{\gamma}, \  \frac{\nu^{4}}{\gamma}, \  \frac{\nu^{8}}{\alpha^{4} \gamma}, \  \frac{\gamma}{\alpha^{2}}, \  \frac{\gamma^{2} \nu^{2}}{\alpha^{4}}\right).
\end{align}
In contrast with our observations in the case of model GL1,
the scaling on the right-hand side of this equation appears not to be sharp.
Indeed, an explicit calculation of the scaling of the spectral bound in the case of a quadratic potential,
which is possible by using the same approach as in~\cref{sec:confirmation_of_the_rate_of_convergence_in_the_quadratic_case},
shows that the resolvent bound scales in fact as $\gamma$ in the limit as $\gamma \to \infty$ for fixed $\alpha$ and $\nu$,
and not as $\gamma^2$ as suggested by~\eqref{eq:resolvent_bound_GL2}.

Let us now make precise how we obtain~\eqref{eq:resolvent_bound_GL2}. Define the coefficients
\begin{align*}
    a_0 &= 2  A^{4} \min \left( 1, \alpha, \frac{\alpha^{2}}{\nu^{2}}, \frac{\alpha^{2}}{\sqrt{\gamma} \nu}\right), \\
    a_1 &= 2 A^{10} \min \left( 1, \alpha, \frac{1}{\nu^{2}}, \frac{\alpha^{4}}{\nu^{6}}, \frac{1}{\sqrt{\gamma} \nu}, \frac{\alpha^ {4}}{\gamma^{\frac{3}{2}} \nu^{3}}\right), \\
    a_2 &= 2 A^{14} \min \left( 1, \alpha, \gamma, \frac{\gamma}{\nu^{4}}, \frac{\alpha^{4} \gamma}{\nu^{8}}, \frac{1}{\sqrt{\gamma} \nu}, \frac{\alpha^{4}}{\gamma^{\frac{3}{2}} \nu^{3}}\right), \\
    a_3 &= 2 A^{16} \min \left( \gamma, \frac{1}{\gamma}, \frac{\gamma}{\nu^{4}}, \frac{\alpha^{2}}{\gamma}, \frac{\alpha^{4} \gamma}{\nu^{8}},\frac{\alpha^{4}}{\gamma^{2} \nu^{2}}\right),
\end{align*}
and
\begin{align*}
    b_0 &= A^{7} \min \left( \alpha, \frac{1}{\alpha}, \frac{\alpha^{3}}{\nu^{4}}, \frac{\alpha^{3}}{\gamma \nu^{2}}\right), \\
    b_1 &= A^{12} \min \left( \sqrt{\gamma} \nu,\frac{1}{\sqrt{\gamma} \nu}, \frac{\sqrt{\gamma}}{\nu^{3}}, \frac{\alpha^{4} \sqrt{\gamma}}{\nu^{7}}, \frac{\alpha^{4}}{\gamma^{\frac{3}{2}} \nu^{3}}\right), \\
    b_2 &= A^{15} \min \left( \gamma, \frac{1}{\gamma}, \frac{\gamma}{\nu^{4}}, \frac{\alpha^{2}}{\gamma}, \frac{\alpha^{4} \gamma}{\nu^{8}},\frac{\alpha^{4}}{\gamma^{2} \nu^{2}}\right).
\end{align*}
It is possible to show that, for any sufficiently small $A$ and for fixed $\beta$,
the following conditions are satisfied for any values of the parameters $\gamma$, $\nu$ and $\alpha$:
\begin{itemize}
    \item
        The following matrix is positive definite:
        \[
            \begin{pmatrix}
                a_0 & -b_0 & 0 & 0 \\
                -b_0 & a_1 & -b_1 & 0 \\
                0 & -b_1 & a_2 & -b_2 \\
                0 & 0 & -b_2 & a_3
            \end{pmatrix}
        \]
        Consequently, the inner product defined by polarization from the norm
        \begin{align*}
          % \label{eq:hypocoercivity:norm}
          \textstyle
          \iip{h}{h}
          & := \textstyle \norm{h}^2 + a_0 \norm{\derivative{1}[h]{ {z_2} }}^2 + a_1 \norm{\derivative{1}[h]{ {z_2} }}^2 + a_2 \norm{\derivative{1}[h]{p}}^2 + a_3 \norm{\derivative{1}[h]{q}}^2 \\
           &  \ \ \textstyle
           - 2 b_0 \ip{\derivative{1}[h]{ {z_2} }}{\derivative{1}[h]{ {z_1} }}
           - 2 b_1 \ip{\derivative{1}[h]{ {z_1} }}{\derivative{1}[h]{p}}
           - 2 b_2 \ip{\derivative{1}[h]{p}}{\derivative{1}[h]{q}},
        \end{align*}
        is equivalent to the usual $\sobolev{1}{\mu}$ inner product,
        by an inequality similar to~\eqref{eq:hypocoercivity:equivalence_with_sobolev_norm}.

    \item The operator $- \mathcal L$ is coercive in $H^1_0(\mu)$ endowed with the inner product $\iip{\dummy}{\dummy}$.
        More precisely, it holds for any $h \in H^1_0(\mu)$ that
        \[
            - \iip{h}{\mathcal Lh} \geq C(A) \lambda(\mu, \nu, \alpha) \iip{h}{h},
        \]
        where
        \begin{equation}
            \label{eq:rate_gl2}
            \lambda(\mu, \nu, \alpha) = \min
                \left( \gamma, \  \frac{1}{\gamma}, \  \frac{\gamma}{\nu^{4}}, \  \frac{\alpha^{4} \gamma}{\nu^{8}}, \  \frac{\alpha^{2}}{\gamma}, \  \frac{\alpha^{4}}{\gamma^{2} \nu^{2}}\right).
        \end{equation}
    \item
        For any $h \in L^2_0(\mu)$,
        it holds $\derivative*{1}[N_h]{t}(t) \leq 0$ for $t \in [0, 1]$,
        where $N_h(t)$ is defined by
        \begin{align*}
          \textstyle
          N_{h}(t)
          & =
          \textstyle \norm{\e^{t \mathcal L}h}^2
          + a_0 t \norm{\derivative{1}[\e^{t \mathcal L}h]{ {z_2} }}^2
          + a_1 t^3 \norm{\derivative{1}[\e^{t \mathcal L}h]{ {z_1} }}^2
          + a_2 t^5 \norm{\derivative{1}[\e^{t \mathcal L}h]{p}}^2
          + a_3 t^7 \norm{\derivative{1}[\e^{t \mathcal L}h]{q}}^2
          \nonumber \\
           & \ \ \textstyle
           - 2 b_0 t^2 \ip{\derivative{1}[\e^{t \mathcal L}h]{ {z_2} }}{\derivative{1}[\e^{t \mathcal L}h]{ {z_1} }}
           - 2 b_1 t^4 \ip{\derivative{1}[\e^{t \mathcal L}h]{ {z_1} }}{\derivative{1}[\e^{t \mathcal L}h]{p}}
           - 2 b_2 t^6 \ip{\derivative{1}[\e^{t \mathcal L}h]{p}}{\derivative{1}[\e^{t \mathcal L}h]{q}}.
        \end{align*}
      \end{itemize}
      This leads finally to~\eqref{eq:resolvent_bound_GL2}.

%------------------------------------------------------
      \section{Technical results used in \texorpdfstring{\cref{sub:gle:the_underdamped_limit}}{the underdamped limit}}
\label{sec:estimates_underdamped_limit}

In this section,
we present technical results that are used in \cref{sub:gle:the_underdamped_limit}.
The following estimate is useful in motivating~\eqref{eq:phi_und}.

\begin{lemma}
  \label{lemma:auxiliary_underdamped}
  Assume that $c > 0$ is a constant,
  that $h$ is a smooth function on the torus $\torus$
  and that $g$ is a smooth, everywhere positive function on $\torus$.
  Then there exists a unique smooth solution~$f$ to the equation
  \begin{equation}
    \label{eq:underdamped_ode}
    g(q) \derivative*{1}[f]{q}(q) - c \, f(q) = h(q), \qquad q \in \torus.
  \end{equation}
  In addition,
  \begin{equation}
    \label{eq:underdamped_estimate_ode}
    \max_{q \in \torus} \abs{f(q) + \frac{h(q)}{c} + \frac{g(q) \derivative*{1}[h]{q}(q)}{c^2}}
    \leq \frac{1}{c^3} \max_{q \in \torus} \abs{g(q) \, (g \, \derivative*{1}[h]{q})'(q)}.
  \end{equation}
\end{lemma}

\begin{proof}
Let $f_a(q)$ denote the solution to
  \[
    \forall q \in [-\pi, \pi], \quad g(q) \derivative*{1}[f_a]{q}(q) - c \, f_a(q) = h(q), \qquad f_a(- \pi) = a.
  \]
  Since $g$ is smooth and everywhere positive,
  this equation admits a unique smooth solution on the interval $[-\pi, \pi]$,
  by the standard theory of ordinary differential equations.
  Rewriting the equation for $f_a$ as
  \[
    g(q) \exp \left( \int_{-\pi}^{q} \frac{c}{g(x)} \, \d x \right) \, \derivative*{1}{q} \left[f_a(q) \,  \exp \left( - \int_{-\pi}^{q} \frac{c}{g(x)} \, \d x \right) \right] = h(q),
  \]
  it is clear that
  \begin{equation}
    f_a(q) = \exp \left( \int_{-\pi}^{q} \frac{c}{g(x)} \, \d x \right) \left[a + \int_{-\pi}^{q} \, \exp \left( - \int_{-\pi}^{y} \frac{c}{g(x)} \, \d x \right) \, \frac{h(y)}{g(y)} \, \d y \right].
  \end{equation}
  The unique value~$a^*$ for which the function is periodic (namely $f_{a^*}(\pi) = a^*$) is
  \[
    a^* = -\left[1- \exp\left(-\int_{-\pi}^{\pi} \frac{c}{g(x)} \, \d x\right)\right]^{-1}\int_{-\pi}^{\pi} \, \exp \left( - \int_{-\pi}^{y} \frac{c}{g(x)} \, \d x \right) \, \frac{h(y)}{g(y)} \, \d y.
  \]
  It is easily checked that all \revision{derivatives} of~$f_{a^*}$ are also continuous,
  so that $f_{a^*} \in C^{\infty}(\torus)$.

  To obtain the estimate~\eqref{eq:underdamped_estimate_ode},
  we introduce $r(q) := f(q) + c^{-1} h(q) + c^{-2} g(q) \, \derivative*{1}[h]{q}(q)$ and note by that, by~\eqref{eq:underdamped_ode},
  \[
    g(q) \, \derivative*{1}[r]{q}(q) - c \, r(q) = \frac{1}{c^2} \, g(q) \, (g \, \derivative*{1}[h]{q})'(q).
  \]
  The term $\derivative*{1}[r]{q}(q)$ vanishes at the extrema of $r(q)$,
  which leads then to~\eqref{eq:underdamped_estimate_ode}.
\end{proof}

\revision{We next} prove another technical result,
which we then employ to show that the leading order term $\overline{\phi}_0$ in the underdamped limit,
defined in~\eqref{eq:expression_psi0}, indeed belongs to~$H^1(\mu)$ when~$V$ is not constant.

\begin{lemma}
  \label{lemma:auxiliary_behavior_Snu}
  If $q \mapsto V(q)$ is not constant,
  then there exists $M > 0$ such that
  \begin{equation}
    \label{eq:bounds_sn}
    \forall E \in (E_0, \infty), \qquad S_{\nu}(E) \geq \frac{\nu^2}{1 + \nu^2 M} \int_\torus \sqrt{2(E_0-V(q))} \, \d q.
  \end{equation}
  In particular, $S_{\nu}(E)$ is bounded below by a positive constant \revision{uniformly in~$E\in(E_0,\infty)$}.
\end{lemma}

\begin{proof}
  Fix $E > E_0$. We integrate~\eqref{eq:ode_for_s_nu} over~$\torus$ and use integration by parts for the first term,
  which gives
  \begin{align}
    \label{eq:lemma_intermediate}
    \int_{\torus} \left( \frac{1}{\nu^2} + \derivative{1}[P]{q}(q, E) \right) s_{\nu}(q, E) \, \d q = \int_{\torus} P(q, E) \, \d q.
  \end{align}
  We now show that that there exists a positive constant~$M$ such that
  \begin{equation}
    \label{eq:bound_derivative}%
    \forall (q, E) \in \torus \times (E_0, \infty), \qquad \abs{\derivative{1}[P]{q}(q, E)} = \frac{\abs{\derivative*{1}[V]{q}(q)}}{\sqrt{2(E - V(q))}} \leq M.
  \end{equation}
  To this end,
  notice that $\abs{\derivative{1}[P]{q}(q, E)}$ is a decreasing function of $E$ for fixed $q$,
  so it is sufficient to show that
  \[
    \sup_{q \in \torus} \left( \lim_{E \to E_0^+} \abs{\derivative{1}[P]{q}(q, E)} \right) < \infty.
  \]
  Since $V$ is smooth,
  it holds that $\derivative*{1}[V]{q}(q) = 0$ for any $q$ such that $V(q) = E_0$, so
  \[
    L(q) := \lim_{E \to E_0^+} \abs{\derivative{1}[P]{q}(q, E)} =
   \begin{cases}
     0 & \text{ if $V(q) = E_0$, } \\
     \dps \frac{\abs{\derivative*{1}[V]{q}(q)}}{\sqrt{2(E_0 - V(q))}} & \text{ otherwise. }
   \end{cases}
  \]
  By using L'H\^opital's rule, we notice that, for values of~$q$ in a neighborhood of any $q^* \in V^{-1}\{ E_0 \}$ with $V'(q) \neq 0$,
  \begin{align*}
    \lim_{q \to q^*} \frac{\abs{\derivative*{1}[V]{q}(q)}}{\sqrt{2(E_0 - V(q))}}
    = \sqrt{\lim_{q \to q^*} \frac{\abs{\derivative*{1}[V]{q}(q)}^2}{2(E_0 - V(q))}}
    = \sqrt{ -\lim_{q \to q^*} \frac{\derivative*{1}[V]{q}(q) \, \derivative*{2}[V]{q^2}(q)}{\derivative*{1}[V]{q}(q)}}
    = \sqrt{ -\lim_{q \to q^*} \derivative*{2}[V]{q^2}(q)},
  \end{align*}
  so $L(q)$ is bounded uniformly from above,
  implying that~\eqref{eq:bound_derivative} holds.
  This concludes the proof because, by~\eqref{eq:lemma_intermediate}, and recalling that $s_\nu \geq 0$ by~\eqref{eq:proof_bounds_underdamped},
  \begin{align*}
    \int_{\torus} s_{\nu}(q, E) \, \d q
    &\geq \int_{\torus} \frac{\nu^{-2} + \abs{\derivative{1}[P]{q}(q, E)}}{ \nu^{-2} + M} s_{\nu}(q, E) \, \d q \\
    &\geq \frac{\nu^2}{1 + \nu^2 M} \, \int_{\torus} \left( \frac{1}{\nu^2} + \derivative{1}[P]{q}(q, E) \right) s_{\nu}(q, E) \, \d q
      = \frac{\nu^2}{1 + \nu^2 M} \, \int_{\torus} P(q, E) \, \d q \\
    & \geq \frac{\nu^2}{1 + \nu^2 M} \, \int_{\torus} P(q, E_0) \, \d q,
  \end{align*}
  \revision{and} the last integral is positive because~$V$ is not constant.
\end{proof}

%Combining this with the fact that $S_{\nu}(E)$ is a continuous function of $E$,

\begin{lemma}
  If $q \mapsto V(q)$ is not constant, then the function~$\overline{\phi}_0$ defined in~\eqref{eq:phi_und} belongs to~$H^1(\mu)$.
\end{lemma}
\begin{proof}
  It is easy to see that the function~$\overline{\phi}_0$ defined in~\eqref{eq:phi_und} belongs to~$L^2(\mu)$.
  We next consider the distributional derivatives $\derivative{1}[\overline \phi_0]{y}$ for $y \in \{q,p\}$, since derivatives in~$z$ vanish: for $H(q,p) < E_0$, these derivatives vanish;
  while for $H(q,p) > E_0$, it holds \revision{$\derivative{1}[\overline \phi_0]{q} = \sign(p) \psi_0'(H(q,p)) V'(q)$} and \revision{$\derivative{1}[\overline \phi_0]{p} =  \psi_0'(H(q,p)) |p|$}.
  Note that all these derivatives belong to $L^{\infty}_{\rm loc}(\torus \times \real \times \real)$ in view of the lower bound provided by~\eqref{eq:bounds_sn}.
  Moreover, $\abs{\nabla \overline \phi_0(q, p, z)}$ grows sufficiently slowly as $|p| \to \infty$.
  This allows to conclude that $\overline{\phi}_0 \in H^1(\mu)$, as claimed.
\end{proof}

%---------------------------------------------------------
\bibliographystyle{plain}
\bibliography{paper}
\end{document}